\documentclass[letterpaper,USenglish,cleveref, pdfa]{lipics-v2021}

\pdfoutput=1 
\hideLIPIcs  

\title{The Primal Pathwidth SETH} 

\author{Michael Lampis}{Universit\'{e} Paris-Dauphine, PSL University, CNRS UMR7243, LAMSADE, Paris, France}{michail.lampis@dauphine.fr}{https://orcid.org/0000-0002-5791-0887}{}

\authorrunning{M. Lampis} 

\Copyright{Michael Lampis}



\category{} 

\relatedversion{} 

\funding{This work was supported by ANR project ANR-21-CE48-0022 (S-EX-AP-PE-AL).}

\nolinenumbers 

\EventEditors{John Q. Open and Joan R. Access}
\EventNoEds{2}
\EventLongTitle{42nd Conference on Very Important Topics (CVIT 2016)}
\EventShortTitle{CVIT 2016}
\EventAcronym{CVIT}
\EventYear{2016}
\EventDate{December 24--27, 2016}
\EventLocation{Little Whinging, United Kingdom}
\EventLogo{}
\SeriesVolume{42}
\ArticleNo{23}

\usepackage{tikz}

\newcommand{\tsat}{3-\textsc{SAT}} 
\newcommand{\kova}{$k$-\textsc{OVA}}
\newcommand{\scc}{\textsc{SCC}}
\newcommand{\seth}{\textsc{SETH}}

\newcommand{\ppseth}{$\pw$\textsc{-SETH}}
\newcommand{\pw}{\textrm{pw}}
\newcommand{\tw}{\textrm{tw}}
\newcommand{\lcw}{\textrm{lcw}}

\newcommand{\eps}{\varepsilon}

\begin{document}

\maketitle

\begin{abstract}

Motivated by the importance of dynamic programming (DP) in parameterized
complexity, we consider several fundamental fine-grained questions, such as the
following representative examples: (i) can \textsc{Dominating Set} be solved in
time $(3-\eps)^{\pw}n^{O(1)}$?  (where $\pw$ is the pathwidth of the input
graph) (ii) can \textsc{Coloring} be solved in time
$\pw^{(1-\eps)\pw}n^{O(1)}$?  (iii) can a short reconfiguration between two
size-$k$ independent sets be found in time $n^{(1-\eps)k}$? Such questions are
well-studied: in some cases the answer is No under the SETH, while in others
coarse-grained lower bounds are known under the ETH.  Even though questions
such as the above seem ``morally equivalent'' as they all ask if a simple DP
can be improved, the problems concerned have wildly varying time complexities,
ranging from single-exponential FPT to XNLP-complete. 

This paper's main contribution is to show that, despite their varying
complexities, these questions are not just morally equivalent, but in fact they
are the same question in disguise.  We achieve this by putting forth a natural
complexity assumption which we call the Primal Pathwidth-Strong Exponential
Time Hypothesis (\ppseth) and which states that \tsat\ cannot be solved in time
$(2-\eps)^{\pw}n^{O(1)}$, for any $\eps>0$, where $\pw$ is the pathwidth of the
primal graph of the input CNF formula. We then show that numerous fine-grained
questions in parameterized complexity, including the ones above, are
\emph{equivalent} to the \ppseth, and hence to each other. This allows us to
obtain sharp fine-grained lower bounds for problems for which previous lower
bounds left a constant in the exponent undetermined, but also to increase our
confidence in bounds which were previously known under the SETH, because we
show that breaking any one such bound requires breaking \emph{all} (old and
new) bounds; and because we show that the \ppseth\ is more plausible than the
SETH.  More broadly, our results indicate that \ppseth-equivalence is a
meta-complexity property that cuts across traditional complexity classes,
because it exactly captures the barrier impeding progress not only for problems
solvable in time $c^kn^{O(1)}$ (such as \tsat\ itself for parameter $\pw$), but
also for problems which are considerably harder.

In more detail, we show that all of the following are equivalent to falsifying
the \ppseth:

\begin{enumerate}

\item Single-exponential FPT problems: solving $k$-\textsc{Coloring} in time
$(k-\eps)^{\pw}n^{O(1)}$ or $(2^k-2)^{\lcw}n^{O(1)}$, for any $k\ge 3$, where
$\lcw$ is the linear clique-width; solving
\textsc{Distance}-$d$-\textsc{Independent Set} in time
$(d-\eps)^{\pw}n^{O(1)}$, for any $d\ge 2$; solving
\textsc{Distance}-$r$-\textsc{Dominating Set} in time
$(2r+1-\eps)^{\pw}n^{O(1)}$ for any $r\ge 1$; solving \textsc{Set Cover} in
time $(2-\eps)^{\pw}n^{O(1)}$.

\item Super-exponential FPT problems: solving \textsc{Coloring} in time
$\pw^{(1-\eps)\pw}n^{O(1)}$; solving $C_4$-\textsc{Hitting Set} in time
$(\sqrt{2}-\eps)^{\pw^2}n^{O(1)}$. For the latter problem we give an algorithm
tightly matching this complexity.

\item XNLP-complete problems: solving \textsc{List Coloring} in time
$n^{(1-\eps)\pw}$; finding a short word accepted by $k$ $n$-state DFAs in time
$n^{(1-\eps)k}$; finding a short reconfiguration sequence between two size-$k$
independent sets of an $n$-vertex graph in time $n^{(1-\eps)k}$.

\end{enumerate}

We argue that the \ppseth\ is a plausible assumption by showing that it is
implied not only by the SETH, but also by a maximization version of the
$k$-Orthogonal Vectors Assumption (itself a consequence of the SETH) and by the
Set Cover Conjecture, whose relation to the SETH is a major open problem.
Hence, lower bounds which were previously known under the SETH, such as those
of point 1 above, are now shown to also be implied by these conjectures.

\end{abstract}

\newpage

\tableofcontents

\newpage

\section{Introduction}

\subsection{Background and Motivation}

Our goal in this paper is to provide a better answer to the following question:
``Are current parameterized algorithms based on dynamic programming over linear
structures optimal?''. Before we explain how we attack this question, let us
briefly review why this question is important and summarize some known results.

Dynamic Programming (DP) is a standard algorithm design technique and one of
the most important tools in the parameterized algorithms toolbox. Its central
role in this setting stems mainly from the fact that when we are dealing with
an NP-hard graph problem and the parameter is a graph width, such as treewidth,
pathwidth, or clique-width, we can very often obtain an FPT algorithm with
running time $f(k)n^{O(1)}$ by a routine application of DP (see the
corresponding chapter of \cite{CyganFKLMPPS15}). Indeed, for many such problems
we even obtain $c^kn^{O(1)}$ time algorithms, where $c$ is a constant.  Because
of the wide applicability (and relative simplicity) of dynamic programming, it
is a question of central interest to parameterized complexity whether the
algorithms we obtain in this way are best possible. Indeed, this question has
been extensively studied, so in order to summarize what is currently known, let
us focus on a specific concrete problem, \textsc{Coloring} (we give a more
extensive review of known results further below). 

On the algorithmic side, \textsc{Coloring} is often used as a first example of
the DP technique for tree decompositions. The standard algorithm needs to
remember the color of each vertex of a bag. This means that, for fixed $k\ge
3$, $k$-\textsc{Coloring} can be solved in time $k^{\tw}n^{O(1)}$; when $k$ is
part of the input \textsc{Coloring} can be solved in time $\tw^{\tw}n^{O(1)}$
(because $k\le \tw+1$ without loss of generality); and when each vertex has an
arbitrary list of available colors, a variant known as \textsc{List Coloring},
the algorithm runs in time $n^{\tw+O(1)}$.  Despite the different running
times, these results are just three manifestations of the same algorithm.

Is the above algorithm optimal? It is widely believed that the answer is yes in
all three cases, even if we replace treewidth by the more restrictive parameter
pathwidth, but this is known with varying degrees of precision. When $k\ge 3$
is a fixed constant, Lokshtanov, Marx, and Saurabh \cite{LokshtanovMS18}
established that no $(k-\eps)^{\pw}n^{O(1)}$ algorithm exists, assuming the
Strong Exponential Time Hypothesis (SETH) -- note that their lower bound
applies to pathwidth, and hence also to treewidth, which is a more general
parameter.  For the other two cases, it is known that under the ETH, we cannot
obtain algorithms with running times $\pw^{o(\pw)}n^{O(1)}$ and $n^{o(\pw)}$
respectively \cite{LokshtanovMS18b,FellowsFLRSST11}. The problem is thus only
approximately resolved in the latter two cases (because the lower bounds leave
an unspecified constant in the exponent) but very well-understood for constant
$k$, albeit by relying on a strong complexity assumption.

The concrete example of \textsc{Coloring} exemplifies the state of the art in
this area, where much work has been expended in investigating whether
``simple'' (and not so simple) DP algorithms are optimal. By now much has been
clarified, but at least three specific weaknesses in the state of the art can
be identified, which we summarize in three points as follows:

\begin{enumerate}

\item There is a lack of fine-grained lower bounds for harder problems. 

\item The SETH is too strong.

\item All reductions are one-way.

\end{enumerate}

Let us explain what we mean by these points. Regarding the first point, the
example of \textsc{Coloring} is typical of the state of the art in the
following sense: for problems whose complexity is $c^wn^{O(1)}$, where $w$ is a
graph width and $c$ a fixed constant, the pioneering work of
\cite{LokshtanovMS18} has spawned a long line of research which has managed to
obtain lower bounds under the SETH which exactly match the constant $c$ in the
best algorithm.  Unfortunately, a similar fine-grained undertaking is currently
lacking in the literature both for FPT problems where the parameter dependence
is more than single-exponential and for XP problems. Instead, what is typically
known for such problems, as we saw for the example of \textsc{Coloring}, are
ETH-based lower bounds which at best match the complexity of the best algorithm
only up to an unspecified constant in the exponent. This is clearly not as
satisfactory as a fine-grained tight lower bound.

Regarding the second point, it is worth mentioning that the SETH is a
hypothesis that is quite strong and currently not universally believed to hold,
so it is certainly part of our motivation to base our lower bounds on a more
plausible hypothesis. However, another way to see that the SETH is too strong
is that it implies lower bounds for parameters so restricted that the matching
algorithm is not DP, but a much simpler branching algorithm (this was recently
investigated thoroughly by Esmer, Focke, Marx, and Rzazewski
\cite{abs-2402-07331}). This means that if the question we care about is
whether DP algorithms are optimal, the SETH does not provide the ``right''
answer (or, as \cite{abs-2402-07331} put it, the ``real source of hardness''). We
would therefore prefer to replace the SETH with a better hypothesis.

Finally, the third point is a larger weakness of modern (fine-grained)
complexity theory, where unlike classical complexity theory, we rarely manage
to classify problems in complexity classes, but rather usually give one-way
reductions from a handful of accepted hypotheses. In our setting this is
particularly disappointing, because simple DP algorithms on tree and path
decompositions can be seen as facets of one global algorithm, as exemplified by
the example of \textsc{Coloring} given above.  One would naturally expect that
improving upon this global algorithm for one problem should mean something for
the others.  However, in the current state of the art, obtaining, say a
$(3-\eps)^{\pw}n^{O(1)}$ algorithm for $3$-\textsc{Coloring} would falsify the
SETH, but is not known to imply anything about any other problem -- in
particular, even for other cases of \textsc{Coloring}, such an algorithm is not
known to imply a faster algorithm when $k$ is part of the input, or even when
$k=4$.

To summarize, the claim that current DP algorithms are optimal could become
more convincing if (i) we were able to extend the strong precise lower bounds
which are currently known for the narrow case of problems solvable in
single-exponential FPT time to wider domains where DP is applied (ii) we were
able to replace the SETH with a more plausible and more appropriate source of
hardness and (iii) if we were able to show that the questions we care about
form an equivalence class, as this would tie all lower bounds together and
lessen the importance of the complexity assumption on which hardness is based.

\subsection{Summary of results}

In this paper we focus on parameterized DP algorithms that work on linear
structures, notably (but not exclusively) path decompositions. We revisit the
question of whether such algorithms are optimal. Our goal is to confirm that
the answer is yes and provide fine-grained lower bounds tightly matching the
performance of the best algorithms while addressing the three weak points we
mentioned above. For this reason we put forth the following complexity
hypothesis:

\begin{conjecture}[\ppseth] 

For all $\eps>0$ we have the following: there exists no algorithm which  takes
as input a \tsat\ instance $\phi$ on $n$ variables and a path decomposition of
its primal graph of width $\pw$ and correctly decides if $\phi$ is satisfiable
in time $(2-\eps)^{\pw}n^{O(1)}$.

\end{conjecture}

To give a brief idea of what we hope to achieve using the \ppseth, let us take
another look at our running example of \textsc{Coloring}. Among the results we
present, we will prove that the following are \emph{equivalent} to falsifying
the \ppseth: (i) obtaining, for some $k\ge 3$, a $(k-\eps)^{\pw}n^{O(1)}$
algorithm for $k$-\textsc{Coloring} (ii) obtaining a
$\pw^{(1-\eps)\pw}n^{O(1)}$ algorithm for \textsc{Coloring} (iii) obtaining a
$n^{(1-\eps)\pw}$ algorithm for \textsc{List Coloring}. Hence the \ppseth\ will
allow us to obtain new sharp bounds for cases where only ETH-based bounds were
known; and to strengthen an existing SETH-based lower bound, by basing it on a
more plausible hypothesis and showing that it is equivalent to the other lower
bounds.

Our results are a clear improvement to the state of the art because (i) even
though our hypothesis involves a single-exponential FPT problem, we are able to
obtain fine-grained bounds for FPT problems with super-exponential parameter
dependence and even for non-FPT (XNLP-complete) problems (ii) the \ppseth\ is
exactly as strong as it needs to be to obtain our bounds, and we show that it
is probably much more plausible than the SETH (iii) all (new and old) lower
bounds form an equivalence class, which makes them more convincing because
disproving one is as hard as disproving all of them.

Importantly, a message of this work is that \ppseth-equivalence cuts across
traditional complexity classes, in the sense that (as we show) improving upon
the best algorithms for some XNLP-complete problem, such as \textsc{List
Coloring}, is both necessary and sufficient to improve upon the best algorithms
for some FPT problems, such as \textsc{Coloring} or $k$-\textsc{Coloring} (for
parameter pathwidth).  Thus, the \ppseth\ appears to be a natural complexity
hypothesis that deserves significant further study because it exactly captures
the fundamental question ``Can simple DP algorithms over linear structures be
improved even by a small amount?'', above and beyond what is captured by
traditional complexity classes.  

We choose to focus on linear structures (e.g.  pathwidth rather than treewidth)
for several reasons: First, linear structures are where ``simple'' DP
algorithms are often expected to be optimal; for tree-like structures, even for
basic problems (e.g.  \textsc{Dominating Set} parameterized by treewidth), much
more sophisticated techniques such as Fast Subset Convolution are needed to
obtain the best algorithm (see the corresponding chapter of
\cite{CyganFKLMPPS15}). Therefore, we expect the study of these questions for
linear parameters to be more fruitful.  Second, DP over linear structures
arises naturally in contexts unrelated to graph width parameters -- indeed in
\cref{sec:xnlp} we identify two natural \ppseth-equivalent questions where the
parameter has nothing to do with graph decompositions -- and we find it
interesting that the \ppseth\ is flexible enough to capture such cases.  Third,
since we care chiefly about proving lower bounds, sharpening or strengthening
lower bounds for linear parameters, such as pathwidth, as we do in this paper,
has the immediate benefit that the same improvements carry over to the more
general case of tree-like parameters, such as treewidth and clique-width.

\medskip

Let us now outline the results presented in this paper in more detail.

\subparagraph*{Robustness.} Since the \ppseth\ will be the main tool we use to
unify our lower bounds, we begin our investigation by looking more closely at
our formulation of the hypothesis.  For this, it is important to recall that
one particularly annoying aspect of the SETH is that the hypothesis appears to
be very sensitive to small perturbations of its definition. For instance, the
standard formulation is that for all $\eps>0$ there exists an arity $q$ such
that $q$-\textsc{SAT} cannot be solved in time $(2-\eps)^n$
\cite{ImpagliazzoP01,ImpagliazzoPZ01}; making the same assumption for
\textsc{SAT} for unbounded arity is, as far as is known, a  more believable but
possibly not equivalent assumption; while changing the order of quantification
(stating ``there exists $q$ such that for all $\eps>0$\ldots'') makes the
hypothesis false.  Similarly, it is not known whether formulating the same
hypothesis for \textsc{Max-SAT} (instead of \textsc{SAT}) renders it strictly
more plausible or whether the statement remains equivalent. Even worse, the
hypothesis that \textsc{Max-3-SAT} is impossible to solve in time $(2-\eps)^n$
is itself also consistent with the state of the art and sometimes used as an
alternative starting point for reductions (e.g. in \cite{abs-2402-07331}),
without having a clear relation to the SETH.

A first contribution of this paper is to show that, thankfully, the \ppseth\
does not suffer from the same drawbacks, because even non-trivial perturbations
of its formulation do not actually change the hypothesis. More precisely, we
consider (i) allowing clauses of unbounded arity (\textsc{SAT} vs \tsat) (ii)
bounding the number of occurrences of each variable (iii) optimization versions
that seek to maximize the number of satisfied clauses or the weight of a
satisfying assignment or both (iv) parameterizing by the pathwidth of the
\emph{incidence} graph, instead of the primal graph. No matter which
combination of these variations one selects, we show that the \ppseth\ remains
equivalent to the original version.  As we explain below, this is particularly
important for the parameterization by the incidence pathwidth, because this
will allow us to place domination and covering type problems in the same
complexity class.  Furthermore, we give equivalent formulations of the \ppseth\
pertaining to Constraint Satisfaction Problems (\textsc{CSP}s) with non-binary
alphabets.  These prove invaluable in simplifying the reductions we will need
to prove that various problems are harder than the \ppseth. These results are
presented in \cref{sec:robust}.

\subparagraph*{Supporting Evidence.} The above immediately show that the
\ppseth\ is implied by the SETH, despite the fact that we have formulated it
for \tsat\ (and not \textsc{SAT}).  We strengthen this observation by showing
that the \ppseth\ is also implied by two other well-known complexity
hypotheses: the $k$-Orthogonal Vectors Assumption (\kova) and the Set Cover
Conjecture (\scc).  The \kova\ is a known implication of the SETH and is
generally considered a slightly more believable assumption, while the Set Cover
Conjecture does not currently have any known relation to the SETH and indeed
proving an implication in either direction is a major open problem.  These
results, which we present in \cref{sec:evidence}, indicate that the \ppseth\ is
perhaps quite a bit more believable than the SETH.

\subparagraph*{Single-Exponential FPT Problems.} Our results up to this point
set up the \ppseth\ as a solid foundation for establishing lower bounds. The
obvious first target is then to consider problems with complexity of the same
flavor as \tsat\ parameterized by pathwidth, that is, single-exponential in the
parameter. These are problems for which tight SETH-based lower bounds are
typically known. Our results, presented in \cref{sec:fpt}, are that \emph{each}
of the following would be \emph{equivalent} to falsifying the \ppseth:

\begin{enumerate}

\item Obtaining, for some $k\ge3, \eps>0$, an algorithm for
$k$-\textsc{Coloring} running in $(k-\eps)^{\pw}n^{O(1)}$.

\item Obtaining, for some $k\ge3, \eps>0$, an algorithm for
$k$-\textsc{Coloring} running in time $(2^k-2-\eps)^{\lcw}n^{O(1)}$, where
$\lcw$ is the input graph's linear clique-width.

\item Obtaining, for some $d\ge2, \eps>0$, an algorithm for
$d$-\textsc{Scattered Set} running in time $(d-\eps)^{\pw}n^{O(1)}$.
$d$-\textsc{Scattered Set} is the problem of selecting a maximum cardinality
set of vertices at pairwise distance at least $d$. (For $d=2$, this is just
\textsc{Independent Set}.)

\item Obtaining, for some $r\ge1, \eps>0$, an algorithm for
$r$-\textsc{Dominating Set} running in time $(2r+1-\eps)^{\pw}n^{O(1)}$. This
is the version of \textsc{Dominating Set} where a vertex is dominated if it is
at distance at most $r$ from a selected vertex; for $r=1$ this is just
\textsc{Dominating Set}.

\item Obtaining, for some $\eps>0$, an algorithm for \textsc{Set Cover} running
in $(2-\eps)^{\pw}(n+m)^{O(1)}$. Here $n$ is the number of elements, $m$ the
number of sets, and $\pw$ is the pathwidth of the incidence graph of the given
instance.

\end{enumerate}

All the above lower bounds were previously known as implications of the SETH
(and in the case of \textsc{Set Cover} also the \scc).  Our contribution here
is not merely to reprove these lower bounds under a more believable conjecture,
but also to give bi-directional reductions to and from the \ppseth, proving
that in fact all the above are equivalent. For example, our results show that a
$(6-\eps)^{\lcw}n^{O(1)}$ algorithm for $3$-\textsc{Coloring} is possible if
and only if \textsc{Dominating Set} can be solved in time
$(3-\eps)^{\pw}n^{O(1)}$. Such an implication was not previously known, even
though if one believed the SETH, the two bounds were both known to be optimal.

The list of problems we give here is not meant to be exhaustive. In fact, we
have merely selected some of the most prominent problems to sample what it
means for a lower bound to be equivalent to the \ppseth. Our wider thesis is
the following: \emph{whenever the best algorithm is ``simple DP over a linear
structure'', then improving this algorithm is typically \ppseth-equivalent}. We
stress here that, if we seek lower bounds which are equivalent to and not just
implied by our hypothesis, it is crucial here that we focus on \emph{linear
structures} and not tree-like structures, such as tree decompositions. Indeed,
if we had performed the same investigation for the treewidth equivalent of the
\ppseth, we would have arrived at a much shorter list, as we would \emph{not}
have been able to obtain the above results for \textsc{Dominating Set} and
\textsc{Set Cover} (we discuss this in more detail in \cref{sec:previous}).

\subparagraph*{Super-Exponential FPT Problems.} Even though the results we have
obtained up to this point improve our confidence that several basic DP
algorithms are optimal, one could argue that they are not too surprising:
\tsat\ itself has parameter dependence $2^{\pw}$, so it is natural that it
would be interreducible with other problems with dependence $2^{\pw}$; somewhat
less obviously, one may also expect it to be interreducible with problems with
dependence $c^{\pw}$ for other constants $c$. An interesting aspect of the
\ppseth, however, is that its reach goes further than the class of problems
that admit such algorithms. In particular, we prove that the following are
equivalent to falsifying the \ppseth: 

\begin{enumerate}

\item Obtaining, for some $\eps>0$, an algorithm for \textsc{Coloring} running
in time $\pw^{(1-\eps)\pw}n^{O(1)}$.

\item Obtaining, for some $\eps>0$, an algorithm for $C_4$-\textsc{Hitting Set}
running in time $(\sqrt{2}-\eps)^{\pw^2}n^{O(1)}$.

\end{enumerate}

Here, $C_4$-\textsc{Hitting Set} is the problem of deleting the minimum number
of vertices to destroy all cycles of length $4$ in a graph. Both bounds are
tight, in the sense that we can obtain algorithms with running times $\pw^\pw
n^{O(1)}$ and $2^{\frac{\pw^2}{2}+O(\pw)}n^{O(1)}$ respectively.  However,
prior to this work, only coarse-grained ETH-based lower bounds were known,
implying that no algorithms exist with complexities $\pw^{o(\pw)}n^{O(1)}$ and
$2^{o(\pw^2)}n^{O(1)}$ respectively \cite{LokshtanovMS18b,SauS21}. We stress
here that it is not possible to obtain the new bounds simply by taking the
existing reductions and plugging in a stronger hypothesis (e.g. the SETH in
place of the ETH). Indeed, even if one believes the SETH, it was previously an
open problem whether, for instance, \textsc{Coloring} can be solved in time
$\pw^{0.99\pw}n^{O(1)}$ (we explain this in more detail in
\cref{sec:previous}).  We are able to rule out such improvements by presenting
new reductions, which need to be significantly more efficient than existing
ones and, especially in the case of the latter problem, become as a result much
more technically involved.

\subparagraph*{XNLP-complete Problems.} Encouraged by the previous results, we
take a high-level view of the \ppseth\ as a conjecture that states that simple
DP cannot be improved upon, independent of the time complexity of the
algorithm. Put this way, it becomes tempting to apply this conjecture even to
problems which are not FPT (under standard assumptions).  We therefore move on
to study three problems which are complete for the class XNLP and hence
strongly suspected not to be FPT: \textsc{List Coloring} parameterized by
pathwidth, \textsc{Bounded-}$k$\textsc{-DFA Intersection}, and
\textsc{Short}-$k$-\textsc{Independent Set Reconfiguration}, parameterized by
$k$.  Our results, presented in \cref{sec:xnlp}, show that each of the
following are equivalent to falsifying the \ppseth:

\begin{enumerate}

\item Obtaining, for some $\eps>0$, an algorithm which solves \textsc{List
Coloring} in time $O(n^{(1-\eps)\pw})$, for sufficiently large $\pw$.

\item Obtaining, for some $\eps>0,p>4$, an algorithm which solves \textsc{List
Coloring} on instances of pathwidth $p$ in time $O(n^{p-4-\eps})$.

\item Obtaining, for some $\eps>0$, an algorithm which solves
\textsc{Bounded-}$k$\textsc{-DFA Intersection} in time
$n^{(1-\eps)k}\cdot\ell^{O(1)}$, for sufficiently large $k$, where $\ell$ is the
length of the desired string.

\item Obtaining, for some $\eps>0,k>1$, an algorithm which solves
\textsc{Bounded-}$k$\textsc{-DFA Intersection} in time
$n^{k-\eps}\cdot\ell^{O(1)}$.

\item Obtaining, for some $\eps>0$, an algorithm which solves
\textsc{Short}-$k$-\textsc{Independent Set Reconfiguration} in time
$n^{(1-\eps)k}\cdot\ell^{O(1)}$, for sufficiently large $k$,  where $\ell$ is the
length of the reconfiguration sequence.

\item Obtaining, for some $\eps>0, k>3$, an algorithm which solves
\textsc{Short}-$k$-\textsc{Independent Set Reconfiguration} in time
$n^{k-1-\eps}\cdot\ell^{O(1)}$.

\end{enumerate}

Here, \textsc{Bounded-}$k$\textsc{-DFA Intersection} is the problem where we
are given $k$ DFAs with $n$ states and an integer $\ell$ and are asked whether
there exists a string of length at most $\ell$ accepted by all the DFAs; while
\textsc{Short}-$k$-\textsc{Independent Set Reconfiguration} is the problem
where we are given two size-$k$ independent sets of an $n$-vertex graph and an
integer $\ell$ and are asked if there is a sequence of length at most $\ell$
that transforms one to the other, exchanging at each step a vertex in the set
with a vertex outside, while keeping the set independent at all times.

Let us first explain the (somewhat less fine-grained) forms of the lower bounds
given in points 1,3, and 5. All three problems admit algorithms of the form
$n^{k+O(1)}$ via some kind of DP: for \textsc{List Coloring} this is the DP
algorithm over path decompositions we have already mentioned, while for the
other two problems the algorithm constructs a graph of $n^k$ vertices (the
intersection DFA or the configuration graph) and solves reachability in that
graph. The lower bounds therefore state that if we assume the \ppseth\ the
coefficient of $k$ in the exponent of these algorithms is exactly correct.  In
fact, obtaining an algorithm with complexity of the form $n^{0.99k}$ is both
necessary and sufficient for disproving all lower bound results we have
discussed so far.  This is perhaps striking, because we tightly link the
complexities of problems which are FPT on one hand and XNLP-complete on the
other and, to the best of our knowledge, it is uncommon for such bi-directional
relationships to be known for problems that belong in distinct complexity
classes.

Going further, it is worth highlighting that we are able to obtain even more
fine-grained bounds in this case, as given in points 2, 4, and 6. To see what
we mean, consider the less ambitious question of whether for one of the
problems we can obtain an algorithm running in, say, time $n^{k-\eps}$, that
is, we seek to improve the exponent by an additive, rather than a
multiplicative amount. Surprisingly, the answer can be yes: we show that for
the independent set reconfiguration problem, there is an algorithm which uses
fast matrix multiplication and runs in time $O(n^{k-3+\omega}\cdot\ell)$, where
$\omega$ is the matrix multiplication constant. Hence, when we are looking for
\emph{short} reconfigurations (say, $\ell=n^{o(1)}$), we can save a factor of
at least $n^{0.7}$ over the trivial $n^k$ algorithm -- and indeed if $\omega=2$
we even arrive at $n^{k-1}$. Can the \ppseth\ say something interesting about
modest improvements of this type?

We respond to this question affirmatively by showing that under the \ppseth\ it
is not just impossible to improve the above algorithm for independent set
reconfiguration and the trivial algorithm for DFA intersection by a
multiplicative constant; in fact even an additive constant shaved off the
exponent is \emph{equivalent} to falsifying the \ppseth. As a consequence, our
lower bounds are extremely sharp for these two problems. Furthermore, they
imply that improving the basic algorithm by either a multiplicative or an
additive constant in the exponent are actually the same question, which a
priori does not seem obvious (why should an $n^{k-1.01}\ell^{O(1)}$
algorithm necessarily imply an $n^{0.999k}\ell^{O(1)}$ algorithm?). We are able
to obtain a similar equivalence for \textsc{List Coloring}, albeit without so
precisely determining the minimum  additive constant improvement which triggers
a falsification of the \ppseth\ (we can, however, state that under the \ppseth,
the best algorithm for this problem has complexity somewhere between $n^{p+1}$
and $n^{p-4}$).

\subsection{Overview of techniques}\label{sec:techniques}

Let us now briefly review the main tools we use to establish our results.

\subparagraph*{Robustness.} For the results of \cref{sec:robust}, establishing
that different variants of \textsc{SAT} are equivalent to each other, most
reductions are straightforward and the main challenge is to organize the
(numerous) variations considered. Similarly, to establish that the \ppseth\ is
equivalent to its natural version for \textsc{CSP}s with non-binary alphabets,
it is sufficient to use techniques extending those for the corresponding
version of the SETH given in \cite{Lampis20}. However, things become more
complicated in \cref{sec:cspmulti} where we consider a promise \textsc{CSP}
that will be invaluable later. Roughly speaking, we consider here \textsc{CSP}
instances where we are also given a path decomposition of the primal graph and
an assignment may be allowed to cheat by flipping the value of a variable
between bags. This is of course too powerful, so we focus on cheating
assignments which must be monotone (that is, they can only increase values as
we move along the decomposition) and consistent for some small set of
variables.  What we need is to establish that it is hard to distinguish two
cases: there is a normal satisfying assignment; no satisfying assignment
exists, even if we cheat in the way described. Though this type of problem may
seem artificial, these are exactly the kind of properties we need to be able to
adapt SETH-based reductions to \ppseth-based reductions in \cref{sec:fpt}.
Thus, the main challenge of \cref{sec:robust} is establishing that this promise
problem is hard under the \ppseth\ (\cref{cor:weird}).

\subparagraph*{Supporting evidence.} We present two reductions to \textsc{SAT},
showing that if the \ppseth\ were false (that is, we had a fast \textsc{SAT}
algorithm) we would obtain algorithms falsifying the other conjectures. For the
reduction from the $k$-\textsc{Orthogonal Vectors} problem, we construct a
\textsc{SAT} instance with $k\log n$ main variables, which are meant to contain
the $k$ indices of the selected vectors.  We then need to add some machinery to
encode the instance, but the pathwidth of the produced formula is still roughly
$k\log n$, meaning that an algorithm that falsifies the \ppseth\ would solve
$k$-\textsc{Orthogonal Vectors} in time $(2-\eps)^{k\log n}=n^{(1-\delta)k}$.
For the reduction from the \scc\ we actually go through \textsc{Subset Sum}. It
is known that solving \textsc{Subset Sum} with target value $t$ and $n$ items
in time $t^{1-\eps}n^{O(1)}$ would falsify the conjecture
\cite{CyganDLMNOPSW16}. We encode such an instance with a \textsc{SAT} formula,
where a series of $tn$ variables encodes the value of the sum after we have
considered each item. Since each sum only depends on the previous sum, the
pathwidth is (essentially) $t$, and a fast \textsc{SAT} algorithm falsifies the
\scc.

\subparagraph*{Single-Exponential FPT Problems.} For each considered problem we
need to supply reductions to and from the \ppseth.  Showing that if the
\ppseth\ is false we obtain faster algorithms is straightforward in most cases,
though we must stress here once again that the fact that the \ppseth\ remains
equivalent if defined for the incidence graph is crucial (otherwise, it would
be hard to encode \textsc{Set Cover} and related problems). For the converse
direction, we need to reduce \textsc{SAT}, or more commonly a \textsc{CSP} with
an appropriate alphabet size, to our corresponding problem.  It now becomes
critical that we have shown that the promise \textsc{CSP} version discussed
above is a valid starting point, because this makes it vastly easier to prove
correctness for our gadgets (we explain this in more detail in
\cref{sec:cspmulti} and \cref{sec:fpt}). Although we often reuse the basic
constructions from existing SETH-based reductions, several of our reductions
are considerably simplified, compared to the original version, thanks to the
preparatory work of \cref{sec:robust}.

\subparagraph*{Super-Exponential FPT Problems.} Although much technical
preparation was needed to build up to the reductions of \cref{sec:fpt}, we had
the good fortune of being able to recycle numerous problem-specific gadgets
from previous reductions, because SETH-based lower bounds were already designed
with a special eye for efficiency. This is no longer the case when we move to
classes of problems for which only rough ETH-based lower bounds are known, such
as those dealt with in \cref{sec:fpt2}. Unfortunately, adapting the ETH-based
reductions is not good enough to obtain our results in this case. This is true
both on a high level (as we explain in \cref{sec:previous}), but also even when
one considers gadget-level constructions. The reason for this is that previous
bounds can afford to be a little wasteful, for instance reductions for
$C_4$-\textsc{Hitting Set} and related problems typically construct a vertex
for each \emph{literal} of a given \textsc{SAT} formula, rather than a vertex
for each \emph{variable} \cite{Pilipczuk11,SauS21}. This extra factor of $2$ is
acceptable if one does not care about the precise multiplicative constant in
the exponent, but is not acceptable in our setting. We are therefore forced to
redesign existing lower bounds, sometimes starting from scratch, with a number
of problem-specific gadgets.  Thankfully, the reductions in the other direction
(from our problems to \textsc{SAT}) are straightforward. It is worth noting
that the fact that we are dealing with linear structures is again crucial: the
natural DP for $C_4$-\textsc{Hitting Set} would require Fast Subset Convolution
techniques to obtain an optimal running time for tree decompositions, hence
even though for pathwidth this problem is \ppseth-equivalent, for treewidth
improving the best algorithm for this problem appears harder than improving the
best \textsc{SAT} algorithm (we discuss this in more detail in
\cref{sec:previous}).

\subparagraph*{XNLP-complete Problems.} Again, we need to supply reductions in
both directions and for two of the three problems considered (\textsc{List
Coloring} and \textsc{Independent Set Reconfiguration}) only rough ETH-based
lower bounds are known, so we need to invest some technical work to design
problem-specific gadgets. The intuition of the reductions from \textsc{SAT} is
that we want to reduce to instances where the new parameter has value $k$ and
the size of the instance is roughly $2^{\pw/k}$, where $\pw$ is the pathwidth
of the input CNF formula. Here, we have to be extremely careful: if the new
parameter value is not $k$ but say $k+1$, we cannot hope to obtain a lower
bound that matches the best algorithm up to the correct additive constant.
Intuitively, this is why the only problem for which we fail to obtain such a
super-precise bound is the one where we parameterize by pathwidth -- here minor
details such as the fact that the pathwidth is not equal to the size of the
largest bag, but equal to that value minus one actually become important.
Similarly, we have to be very careful so that the size of the graph is of the
order $2^{\pw/k}$ and not simply $2^{O(\pw/k)}$; in particular, standard
\textsc{List Coloring} gadgets, such as those used in \cite{BodlaenderGNS21}
fail here, because they are inherently quadratic, and we are forced to design
new ones. Interestingly, reductions in the other direction (to \textsc{SAT}) do
not need to be as stingy: ideally we would like to reduce an instance with
parameter value $k$ to a \textsc{SAT} instance with pathwidth $p=k\log n$, but
achieving only $p=(k+O(1))\log n$ or even $p=(1+o(1))k\log n$ is good enough,
because a \textsc{SAT} algorithm with dependence $(2-\eps)^p$ is still fast
enough to give an $n^{(1-\eps)k}$ algorithm for our problems (for sufficiently
large $k$). Intuitively, this asymmetry in the rigidity of the reductions is
the reason why we obtain two seemingly different versions of our lower bounds,
which turn out to be equivalent.

\subsection{Previous work}\label{sec:previous}

The results we present in this paper fall within the larger FPT and XP
optimality research program, which lies at the intersection of parameterized
complexity and fine-grained complexity. The goal of this program is to take
problems for which the best algorithm has complexity of the form $f(k)n^{O(1)}$
(FPT) or $n^{f(k)}$ (XP) and determine precise bounds on the best function $f$
possible.  In this context, a large number of results is known assuming the
Exponential Time Hypothesis (ETH) \cite{ImpagliazzoPZ01}. We cite as a few
characteristic examples the results of
\cite{BergougnouxBBK20,BonamyKNPSW18,ChitnisFHM20,CyganPP16,FominGLSZ19,abs-2307-08149,HarutyunyanLM21,LampisMV23,LampisM17,LokshtanovMS18b,MarxM16,MarxP14}.
Our results are different from this line of work, because we seek more
fine-grained lower bounds, that is, we do not accept to leave an unspecified
constant in the exponent.

\subparagraph*{Why ETH-based bounds do not translate to our setting.} An astute
reader may rightfully complain here that the comparison we make between our
results and the line of work on ETH-based lower bounds is unfair.  Even though
it is true that, say the $\pw^{o(\pw)}n^{O(1)}$ lower bound on
\textsc{Coloring} given in \cite{LokshtanovMS18b} is less precise than the
$\pw^{(1-\eps)\pw}n^{O(1)}$ bound we give here, this is because
\cite{LokshtanovMS18b} is relying on a weaker, less fine-grained assumption, so
it is natural that the obtained result is less fine-grained. One may hope that
surely if we take the proof of \cite{LokshtanovMS18b}  and plug in a stronger
assumption (perhaps the SETH), this would give a bound as good as the one we
get here.

Unfortunately, this line of reasoning is false, because this reduction (and
typical reductions in this area) starts from a variant of $k$-\textsc{Clique}.
This is problematic not only for the practical reasons we have already
discussed (e.g. the reduction must preserve $k$ exactly and cannot afford to
modify it by a multiplicative constant), nor even for the reason that the SETH
is not known to imply any concrete lower bound for the complexity of
$k$-\textsc{Clique}; but for the simple reason that $k$-\textsc{Clique} does in
fact admit a non-trivial faster than brute-force algorithm with complexity
$n^{\frac{\omega k}{3}}$, where $\omega$ is the matrix multiplication constant.
What this means concretely is that even if we take previous reductions for
\textsc{Coloring} \cite{LokshtanovMS18b} and \textsc{List Coloring}
\cite{FellowsFLRSST11} and optimize them so that the parameter is preserved
exactly, the reductions will never be able to rule out a \textsc{Coloring}
algorithm running in $\pw^{\frac{\omega\cdot\pw}{3}}n^{O(1)}$ or a \textsc{List
Coloring} algorithm running in $n^{\frac{\omega\cdot\pw}{3}}$, because such an
algorithm would at best imply a $k$-\textsc{Clique} algorithm that is already
known to exist. 

\medskip

\subparagraph*{SETH-based lower bounds.} Closer to what we do in this paper is
a line of research that tackles FPT optimality questions assuming the SETH.  As
mentioned, the question to determine the best base $c$ for which various
problems admit $c^kn^{O(1)}$ algorithms, for various parameters, started with
\cite{LokshtanovMS18}. Since then, numerous such results have been found, with
some examples being given in the following works:
\cite{BojikianCHK23,CurticapeanM16,FockeMINSSW23,FockeMR22,GanianHKOS22,HegerfeldK23,HegerfeldK23b,JaffkeJ23,MarxSS21,OkrasaR21}.
Relatively fewer results are known based on the SETH for determining the best
complexity of an XP algorithm, with a notable exception being the classical
result of P{\u{a}}tra{\c{s}}cu and Williams stating that if the SETH is true
then $k$-\textsc{Dominating Set} cannot be solved in time $n^{k-\eps}$
\cite{PatrascuW10}. A similar result of this type for (a version of) the DFA
intersection problem we consider in this paper was given by Oliveira and Wehar
\cite{OliveiraW20}, who showed that if the SETH is true then this problem
cannot be solved in time $n^{k-\eps}$.

\subparagraph*{Why focus on pathwidth and not treewidth.} Among works in this
area, it is worthwhile to particularly discuss the work of Iwata and Yoshida
\cite{IwataY15}, which is a major antecedent of our work.  The main animating
idea of \cite{IwataY15} is to observe, exactly as we do here, that a basic
weakness of SETH-based lower bounds is that reductions are one-way. Iwata and
Yoshida then go on to consider the treewidth variant of the \ppseth\ as a
complexity hypothesis, with the goal of showing that improving the basic DP
algorithm for various problems parameterized by treewidth (or clique-width) is
equivalent to this hypothesis.  Despite the fact that the broad thrust of their
approach is exactly the same as what we attempt in this paper, the scope of the
results they present is significantly more restricted: Iwata and Yoshida are
able to establish equivalence to their hypothesis only for problems for which
the complexity has the form $2^\tw n^{O(1)}$ and indeed, the only such problem
in their list that is not a \textsc{SAT} variant is \textsc{Maximum Independent
Set}.  Nevertheless, their work is a major inspiration for what we do here and
some of their ideas are extended and reused in the reductions of
\cref{sec:fpt}.  In addition to their results on treewidth, Iwata and Yoshida
also consider the \ppseth\ (under a different name) and give a Cook-Levin type
theorem proving that falsifying the \ppseth\ is equivalent to the problem of
deciding whether a non-deterministic Turing machine with space $k+O(\log n)$
will accept a given input, using time $(2-\eps)^kn^{O(1)}$. They call EPNL the
class of FPT problems solvable by such a machine. This result fits nicely in
the picture we present, as it gives further evidence that the \ppseth\ is a
natural complexity hypothesis.

At this point an interested reader may be wondering, given that the main idea
we build upon in this work was already present in \cite{IwataY15}, why are the
results of \cite{IwataY15} so much more limited than the results we present in
this paper? In particular, why does \cite{IwataY15} cover only problems with
complexity $c^{\tw}n^{O(1)}$, only for $c=2$, while we are able to extend the
reach of the \ppseth\ all the way to XNLP-complete problems? A partial answer
is that we rely on numerous (sometimes problem-specific) technical ingredients,
some of them introduced here (see \cref{sec:techniques}) and some of them
existing in the literature but introduced after \cite{IwataY15}.  This
considerable machinery allows us to deal with bases $c>2$ and with problems of
higher complexity.  However, we do not want to leave the answer at that,
because there is an important point to be made here regarding our choice to
focus on pathwidth.

Recall that one of the results we present in \cref{sec:fpt} is that solving
\textsc{Set Cover} in time $(2-\eps)^{\pw}n^{O(1)}$ is equivalent to falsifying
the \ppseth. This problem has a nice round complexity which superficially looks
like it would be a good target for \cite{IwataY15}, but is still conspicuously
absent from their results. We claim that the reason for this goes much deeper
than missing technical machinery: rather, what is relevant here is that the
best DP algorithm for \textsc{Set Cover} parameterized by treewidth is
\emph{not} ``simple'', but relies on Fast Subset Convolution techniques.  It is
far from clear whether a faster than brute-force \textsc{SAT} algorithm could
be useful to speed up such an algorithm, because such techniques rely on
solving a counting problem, indicating that to improve upon this complexity we
would need a fast algorithm for \#\textsc{SAT} or perhaps $\oplus$\textsc{SAT}.
This phenomenon is actually wide-spread and covers other famous and not so
famous problems parameterized by treewidth (e.g. \textsc{Dominating Set}  and
$C_4$-\textsc{Hitting Set}).  As a result, if we were working with a treewidth
variant of the \ppseth, as Iwata and Yoshida did, it would have been impossible
to obtain equivalence for these problems (barring some algorithmic
breakthrough), because for treewidth even these basic problems do not fit in
the paradigm of ``simple DP''. We therefore now see why it is important that in
\cref{sec:robust} we are able to prove that the \ppseth\ remains equivalent if
the parameter is the incidence pathwidth -- the same does not seem likely to
happen for treewidth, because again the best algorithm for \textsc{SAT}
parameterized by incidence treewidth needs Fast Subset Convolution. This
discussion indicates that the \ppseth\ is in a sense the more fruitful
hypothesis,  explains in part the wide reach of our results, and justifies our
choice to focus on linear structures.

\medskip

Finally, moving to less fine-grained previous work, we mention a recent series
of works, starting with \cite{BodlaenderGNS21}, that studies the class XNLP,
previously studied (under a different name) in \cite{ElberfeldST15} and in a
slight variation in \cite{Guillemot11}. XNLP is the set of problems decidable
by a non-deterministic Turing machine using FPT time and space $f(k)\log n$.
The class XNLP was recently shown to be the natural home of many W[1]-hard
problems for linear structural parameters, such as pathwidth and linear
clique-width.  Although the problems considered are directly relevant to what
we do, the class XNLP has so far been studied only from a coarse-grained
perspective and, since reductions inside this class are allowed to manipulate
the parameter as standard FPT reductions, XNLP-hardness for a problem does not
imply any concrete fine-grained running time bounds (other than fixed-parameter
intractability). However, there is a close connection between the spirit of
this work and this paper, because the underlying goal is to investigate the
complexity of problems parameterized by a linear structure on which one
normally performs DP and then give evidence that this DP algorithm is optimal.
In a similar spirit, we also mention the work of Drucker, Nederlof, and
Santhanam \cite{DruckerNS16} which, among other results, also considered the
question of optimality for DP algorithms for a subclass of algorithms they
called ``disjunctive dynamic programming''. They showed that
\textsc{Independent Set} parameterized by pathwidth is complete for this class,
which they model via a succinct CNF reachability problem, albeit with
reductions which are not fine-grained.

\subsection{Discussion and Directions for Further Work}

A natural question that arises from this work is what happens to versions of
the \ppseth\ for other widths, and in particular for treewidth. Indeed, we
believe that this is a very promising direction that merits further
investigation, in the same way that XNLP more recently gave rise to related
classes XALP and XSLP (for treewidth \cite{BodlaenderGJPP22} and tree-depth
respectively \cite{BodlaenderGP23}). However, it is important to be aware of
the fact that one of the obstacles likely to be faced in this direction is
that, as we discussed in \cref{sec:previous}, the primal and incidence graph
versions of the conjecture are not (easily) equivalent for treewidth. The
intuitive reason for this is that the ``obvious'' DP algorithm for \textsc{SAT}
parameterized by treewidth has parameter dependence $3^{\tw}$, while obtaining
an algorithm with dependence $2^{\tw}$ requires sophisticated techniques, such
as fast subset convolution.  We therefore expect the treewidth version of the
\ppseth\ to be equivalent to fewer problems than the \ppseth. 

The discussion above brings us back to the \ppseth\ and the question of which
natural problems parameterized by pathwidth are \emph{not} equivalent to this
hypothesis. We expect a typical example to be problems associated with the
Cut\&Count technique \cite{CyganNPPRW22}, because again the corresponding DP
algorithms are not ``simple'' and rely (in addition to randomization) on
solving a counting problem.  As a result, the blocking point for applying the
\ppseth\ here is not reducing \textsc{SAT} to the corresponding problem in a
way that preserves pathwidth, but rather reducing the corresponding
connectivity problem back to \textsc{SAT}. It therefore seems natural to
investigate the complexity of both counting and parity versions of the \ppseth\
and see if these hypotheses can be shown equivalent to (counting or parity
versions of) connectivity problems.  This type of hypothesis is likely to be of
interest also for handling techniques such as fast subset convolution,
mentioned in the previous paragraph. This of course leads to the natural
question: what is the relation between the \ppseth\ and the parity version
thereof? More concretely, suppose we have an algorithm falsifying the \ppseth.
Does this imply any improvement in the complexity of the best algorithm for
determining the parity of the number of satisfying assignments of a CNF
formula, parameterized by its primal pathwidth? This seems like a challenging
problem.

\section{Preliminaries}

We assume that the reader is familiar with the basics of parameterized
complexity, as given for example in \cite{CyganFKLMPPS15}. We will consider
problems parameterized by structural parameters, mainly pathwidth and linear
clique-width. In all cases, we are chiefly concerned with the complexity of
solving the problem, not that of computing the corresponding decomposition, so
we will generally assume that appropriate decompositions are supplied in the
input.

\paragraph*{Graph widths}

We recall the definitions of pathwidth and linear clique-width. A path
decomposition of a graph $G=(V,E)$ is an ordered sequence of bags $B_1,\ldots,
B_t$, where each $B_j$ for $j\in[t]$ is a subset of $V$, for each $xy\in E$
there exists a bag $B_j$ with $x,y\in B_j$, and for all $x\in V$ and
$j_1,j_2\in[t]$ we have that if $x\in B_{j_1}\cap B_{j_2}$, then $x\in B_{j'}$
for all $j'\in[j_1,j_2]$. The width of a path decomposition of a graph is the
maximum size of any bag minus one. The pathwidth of a graph $G$, denoted
$\pw(G)$, is the minimum width of any path decomposition. A path decomposition
is \emph{nice} if the symmetric difference of any two consecutive bags contains
at most one vertex.  Any path decomposition can be transformed in polynomial
time into a nice path decomposition without increasing the width. We also
recall the standard fact that if $G$ contains a clique, in any decomposition
there must be a bag that contains all vertices of the clique.

For an integer $k\ge1$, a \emph{linear clique-width expression} with $k$ labels
is a string made up of the following symbols:

\begin{enumerate}

\item (Introduce): $I(i)$, where $i\in[k]$

\item (Rename): $R(i\to j)$, where $i,j\in[k], i\neq j$

\item (Join): $J(i,j)$, where $i,j\in[k], i\neq j$

\end{enumerate}

We can associate with each such string a graph $G$ where each vertex is
assigned a label from $[k]$ inductively as follows: (i) the empty string is
associated with the empty graph (ii) for each non-empty string $\chi$ let
$\chi'$ be its prefix of length $|\chi|-1$ and $G'$ the graph associated with
$\chi'$. Then the graph associated with $\chi$ is defined as follows: (i) if
the last character of $\chi$ is $I(i)$, introduce in $G'$ an isolated vertex
with label $i$ (ii) if the last character of $\chi$ is $R(i\to j)$, modify the
labeling of $G'$ so that all vertices with label $i$ now have label $j$ (iii)
if the last character is $J(i,j)$, add in $G'$ all possible edges with one
endpoint with label $i$ and the other with label $j$. The \emph{linear
clique-width} of a graph $G$, denoted $\lcw(G)$, is the minimum $k$ such that
there exists a linear clique-width expression with $k$ labels associated with
some labeling of $G$.

\paragraph*{Primal and Incidence Graphs}

We will consider both satisfiability problems and more general constraint
satisfaction problems (\textsc{CSP}s) over larger alphabets. The \emph{primal
graph} of such an instance is a graph that has one vertex for each variable,
and an edge $xy$ if the variables $x,y$ appear together in some clause or
constraint.  The \emph{incidence graph} is the bipartite graph that has one
vertex for each variable and each clause (or constraint) and an edge $xc$
whenever variable $x$ appears in constraint $c$.

For a \textsc{SAT} or \textsc{CSP} instance $\phi$, we will denote by
$\pw(\phi)$ and $\pw^I(\phi)$ the pathwidth of its primal and of its incidence
graph respectively. We recall the standard fact that for all $\phi$ we have
$\pw^I(\phi)\le \pw(\phi)+1$. This property is not hard to see, because in the
primal graph every clause $c$ of $\phi$ is a clique, therefore a path
decomposition of the primal graph must contain a bag with all the variables of
$c$. We can therefore obtain a path decomposition of the incidence graph by
inserting after this bag, a copy where we also add $c$.

Note that the above imply that parameterizing by incidence pathwidth always
produces a \emph{harder} problem: if there is an algorithm deciding $\phi$ in
$c^{\pw^I}|\phi|^{O(1)}$, then there is an algorithm in $c^{\pw}|\phi|^{O(1)}$,
because $\pw^I(\phi)\le \pw(\phi)+1$.

The following lemma states that a path decomposition of the primal graph of any
CNF formula can efficiently be made nice and we can injectively assign for each
clause a bag that contains all its variables. The same lemma clearly holds for
CSPs, because the only property we need is that for each clause (or constraint)
the primal graph contains a clique on its variables.

\begin{lemma}\label{lem:nice} There is a linear-time algorithm that takes as
input a CNF formula $\psi$ with $n$ variables and $m$ clauses and a path
decomposition of its primal graph of width $p$ and outputs a nice path
decomposition $B_1,\ldots,B_t$ of $\psi$ containing at most $t=O(pm)$ bags, as
well as an injective function $b$ from the set of clauses of $\psi$ to $[t]$
such that for each clause $c$, $B_{b(c)}$ contains all the variables of $c$.
\end{lemma}

\begin{proof} 

Assume that the initial path decomposition we have contains $t'$ bags, numbered
from left to right $B_1,\ldots, B_{t'}$. For each clause $c$ of $\psi$ we
associate a distinct bag with index $b(c)$ that contains all variables of $c$
(note that we can easily make this association injective by repeating bags in
the path decomposition).  Such a bag always exists, because the variables of
$c$ form a clique in the primal graph, and it is a standard property of path
decompositions that every clique must appear as a subset of a bag.  It is not
hard to see that a bag $B_i$ such that there is no clause $c$ for which
$b(c)=i$ is redundant and can be removed from the decomposition. We can
therefore assume that $t'$ is equal to the number of clauses of $\psi$. We then
make the path decomposition nice by inserting between each $B_i$ and $B_{i+1}$
at most $2p$ bags, so that the decomposition remains valid and we have the
property that the symmetric difference of any two consecutive bags contains at
most $1$ variable.  We renumber the bags $B_1,\ldots, B_{t}$ and observe that
we now have at most $t'=O(pm)$ bags.  \end{proof}

The following lemma allows us to efficiently partition the variables of any CNF
formula in $k$ sets in a way that each bag is more-or-less equitably
partitioned.

\begin{lemma}\label{lem:pwcolor}

There is a linear-time algorithm that takes as input a graph $G=(V,E)$, a nice
path decomposition of $G$ of width $p$, and an integer $k$, and outputs a
partition of $V$ into $k$ sets $V_1,\ldots, V_k$ such that for all bags $B$ of
the decomposition and $i\in[k]$ we have $|V_i\cap
B|\le\lceil\frac{p+1}{k}\rceil$.

\end{lemma}

\begin{proof} 

We can achieve this via a simple greedy procedure: suppose that the bags of the
decomposition are numbered $B_0,B_1,\ldots,B_t$. Begin at $B_0$ and arbitrarily
partition its (at most $p+1$) vertices into the $k$ sets in an equitable way;
consider now a bag $B_j$ such that we have already partitioned all variables
which appear in a bag before $B_j$. Because we have a nice path decomposition,
we either have $B_j\subseteq B_{j-1}$, in which case there is nothing to do, or
$B_j=B_{j-1}\cup\{v\}$. We place $v$ into the set $V_i$ such that $|V_i\cap
B_{j-1}|$ is minimum.  It is not hard to see that $V_i$ cannot now contain more
than $\lceil\frac{p+1}{k}\rceil$ variables of $B_j$, because otherwise
$|B_{j-1}\cap V_i|\ge\lceil\frac{p+1}{k}\rceil$ for all $i\in [k]$, hence
$|B_{j-1}|\ge p+1$, which would imply that $|B_j|\ge p+2$, contradiction.
\end{proof}

\section{Robustness}\label{sec:robust}

The main topic of this paper is to investigate the implications of the \ppseth.
However, before we proceed it is natural to ask if we have formulated the
hypothesis correctly. In particular, there are several other natural ways in
which one may posit a hypothesis that seems ``morally'' the same as the
\ppseth. Our goal in this section is to explore several such variations and
show that the precise choices that we make in the definition of the \ppseth\
matter little. We consider this robustness as positive evidence that the
\ppseth\ captures a natural notion of computation and is worthy of further
study.

The variations one may consider are the following: instead of \tsat\ we could
formulate the \ppseth\ for $k$-\textsc{SAT} for $k$ sufficiently large, or even
for \textsc{SAT}; we could consider maximization variants where we seek to
maximize the number of satisfied clauses, or the number of True variables, or
both; we could consider restrictions on the number of variable occurrences; or
we could formulate any of the previous hypotheses using the pathwidth of the
incidence (rather than the primal) graph as a parameter.

It is worth noting that the type of modifications we discuss are not a priori
expected to be completely innocuous. Indeed, if one asks similar questions
regarding the \textsc{SETH}, one easily arrives either at formulations of the
hypothesis which are clearly false (e.g.~\textsc{3-SAT} can be solved
significantly faster than in $2^n$ time) or, even worse, where it is a wide
open problem whether the new version is equivalent to the original one -- we
state as examples that the standard formulation of the \textsc{SETH} is for
$k$-SAT for $k$ sufficiently large, but it it not known to be equivalent to the
formulation for \textsc{SAT}, and that the corresponding assumption for
\textsc{Max-SAT} is considered more plausible than the SETH itself (see
e.g.~\cite{AbboudBW15}).  Given this state of the art, we consider it quite
positive that the \ppseth\ does not present the same drawbacks. 

The first part of this section (\cref{sec:satrobust}) is devoted to proving
that the various formulations mentioned above are equivalent to the stated form
of the \ppseth, with the results summarized in \cref{thm:robust}. Then, in
\cref{sec:csprobust} we consider a different reformulation choice: given that
\textsc{SAT} can be seen as a constraint satisfaction problem over a binary
alphabet, we investigate the complexity of CSPs over a larger alphabet,
parameterized by the primal pathwidth. Here our motivation is more practical
than philosophical: in some of the reductions we perform later we will aim to
prove that the \ppseth\ is equivalent to solving a problem in time
$(c-\eps)^kn^{O(1)}$, for $c$ some constant and $k$ the parameter. This type of
reduction becomes quite tedious when $c$ is not a power of $2$, and indeed for
many natural problems parameterized by standard parameters (pathwidth, linear
clique-width, etc.) the base $c$ often takes such values. It therefore becomes
handy to have an equivalent formulation of the \ppseth\ for a version of the
problem where the natural running time is of the order $c^{\pw}$, for any
desired constant $c$. We achieve this by showing that solving CSPs of alphabet
size $c$ faster than $(c-\eps)^{\pw}$ is equivalent to the \ppseth, for all
$c\ge 3$.  This is similar to the $q$-\textsc{CSP}-$B$ problem, used in
\cite{Lampis20} with the same motivation.

We conclude this part in \cref{sec:cspmulti} where we extend our discussion of
CSPs to cover a seemingly artificial promise problem, with our motivation again
being strictly practical.  In this setting, we consider CSP instances with two
classes of variables, $V_1, V_2$, where we are also supplied with a path
decomposition of the primal graph in the input. The path decomposition is made
up mostly of variables of $V_1$ -- say each bag contains $O(\log n)$ variables
of $V_2$, so performing DP over them is not a problem. We want to decide the
following promise problem: in the Yes case, the \textsc{CSP} instance is
satisfiable as usual; while in the No case, it is not possible to satisfy it
even if we ``cheat'' by assigning to the variables of $V_1$ distinct values in
different bags, but with the constraint that values must increase as we move
along in the decomposition, and that variables of $V_2$ must be consistently
assigned.  The motivation for this strange version is that, as we explain in
detail, for most of the reductions of \cref{sec:fpt} (and more generally for
most SETH-based reductions for FPT problems), the state-of-the-art gadgets are
only able to guarantee monotonicity, rather than consistency, in the way that
CSP variables are represented. Hence, proving that this promise version of the
problem is still hard (that is, solving it faster than $(B-\eps)^{\pw}$
falsifies the \ppseth) will considerably simplify the reductions of
\cref{sec:fpt}, because in the converse direction we will not have to prove
that our gadgets force a consistent satisfying assignment, but only a monotone
one.  We provide such a proof in \cref{lem:weird} and \cref{cor:weird} in
\cref{sec:cspmulti}.

\subsection{Robustness for SAT}\label{sec:satrobust}

We consider several variations of \textsc{SAT}, including versions where we
seek an assignment that maximizes the number of clauses satisfied, or the
weight (number of True variables) of the satisfying assignment, or both. The
most general version of the problem we consider is the following:

\begin{definition} In the \textsc{Max-}$(t,w)$-$k$-\textsc{SAT} problem we are
given as input a $k$-CNF formula and two integers $t,w$ and are asked if the
formula admits an assignment that sets at least $w$ variables to True and
satisfies at least $t$ clauses.  \textsc{Max-}$k$-\textsc{SAT} is the
restriction of the problem when $w=0$, while \textsc{MaxW-}$k$-\textsc{SAT} is
the restriction of the problem where $t=m$ is the number of clauses. If we drop
$k$ from the name of a problem, this indicates that the input is a CNF formula
of unbounded arity. \end{definition}

The main theorem of this section is now the following:

\begin{theorem}\label{thm:robust}

All of the following statements are equivalent to the \ppseth\ (where in all
cases $\psi$ denotes the input formula).

\begin{enumerate}

\item For all $\eps>0$, \tsat\ restricted to instances where all variables
appear at most $4$ times cannot be solved in time
$O((2-\eps)^{\pw(\psi)}|\psi|^{O(1)})$.

\item For all $\eps>0$, \textsc{Max-2-SAT} restricted to instances where all
variables appear at most $16$ times cannot be solved in time
$O((2-\eps)^{\pw(\psi)}|\psi|^{O(1)})$.

\item For all $\eps>0$, \textsc{MaxW-3-SAT} restricted to instances where all
variables appear at most $4$ times cannot be solved in time
$O((2-\eps)^{\pw(\psi)}|\psi|^{O(1)})$.

\item For all $\eps>0$, \textsc{Max-}$(t,w)$-2-\textsc{SAT} restricted to
instances where all variables appear at most $16$ times cannot be solved in
time $((2-\eps)^{\pw(\psi)}|\psi|^{O(1)})$.

\end{enumerate} 

Furthermore, the same holds even if we modify any of the statements by making
any of the following changes: we replace $\pw(\psi)$ by $\pw^I(\psi)$; we drop
the restriction on the number of variable appearances; or we consider $k$-CNF
formulas, for any $k\ge3$, or CNF formulas of unbounded size.

\end{theorem}

\begin{proof}

Because we allow any combination of the mentioned modifications, this statement
refers to numerous different variations of \textsc{SAT}. In order to organize
things, we start with the four base problems given in the statement: \tsat,
\textsc{Max-2-SAT}, \textsc{MaxW-3-SAT}, and
\textsc{Max-}$(t,w)$-2-\textsc{SAT}. We observe that any of the allowed
modifications can only make the considered problem harder: dropping the
restriction on the number of variable appearances or the arity of clauses
trivially generalizes the problem, while the inequality $\pw^I(\phi)\le
\pw(\phi)+1$, which holds for any $\phi$, ensures that parameterizing by
incidence pathwidth is always at least as hard as by primal pathwidth.

Furthermore, we observe that if we drop the arity constraint,
\textsc{Max-}$(t,w)$-\textsc{SAT} is a problem that generalizes the other three
problems. Furthermore, \textsc{MaxW-3-SAT} generalizes \tsat\ and
\textsc{Max-}$(t,w)$-2-\textsc{SAT} generalizes \textsc{Max-2-SAT}. We have
therefore arrived at a single problem which has maximal complexity among all
problems considered (\textsc{Max-}$(t,w)$-\textsc{SAT} without restrictions on
variable occurrences, parameterized by $\pw^I$) and two problems which have
minimal complexity (\tsat\ and \textsc{Max-2-SAT} with the restrictions stated
in the theorem). We will reduce the hardest case of
\textsc{Max-}$(t,w)$-\textsc{SAT} to the easiest case of \tsat, and then 
\tsat\ to the easiest case of \textsc{Max-2-SAT}, using fine-grained reductions
that preserve the type of algorithm we are interested in, and thus we will
establish that all problems in our group have equivalent complexity and are
therefore equivalent to the \ppseth.

To ease presentation we give the claimed reductions as a collection of
self-contained lemmas in \cref{sec:lemrobust}. Composing \cref{lem:psi1},
\cref{lem:psi2}, \cref{lem:psi3}, and \cref{lem:psi4} we obtain a
polynomial-time algorithm that takes as input an instance $\phi$ of
\textsc{Max-}$(t,w)$-\textsc{SAT} and a path decomposition of its incidence
graph of width $\pw^I(\phi)$ and produces an equivalent instance $\psi$ of
\tsat\ where each variable appears at most 4 times and a path decomposition of
its primal graph of width $\pw^I(\phi)+O(\log |\phi|)$. As a result, solving
the resulting \tsat\ instance in $((2-\eps)^{\pw(\psi)}|\psi|^{O(1)})$ would
produce an algorithm running in $((2-\eps)^{\pw^I(\phi)}|\phi|^{O(1)})$ for the
original problem. The reduction from \tsat\ to \textsc{Max-2-SAT} given in
\cref{lem:garey} is just the classical reduction from \cite{GareyJS76}, where
we also observe that primal pathwidth is preserved.  \end{proof}

\cref{thm:robust} shows that the \ppseth\ is quite robust. Let us conclude this
discussion by adding two relevant remarks. First, regarding the minimum number
of variable occurrences needed to obtain hardness, we have not made an effort
to optimize the constant $16$ for \textsc{Max-2-SAT}; the constant we obtain is
a consequence of using the classical reduction of Garey, Johnson, and
Stockmeyer \cite{GareyJS76}. We have, however, made an effort to optimize this
constant for \tsat. There appears to be a natural reason why it is not easy to
get this constant down to $3$, which would be the expected barrier: the
standard polynomial time reduction which achieves this replaces every variable
$x$ with a cycle of implications $x_1\to x_2\to\ldots\to x_k\to x_1$.
Performing this reduction then risks doubling the pathwidth, instead of
increasing it by an additive constant, because in every bag where $x$ used to
appear we are likely to be forced to place two of the new variables, in order
to separate the cycle.  We consider it an intriguing open problem whether
\tsat\ can be solved more efficiently when each variable appears at most 3
times.  Second, we have only established hardness for \textsc{MaxW-SAT} for
arity at least $3$.  This is because when the arity is $2$ the problem is
essentially \textsc{Max Independent Set}, which is handled in a different
section, so the problem for arity $2$ is in fact still equivalent to the
\ppseth.

\subsubsection{Lemmas Needed for \cref{thm:robust}}\label{sec:lemrobust}

We present here a sequence of reductions needed to prove \cref{thm:robust}.
These work as follows:

\begin{enumerate}

\item \cref{lem:psi1} takes as input an instance of the problem where we
parameterize by $\pw^I$ and seek to optimize both the weight of the selected
assignment and the number of satisfied clauses. It produces an equivalent
instance where we parameterize by the pathwidth of the primal graph $\pw$ and
we only need to optimize the weight of the assignment (while satisfying all
clauses).  This is achieved by replacing clauses in the decomposition by new
variables that indicate whether a clause was satisfied by a literal.  The
pathwidth increases by an additive constant.

\item \cref{lem:psi2} takes the previous instance and produces an equivalent
instance where we are simply looking for a satisfying assignment. This is
achieved by adding $n\log n$ new variables which act as counters. Pathwidth
increases by $O(\log n)$.

\item \cref{lem:psi3} is just the standard reduction from \textsc{SAT} to
\tsat, performed in a way that increases pathwidth by at most an additive
constant.

\item \cref{lem:psi4} is a variation of the standard reduction that reduces the
number of occurrences of each variable. As explained after \cref{thm:robust} we
cannot use the standard trick, as we then risk doubling the pathwidth. We use a
slightly more complicated construction where each variable appears at most $4$
times and pathwidth increases by at most an additive constant.

\item \cref{lem:garey} is the standard reduction from \tsat\ to
\textsc{Max-2-SAT} from \cite{GareyJS76}, with the observation that pathwidth
increases by at most an additive constant.

\end{enumerate}

\begin{lemma}\label{lem:psi1}

There exists a polynomial-time algorithm which takes as input a CNF formula
$\phi$ and a path decomposition of its incidence graph of width $p$, as well as
two integers $t,w$, and produces a CNF formula $\psi$, a path decomposition of
the primal graph of $\psi$, and two integers $w_1,w_2$ such that we have the
following:

\begin{enumerate}

\item There exists an assignment setting at least $w$ variables of $\phi$ to
True and satisfying at least $t$ clauses if and only if $\psi$ admits a
satisfying assignment whose number of True variables is in the interval
$[w_1,w_2]$.

\item The path decomposition of the primal graph of $\psi$ has width $p+O(1)$.

\end{enumerate}

\end{lemma}

\begin{proof}

Suppose that $\phi$ has $n$ variables and $m$ clauses. We will construct $\psi$
by taking the variables of $\phi$ and adding some new variables; we will then
construct some clauses to represent the clauses of $\phi$.  Along the way, we
will edit the given path decomposition, removing the clauses of $\phi$ and
inserting into it the new variables, in a way that preserves the width and
covers all new clauses.

We begin by editing the given path decomposition of the incidence graph of
$\phi$ so that it becomes nice: if the symmetric difference of two consecutive
bags $B_j, B_{j+1}$ contains more than one element, we insert between them a
sequence of bags that removes one by one the elements of $B_j\setminus B_{j+1}$
and then adds one by one the elements of $B_{j+1}\setminus B_j$. Once we do
this, we number the $n$ variables of $\phi$ as $x_1,\ldots, x_n$, in such a way
that for all $i_1<i_2$ the first bag that contains $x_{i_1}$ is before the
first bag that contains $x_{i_2}$.  Observe that this is always possible, as
each bag introduces at most one new vertex.  Similarly, we renumber the $m$
clauses $c_1,\ldots, c_m$, such that clauses with lower indices appear before
clauses with higher indices in the decomposition. We now construct an injective
function $b$ from the edges of the incidence graph of $\phi$ to the indices of
the bags of our decomposition which satisfies the following: for all $j\in[m],
i_1,i_2\in[n]$ with $i_1<i_2$ such that $c_j$ contains $x_{i_1}, x_{i_2}$, we
have $b(c_jx_{i_1})<b(c_jx_{i_2})$, that is, edges incident on variables with
lower indices are mapped earlier in the decomposition. It is not hard to
achieve this, because if a bag contains $\{c_j,x_{i_1},x_{i_2}\}$ then we can
make two copies of this bag and map $c_jx_{i_1}$ to the first and $c_jx_{i_2}$
to the second, while otherwise every bag that contains $\{c_j,x_{i_1}\}$ (and
such a bag must exist) is before every bag containing $\{c_j,x_{i_2}\}$ by the
numbering of the variables. In the remainder, suppose that the bags of the path
decomposition are numbered $B_1,B_2,\ldots$.

We now do the following for each clause $c_j$, for $j\in[m]$. Suppose that
$c_j$ contains $r$ variables, $x_{j_1}, x_{j_2}, \ldots, x_{j_r}$. We introduce
$(mn)^5$ new variables $c_{j,s}$ for $s\in\{0,\ldots,(nm)^5-1\}$.  Intuitively,
the meaning of $c_{j,s}$ for $s\le r$ is that $c_j$ has been satisfied by one
of the $s$ variables $x_{j_1}, x_{j_2},\ldots, x_{j_s}$.  We add for $s\in[r]$
the clauses $(\neg c_{j,s} \lor (\neg)x_{j,s} \lor c_{j,s-1})$, where we use
the literal $x_{j,s}$ or $\neg x_{j,s}$ that appears in $c_j$; the clause
$(\neg c_{j,0})$; and for $s\in\{r+1,\ldots,2n^5-1\}$ we add the clauses $(\neg
c_{j,s} \lor c_{j,s-1})$ and $(\neg c_{j,s-1} \lor c_{j,s})$.  

Observe that the clauses we have added have the following properties: fix an
assignment to $x_1,\ldots,x_n$ and consider a clause $c_j$. If the assignment
does not satisfy $c_j$ in $\phi$, then all $(mn)^5$ variables we have added
must be set to False in any satisfying assignment of $\psi$, because all the
literals involving variables $x_{j_1},\ldots,x_{j_r}$ are falsified.  If the
assignment does satisfy $c_j$, then there is always a way to extend it to an
assignment that sets at least $(mn)^5-r>(mn)^5-n$ of the new variables to True.
Furthermore, if a $c_{j,s}$ with $s>n$ is set to True, then indeed at least
$(mn)^5-n$ of the new variables must be set to True.

To complete the construction, for each variable $x_i$, we introduce $(mn)^2$
new variables $x_{i,s},\ldots, x_{i,(mn)^2}$ and add the clauses $(\neg x_{i,s}
\lor x_{i,s-1})$ for all $s\in\{2,\ldots,n^2\}$, as well as the clause $(\neg
x_{i,1} \lor x_i)$.  These variables have the property that if $x_i$ is False,
all of them must be set to False.

Finally, we set $w_1 = t((mn)^5-n)+w((mn)^2+1)$ and $w_2 =
t(mn)^5+n((mn)^2+1)$. 

We have now described $\psi$. To describe how the path decomposition of its
primal graph is constructed, begin with the path decomposition of the incidence
graph of $\phi$ and remove all the clauses from the bags. We now need to
include the new variables and cover all the newly constructed clauses. To do
this, first, for each $j\in[m]$ consider clause $c_j$ of $\phi$  which contains
variables $x_{j_1},\ldots, x_{j_r}$. For each $s\in\{2,\ldots,r\}$ we place
variable $c_{j,s-1}$ in all bags that contained $c_j$ in the original
decomposition and have indices in the interval $[b(c_jx_{j_{s-1}}),
b(c_jx_{j_s})]$. Furthermore, we place $c_{j,0}$ in bag $B_{b(c_jx_{j_1})}$ and
$c_{j,r}$ in bag $B_{b(c_jx_{j_r})}$.  Observe that this covers the clauses of
size $3$ we have constructed for $c_j$. To cover the remaining clauses, take
the bag $B_{b(c_jx_{j_r})}$ (which contains $c_{j,r}$ and $c_{j,r-1}$), and
insert after it $(mn)^5-r$ copies, removing $c_{j,r}$ and $c_{j,r-1}$ from
each. For $s\in\{r,\ldots,(mn)^5-2\}$ we add in each bag the variables
$c_{j,s}, c_{j,s+1}$. This covers all clauses added for clauses of $\phi$ and
we observe that in all modified bags of the original decomposition we have
removed one element ($c_j$) and added one or two elements. Furthermore, the
bags where the size has increased are bags where $b$ maps an edge incident on
$c_j$, or newly constructed bags.  Since $b$ is an injective function, this
ensures that performing this process over all $j\in [m]$ will never increase
the width of the decomposition by more than $1$.

To cover the variables and clauses added for variables of $\phi$ in the path
decomposition, for each $i\in[n]$ we locate a bag that contained $x_i$ in the
decomposition of the primal graph of $\phi$ and insert immediately after it a
sequence of $(mn)^2$ copies, adding to each a consecutive pair $x_{i,s},
x_{i,s+1}$. Again, this does not increase the width of any old bag, while new
bags have at most two extra elements each. We conclude that the width of the
path decomposition of the primal graph of $\psi$ is $p+O(1)$.

To argue for correctness, for the forward direction, if we have a good
assignment for $\phi$, start with the same assignment for $\psi$. Select $t$ of
the satisfied clauses of $\phi$ and set the corresponding $c_{j,s}$ variables
so that at least $(mn)^5-n$ variables are set to True for each clause; for the
remaining clauses set all $c_j$ variables to False. Finally, for each $x_{i,s}$
give it the same value as $x_i$. It is not hard to see that this assignment is
satisfying for $\psi$ and has weight in $[w_1,w_2]$.  

For the converse direction, we claim that if there is a satisfying assignment
of $\psi$ whose number of True variables is in $[w_1,w_2]$, there are exactly
$t$ clauses for which some $c_{j,s}$ variables with $s>n$ have been set to
True. This is because $(t+1)((mn)^5-n)>w_2$ and $(t-1)(mn)^5<w_1$, for
sufficiently large $n,m$. We now observe that if $\psi$ has a satisfying
assignment with the claimed weight, restricting the assignment to the variables
of $\phi$ must set at least $w$ variables to True, because
$t(mn)^5+(w-1)((mn)^2+1) < w_1$.  \end{proof}

\begin{lemma}\label{lem:psi2}

There exists a polynomial-time algorithm that takes as input a CNF formula
$\phi$ and a path decomposition of its primal graph of width $p$, as well as
two integers $w_1,w_2$ and outputs a CNF formula $\psi$ and a path
decomposition of its primal graph of width $p+O(\log(|\phi|)$, such that $\psi$
is satisfiable if and only if $\phi$ has a satisfying assignment whose number
of True variables is in the interval $[w_1,w_2]$.

\end{lemma}

\begin{proof}

Let $n$ be the number of variables of $\phi$.  Assume without loss of
generality that $n$ is a power of $2$ as this can be achieved by adding dummy
variables and clauses that require them to be False. Suppose that the bags of
the given decomposition are numbered $B_1,B_2,\ldots$, the decomposition is
nice, and that the variables are numbered $x_1,\ldots, x_n$ in a way that
variables with higher indices are introduced in later bags (this is the same as
in the proof of \cref{lem:psi1} and can also be achieved using \cref{lem:nice}
and renumbering the variables in the order in which they were introduced).  We
will construct $\psi$ by adding some variables and clauses to $\phi$.

We now introduce $n\log n$ variables $w_{i,s}$ for $i\in[n], s\in[\log n]$.
These should be thought of as counters: for a fixed $i$ the variables $w_{i,s}$
tell us how many among the variables $x_1,\ldots, x_i$ are True in the current
assignment.  For each $i\in\{2,\ldots,n\}$ we add clauses which for each two
assignments $\sigma,\sigma'$ to variables $w_{i-1,s}$ and $w_{i,s}$
respectively, ensure that $\sigma,\sigma'$ cannot be selected unless $\sigma$
encodes a binary number $s_1$, $\sigma'$ encodes a binary number $s_2$, and
$s_2=s_1+x_i$, where we interpret the boolean value of $x_i$ as 1 if True and 0
if False. These are at most $O(n^2)$ clauses per variable, so $O(n^3)$ clauses
overall. We also add clauses ensuring that $w_{1,s}$ encodes an assignment that
has binary value $0$ or $1$, representing the value of $x_1$; and clauses that
ensure that $w_{n,s}$ encodes a binary number in $[w_1,w_2]$.  

This completes the construction of $\psi$ and it is not hard to see that $\psi$
is satisfiable if and only if $\phi$ had a satisfying assignment with weight in
$[w_1,w_2]$.  

What remains is to show that we can edit the decomposition so that we have
width $p + O(\log n)$ while covering all added clauses.  Let $b$ be an
injective function which for each $x_i$ returns the index of the bag in which
$x_i$ is introduced.  For $i\in\{2,\ldots,n\}$ We add the variables $w_{i,s}$
in all bags with indices in the interval $[b(x_{i-1}),b(x_i)]$. We also add all
variables $w_{i,1}$ in $B_{b(x_1)}$.  Observe that this adds at most $2\log n$
variables to each bag and covers all added clauses.  \end{proof}

\begin{lemma}\label{lem:psi3}

There exists a polynomial time algorithm which takes as input a CNF formula
$\phi$ and a path decomposition of its primal graph of width $p$ and outputs a
3-CNF formula $\psi$ and a path decomposition of its primal graph of width
$p+O(1)$, such that $\phi$ is satisfiable if and only if $\psi$ is.

\end{lemma}

\begin{proof}

As in the proof of \cref{lem:psi1} and \cref{lem:psi2} we can assume that the
given path decomposition is nice, its bags are numbered $B_1, B_2,\ldots$ and
that we number the variables of $\phi$ as $x_1,x_2,\ldots, x_n$ so that
variables with higher indices are introduced in later bags (again, this can be
achieved using \cref{lem:nice} and renumbering variables in the order in which
they are introduced).  Furthermore, suppose we have an injective function $b$
which takes as input a clause $c_j$ and returns the index of a bag that
contains all the literals of $c_j$ (again, guaranteed by \cref{lem:nice}).
Suppose that clauses are numbered $c_1,\ldots, c_m$ with $b$ mapping clauses
with higher indices to later bags.

We will apply the standard reduction from \textsc{SAT} to \tsat, taking care
not to increase the width of the path decomposition. For each $j\in[m]$, if the
arity of $c_j$ is $k>3$ we introduce $k-2$ new variables $z_{j,1},\ldots,
z_{j,k-2}$. Suppose that $c_j=(\ell_1 \lor \ell_2\lor \ldots \lor \ell_k)$,
where the $\ell_j$ are literals. We replace $c_j$ with the clauses $(\ell_1\lor
\ell_2\lor z_{j,1}) \land (\neg z_{j,1}\lor \ell_3\lor z_{j,2}) \land \ldots
\land (\neg z_{j,k-2} \lor  \ell_k)$. It is standard that this preserves
satisfiability.

What remains is to argue that we have not increased the width of the
decomposition significantly. For each $j\in[m]$ for which we performed the
above, if $c_j$ has arity $k>3$ we insert in the decomposition immediately
after $B_{b(c_j)}$ another $k-3$ bags, each containing the vertices of
$B_{b(c_j)}$. We place into the first such bag $z_{i,1}, z_{i,2}$, the
following bag $z_{i,2}, z_{i,3}$, and so on, covering all new clauses. Because
$b$ is injective, this will increase the size of the largest bag by at most
$2$.  \end{proof}

\begin{lemma}\label{lem:psi4}

There exists a polynomial time algorithm which takes as input a 3-CNF formula
$\phi$ and a path decomposition of its primal graph of width $p$ and outputs a
3-CNF formula $\psi$ and a path decomposition of its primal graph of width
$p+O(1)$, such that each variable of $\psi$ appears at most $4$ times and
$\phi$ is satisfiable if and only if $\psi$ is.

\end{lemma}

\begin{proof}

We begin by invoking \cref{lem:nice}. As in the proof of \cref{lem:psi3}, the
bags of the decomposition are numbered $B_1, B_2,\ldots$, the variables are
numbered in the order in which they are introduced, there exists an an
injective function $b$ mapping clauses to bags that contain all their
variables, and clauses are numbered so that $b$ maps clauses with higher
indices to later bags.

We will construct $\psi$ by editing $\phi$ and adding some new variables and
clauses.  For each $i\in [n]$, if $x_i$ appears $k>4$ times in $\phi$, suppose
that $x_i$ appears in clauses $c_{i_1}, c_{i_2},\ldots, c_{i_k}$, with
$i_1<i_2<\ldots<i_k$. We introduce $k$ new variables $x_{i,k}$ and replace
$x_i$ in $c_{i_j}$ with $x_{i,j}$ for all $j\in [k]$. We also add $k-1$ new
variables $y_{i,s}$, for $s\in[k-1]$. For each $s\in [k-2]$ we add the clause
$(\neg y_{i,s}\lor y_{i,s+1})$. For each $s\in[k-1]$ we add the clause $(\neg
y_{i,s}\lor x_{i,s})$. For each $s\in \{2,\ldots,k\}$ we add the clause $(\neg
x_{i,s}\lor y_{i,s-1})$. We also add the clauses $(\neg x_{i,1} \lor y_{1,1})$,
$(\neg y_{i,k-1} \lor x_{i,k})$. The intuitive way to think about this
construction is that we have a sequence of variables $x_{i,1}, y_{i,1},
x_{i,2}, y_{i,2},\ldots, y_{i,k-1}, x_{i,k}$, where each variable implies the
previous one in the sequence, while each $y$ variable implies the next $y$
variable (and we also add the missing implications to ensure consistency for
the first and last $x$ variable). It is not hard to see that each variable
appears at most $4$ times in $\psi$ if we perform this transformation
exhaustively. Furthermore, since the only way to satisfy the new clauses is to
give the same assignment to all $x_{i,s}$, $\psi$ is equisatisfiable with
$\phi$.

What remains is to edit the path decomposition so that we cover the new
clauses. For each $i\in[n]$ for which we applied the above transformation to
$x_i$, which appeared in $k$ clauses $c_{i_1},\ldots, c_{i_k}$, we do the
following: remove $x_i$ from the decomposition, since it no longer appears in
the formula; for each $j\in[k]$ place $x_{i,j}$ in $B_{b(c_{i_j})}$ so that
$B_{b(c_{i_j})}$ still contains all variables of $c_{i_j}$; for each $j\in
[k-1]$ place $y_{i,j}$ in all bags with indices in the interval
$[b(c_{i_j}),b(c_{i_{j+1}})]$. Note that this transformation may increase the
size only of the bags on which $b$ maps a clause, because the bags with indices
in the interval $(b(c_{i_j}), b(c_{i_{j+1}}))$ used to contain $x_i$ but now
contain $y_{i,j}$ in its place.  Since $b$ is injective, the size of a bag
$B_{b(c_j)}$ may increase at most 3 times, once for each variable of $c_j$.
\end{proof}

\begin{lemma}\label{lem:garey}

There exists a polynomial-time algorithm that takes as input a 3-CNF formula
$\phi$ where each variable appears at most $r$ times and output a 2-CNF formula
$\psi$ where each variable appears at most $4r$ times and an integer $t$ such
that $\phi$ is satisfiable if and only if $\psi$ has an assignment satisfying
at least $t$ variables. Furthermore, $\pw(\psi) \le \pw(\phi)+1$.

\end{lemma}

\begin{proof} 

We use the classical reduction of Garey, Johnson, and Stockmeyer
\cite{GareyJS76}.  Suppose we have a 3-CNF instance $\phi$ where each variable
appears $r$ times. For each clause $c=(\ell_1\lor \ell_2\lor \ell_3)$ we
produce 10 clauses of size at most 2, using a new variable $d_c$: $\{ (\ell_1),
(\ell_2), (\ell_3), (d_c), (\neg \ell_1\lor \neg \ell_2), (\neg\ell_2\lor \neg
\ell_3), (\neg \ell_1\lor \neg\ell_3), (\ell_1\lor\neg d_c), (\ell_2\lor \neg
d_c), (\ell_3\lor\neg d_c)\}$. An assignment that satisfies the original clause
can be extended to satisfy 7 of the new clauses, while other assignments can
satisfy at most 6, so we can compute a target number of clauses $t$ to satisfy
that is achievable if and only if the original instance is satisfiable. It is
not hard to see that now each variable appears at most $4r$ times.

To bound the pathwidth of $\psi$ recall that, given any path decomposition of
$\phi$ we can construct an injective function that maps clauses to bags that
contain all their literals (\cref{lem:nice}). For each clause $c$ we add the
variable $d_c$ to the bag $B_{b(c)}$, and this increases the width by at most
$1$.  \end{proof}

\subsection{Equivalence for Larger Alphabets}\label{sec:csprobust}

In this section we consider constraint satisfaction problems (CSPs) with
alphabets of size larger than $2$. A \textsc{CSP} instance over an alphabet of
some fixed size $B\ge 2$ is made up of $n$ variables which, without loss of
generality, take values in $[B]$. We are given as input a collection of $m$
constraints.  Each constraint is described by an ordered tuple of the variables
involved in the constraint plus a list of all the acceptable assignments to the
variables that satisfy this constraint. The arity of a constraint is the number
of variables involved in it. We will in general only focus on CSPs where the
arity $r$ is a fixed constant, so the space required to describe a constraint
will be $O(B^r) = O(1)$. We denote by $r$-\textsc{CSP} the restriction of the
\textsc{CSP} problem to constraints of maximum arity $r$. 

We will also consider a slight generalization of the same problem. Suppose we
are supplied with a fixed weight function $w:[B]\to \mathbb{N}$. The
\textsc{MaxW-CSP} problem is the problem of determining whether a given
\textsc{CSP} instance has a satisfying assignment whose total weight is at
least a given target $w_t$. Here, when a variable is assigned a value $v\in B$,
its weight is $w(v)$, and the weight of an assignment is just the sum of the
weights of all the variables.

In the remainder we are chiefly concerned with the case $B\ge 3$, as for $B=2$,
\tsat\ is already a \textsc{CSP} with alphabet of size $2$. We prove that the
\ppseth\ works as expected in the case of larger alphabets, even for CSPs where
all constraints involve at most $2$ variables.

\begin{theorem}\label{thm:csp}

For each $B\ge 3$ the following statements are equivalent:

\begin{enumerate}

\item The \ppseth\ is false.

\item There exist $\eps>0, c>0$ and an algorithm that takes as input a
\textsc{2-CSP} instance $\psi$ on alphabet $[B]$ and decides if $\phi$ is
satisfiable in time $O((B-\eps)^{\pw}|\psi|^c)$.

\item For each weight function $w:[B]\to\mathbb{N}$ and arity $r\ge2$, there
exist $\eps>0, c>0$ and an algorithm that takes as input a
\textsc{MaxW-}$r${-CSP} instance $\psi$ on alphabet $[B]$ and weight function
$w$ and decides $\psi$ in time $O((B-\eps)^{\pw}|\psi|^c)$.

\end{enumerate}

\end{theorem}

\begin{proof}

The implication 3$\Rightarrow$2 is trivial, because restricting the arity and
lifting the weight objective can only make the problem easier.  We prove
2$\Rightarrow$1 as \cref{lem:csp1} and 1$\Rightarrow$3 as \cref{lem:csp2}.
Observe that in the statements of the two lemmas we use different formulations
of the \ppseth, which are, however, known to be equivalent by
\cref{thm:robust}.  \end{proof}

Before we proceed, let us state a simple lemma that will be useful in various
places and allows us to decrease the arity of any \textsc{CSP} with non-binary
alphabet to $2$.

\begin{lemma}\label{lem:csparity} There is a polynomial-time algorithm that
takes as input a \textsc{CSP} instance $\psi$ over an alphabet of size $B\ge 3$
and outputs an equivalent 2-\textsc{CSP} instance $\psi'$ over the same
alphabet.  Furthermore, there is a way to transform any path decomposition of
the primal graph of $\psi$ of width $p$ to one for $\psi'$ with width at most
$p+3$.  \end{lemma}

\begin{proof}

We will construct $\psi'$ by considering each constraint $c$ of $\psi$ of arity
greater than $2$ in some order, and at each step replacing $c$ by a gadget over
the variables involved in $c$, some new variables, and some new constraints of
arity $2$, such that there is a way to extend an assignment of the original
variables to the new variables satisfying new constraints if and only if $c$
was satisfied by this assignment. Doing this exhaustively will drop the arity
to $2$.

For each constraint $c$ of $\psi$, let $\mathcal{S}_c$ be the set of all
satisfying assignments to this constraint, with $\ell_c=|\mathcal{S}_c|$. We
construct $\ell_c$ new variables $t_{c,\sigma}, \sigma\in\mathcal{S}_c$, and
another $\ell_c+1$ helper variables. We place binary constraints among the
$2\ell_c+1$ new variables, so that the primal graph forms an odd cycle where no
$t_{c,\sigma}$ are adjacent.  We adjust these constraints to ensure that: (i)
any two adjacent variables must receive distinct values (ii) $t_{c,\sigma}$
variables must receive a value in $\{1,2,3\}$ (iii) helper variables must
receive a value in $\{1,2\}$. For each variable $x$ implicated in the
constraint $c$ and each $\sigma\in\mathcal{S}_c$ we add a constraint involving
$x$ and $t_{c,\sigma}$ which is always satisfied if $t_{c,\sigma}$ has value
$1$ or $2$; while if $t_{c,\sigma}$ has value $3$ this constraint is only
satisfied if $x$ receives the value dictated by $\sigma$. After performing this
procedure, we remove all the old constraints and only keep the new constraints
of arity $2$. Note that exhaustively applying this procedure takes time
polynomial in $\psi$, as for each constraint $c$ we are given in the input
$\mathcal{S}_c$.

Equivalence between $\psi$ and $\psi'$ is easy to see: if we have an assignment
to $\psi$ that satisfies a constraint $c$, it must agree with some
$\sigma\in\mathcal{S}_c$; we assign value $3$ to $t_{c,\sigma}$ and values
$1,2$ to the other variables of the odd cycle, satisfying all new constraints.
In the converse direction, since the $2\ell_c+1$ new variables form an odd
cycle, at least one variable $t_{c,\sigma}$  must have value $3$. But then, the
original variables involved in $c$ must take assignment $\sigma$, which
satisfies $c$.

To see that the pathwidth of $\psi$ does not increase much, suppose that we
start with some path decomposition of its primal graph. For each constraint $c$
some bag must contain all variables of $c$. We insert in the decomposition
immediately after this bag a sequence of $2\ell_c$ bags containing all the
vertices of the bag plus a decomposition of the odd cycle formed by the
$2\ell_c+1$ new vertices. This process does not increase pathwidth by more than
$3$, because there is a path decomposition of a cycle of any length that
contains $3$ vertices in each bag, and because for each constraint $c$ we can
select a bag of the original decomposition containing the variables of $c$ to
perform this process.  \end{proof}

\begin{lemma}\label{lem:csp1}

For each $B\ge 3$, if there exist $\eps>0, c>0$ and an algorithm that takes as
input a \textsc{2-CSP} instance $\psi$ on alphabet $B$ and decides if $\psi$ is
satisfiable in time $O((B-\eps)^{\pw}|\psi|^c)$, then there exist $\eps'>0,
c'>0$ such that \tsat\ on $n$ variables can be solved in time
$O((2-\eps')^{\pw}n^{c'})$.

\end{lemma}

\begin{proof}

Fix the $\eps,c$ for which we have a fast \textsc{2-CSP} algorithm. We are
given a 3-CNF formula $\psi$ and a path decompostion of its primal graph of
width $p$, on which we invoke \cref{lem:nice} to obtain a numbering of the bags
$B_1,\ldots, B_t$, and an injective function $b$ mapping each clause to the
index of a bag that contains its variables.

We will first define positive integers $\gamma, \delta, k$ and a real number
$\eps'>0$ such that we have the following:

$$ (B-\eps)^\gamma < (2-\eps')^{\delta} < 2^{\delta+1} < B^{\gamma}$$

and

$$ k \le \frac{p}{\delta}+2 $$

We will construct a \textsc{2-CSP} instance $\phi$ that satisfies the
following:

\begin{enumerate}

\item $\phi$ is satisfiable if and only if $\psi$ is.

\item The new instance can be constructed in polynomial time.

\item We can construct a path decomposition of the new instance of width
$k\gamma+O(1)$.

\end{enumerate}

Before we proceed, let us explain why we obtain the lemma if we achieve the
above. We will decide the \tsat\ instance by constructing the equivalent CSP
instance and feeding it to the supposed algorithm. The running time will be at
most $O((B-\eps)^{k\gamma}|\psi|^{O(1)})$, however, $(B-\eps)^{k\gamma} <
(2-\eps')^{\delta k } < (2-\eps')^{p+O(1)}$, so we get the promised running
time by selecting an appropriate $c'$.

To see that the numbers $\gamma, \delta,k, \eps'$ exist, observe that we can
always select $\gamma$ large enough so that $B^\gamma>4(B-\eps)^\gamma$, in
particular, set $\gamma = \lceil \log_{B/(B-\eps)} 4 \rceil+1$. As a result,
there exists integer $\delta$ such that $(B-\eps)^\gamma <
2^\delta<2^{\delta+1} < B^\gamma$. We now define an $\eps''>0$ so that
$2^{(1-\eps'')\delta} = (B-\eps)^{\gamma/2}2^{\delta/2}$, in particular set
$\eps'' = \frac{1}{2} - \frac{\gamma}{2\delta}\log(B-\eps)$.  Note that we do
indeed have $\eps''>0$ as $\delta>\gamma{\log (B-\eps)} \Leftrightarrow
(B-\eps)^\gamma < 2^\delta$, which we know to hold.  By the definition of
$\eps''$ we have $(B-\eps)^\gamma<2^{(1-\eps'')\delta} <2^{\delta}$. Now set
$\eps'>0$ so that $2-\eps' = 2^{(1-\eps'')}$, that is, $\eps' =
2-2^{1-\eps''}>0$. Note that $\gamma, \delta, \eps'$ are absolute constants
depending only on $B,\eps$.

We now define $k$ to be the smallest positive integer so that $\delta+1 \ge
\lceil\frac{p+1}{k}\rceil$. In particular, we have
$\lceil\frac{p+1}{k-1}\rceil> \delta+1$. We can assume without loss of
generality that $p$ is sufficiently large for $k$ to be well-defined, as this
is the interesting case. The last inequality gives $\frac{p+1}{k-1}+1>\delta+1
\Leftrightarrow k<\frac{p+1}{\delta}+1$ so $k<\frac{p}{\delta}+2$ as desired.

Our plan is now to use \cref{lem:pwcolor} to partition the variables of $\psi$
into $k$ sets $V_1,\ldots, V_k$, such that for all bags $B_j$ of the
decomposition we have $|B_j\cap V_i|\le \lceil\frac{p+1}{k}\rceil \le
\delta+1$.  We will use $\gamma$ CSP variables to represent the value of each
group $B_j\cap V_i$. Since $2^{\delta+1}<B^{\gamma}$, this gives us
sufficiently many CSP assignments to represent all the boolean assignments to
the at most $\delta+1$ variables of the group. We will present the reduction in
two steps, where we first describe a CSP instance of arity $3\gamma$, and the
use \cref{lem:csparity} to reduce this to 2-\textsc{CSP} while keeping the
pathwidth essentially unchanged.  Let us give some more details.

We begin by fixing an injective function $\tau$ from the set $\{False,
True\}^{\delta+1}$ to the set $[B]^{\gamma}$. Such an injective function always
exists, as $2^{\delta+1}<B^{\gamma}$.  Suppose that the bags of the
decomposition are numbered $B_1,\ldots, B_t$, the decomposition is nice, and
the first bag is empty. The boolean variables are numbered $y_1,\ldots, y_n$
according to the order of their introduction in the decomposition.

\begin{enumerate}

\item (Variable groups): For each $i\in[k]$ and $j\in[t]$ we define $\gamma$
CSP variables $x_{i,j,s}$ for $s\in[\gamma]$, meant to represent the assignment
of the at most $\delta+1$ boolean variables of the group $B_j\cap V_i$.

\item (Legal assignment): For each $i\in [k]$ and $j\in[t]$ consider every
assignment $\sigma$ to the $\gamma$ variables $x_{i,j,s}$. If there does not
exist $\sigma'$ such that $\tau(\sigma')=\sigma$, then add a constraint
falsified by this assignment.

\item (Consistency): For each $i\in [k]$ and $j\in[t-1]$, consider every
assignment $\sigma_1$ to the $\gamma$ variables $x_{i,j,s}$ and every
assignment $\sigma_2$ to the $\gamma$ variables $x_{i,j+1,s}$. If there exist
$\sigma_1', \sigma_2'$ such that $\tau(\sigma_1')=\sigma_1$,
$\tau(\sigma_2')=\sigma_2$, but the assignments $\sigma_1', \sigma_2'$ disagree
on a variable of $B_j\cap B_{j+1}\cap V_i$, add a constraint falsified by
$\sigma_1,\sigma_2$.

\item (Satisfaction): For each clause $c$ of the \tsat\ instance, suppose that
$c$ uses variables from three groups $V_{i_1}, V_{i_2}, V_{i_3}$. All three
variables appear in $B_{b(c)}$. For each three assignments $\sigma_1, \sigma_2,
\sigma_3$ to the variables $x_{i_1,b(c),s}$, $x_{i_2,b(c),s}$, $x_{i_3,b(c),s}$
such that there exists $\sigma_{\alpha}'$ for which $\tau(\sigma_{\alpha}') =
\sigma_\alpha$, for $\alpha\in[3]$, and such that $\sigma_1', \sigma_2',
\sigma_3'$ falsifies the clause $c$, add a constraint falsified by $\sigma_1,
\sigma_2, \sigma_3$.

\end{enumerate}

We are now almost done, except that the constraints we have added have arity at
most $3\gamma$ and we would like to have constraints of arity $2$. Before we
deal with this, let us argue that our reduction is correct up to this point.
First, it is not hard to see that $\psi$ is satisfiable if and only if $\phi$
is satisfiable: because of the constraints of step 2 we only consider CSP
assignments which are in one-to-one correspondence with boolean assignments,
and because of step 3 such assignments are consistent across bags. Because of
the constraints of step 3 it is possible to satisfy all constraints if and only
if it is possible to satisfy all clauses.

We argue that the CSP instance we have constructed so far has pathwidth
$k\gamma+O(1)$. To see this, start with the same path decomposition we had for
the \tsat\ instance, but place in each bag $B_j$ the $k\gamma$ variables
$x_{i,j,s}$, for $i\in [k]$ and $s\in [\gamma]$. Note that this covers
constraints of steps 2 and 4. To cover constraints of step 3, consider for each
$j\in[t-1]$ two consecutive bags $B_j, B_{j+1}$ and insert between them a
sequence of $k$ bags, where at each step we first add the group of variables
$x_{i,j+1,s}$ to the bag and at the subsequent bag we remove the group
$x_{i,j,s}$ (that is, we are gradually transforming $B_j$ to $B_{j+1}$). The
largest bag has width at most $(k+1)\gamma = k\gamma+O(1)$.

Finally, to transform this to an equivalent \textsc{CSP} of arity $2$ we invoke
\cref{lem:csparity}. \end{proof}

\begin{lemma}\label{lem:csp2}

For each $B\ge 3$, and for each weight function $w:[B]\to\mathbb{N}$, if there
exist $\eps>0, c>0$ and an algorithm that solves \textsc{MaxW-SAT} on an
instance $\phi$ in time $O((2-\eps)^{\pw}|\phi|^{c})$, then there exist
$\eps'>0, c'>0$ and an algorithm that takes as input an $r$\textsc{-CSP}
instance $\psi$ on alphabet $B$ and a target $w_t$ and decides if $\psi$ has a
satisfying assignment of weight at least $w_t$ in time
$O((B-\eps')^{\pw}|\psi|^{c'})$.

\end{lemma}

\begin{proof}

Fix the $\eps,c$ for which we have a fast algorithm for \textsc{MaxW-SAT}. We
are given as input an $r$\textsc{-CSP} instance and a path decomposition of its
primal graph of width $p=\pw(\psi)$. Apply \cref{lem:nice} to the decomposition
to make it nice and obtain a function $b$ that maps constraints of $\psi$ to
bags that contain all their variables.

We will first define positive integers $\gamma, \delta, k$ and a real number
$\eps'>0$ such that we have the following:

$$(2-\eps)^\delta <(B-\eps')^\gamma< B^{\gamma+1} < 2^\delta$$

and 

$$ k \le \frac{p}{\gamma}+2 $$

We will construct an instance of \textsc{MaxW-SAT} that satisfies the
following:

\begin{enumerate}

\item The new instance is equivalent to the original instance.

\item The size of the new instance is polynomial in the size of the old
instance.

\item We can construct a path decomposition of the new instance of width
$k\delta+O(1)$.

\end{enumerate}

Before we proceed, let us explain why we obtain the lemma if we achieve the
above. We will decide the CSP instance by constructing the equivalent
\textsc{MaxW-SAT} instance and feeding it to the supposed algorithm. The
running time will be at most $O((2-\eps)^{k\delta}|\psi|^{O(1)})$, however,
$(2-\eps)^{k\delta} < (B-\eps')^{\gamma k } < (B-\eps')^{p+O(1)}$, so we get
the promised running time by selecting an appropriate $c'$.

To see that the numbers $\gamma, \delta,k, \eps'$ exist, observe that we can
always select $\delta$ large enough so that $2^\delta>B^2(2-\eps)^\delta$, in
particular, set $\delta = \lceil 2\log_{2/(2-\eps)} B \rceil+1$. As a result,
there exists integer $\gamma$ such that $(2-\eps)^\delta <
B^\gamma<B^{\gamma+1} < 2^\delta$. We now define an $\eps''>0$ so that
$B^{(1-\eps'')\gamma} = B^{\gamma/2}(2-\eps)^{\delta/2}$, in particular set
$\eps'' = \frac{1}{2} - \frac{\delta}{2\gamma}\log_B(2-\eps)$.  Note that we do
indeed have $\eps''>0$ as $\gamma>\delta\frac{\log (2-\eps)}{\log B}
\Leftrightarrow (2-\eps)^\delta < B^\gamma$, which we know to hold.  Now set
$\eps'>0$ so that $B-\eps' = B^{(1-\eps'')}$, that is, $\eps' =
B-B^{1-\eps''}>0$. Note that $\gamma, \delta, \eps'$ are absolute constants
depending only on $B,\eps$.

We now define $k$ to be the smallest positive integer so that $\gamma+1 \ge
\lceil\frac{p+1}{k}\rceil$. In particular, we have
$\lceil\frac{p+1}{k-1}\rceil> \gamma+1$. We can assume without loss of
generality that $p$ is sufficiently large for $k$ to be well-defined, as this
is the interesting case. The last inequality gives $\frac{p+1}{k-1}+1>\gamma+1
\Leftrightarrow k<\frac{p+1}{\gamma}+1$ so $k<\frac{p}{\gamma}+2$ as desired.

Our plan is now to use \cref{lem:pwcolor} to partition the variables of $\psi$
into $k$ sets $V_1,\ldots, V_k$, such that for all bags $B_j$ of the
decomposition we have $|B_j\cap V_i|\le \lceil\frac{p+1}{k}\rceil \le
\gamma+1$. We will use $\delta$ boolean variables to represent the value of
each group $B_j\cap V_i$. Since $2^\delta>B^{\gamma+1}$, this gives us
sufficiently many boolean assignments to represent all the assignments to the
at most $\gamma+1$ CSP variables of the group. Let us give some more details.

We begin by fixing an injective function $\tau$ from the set $[B]^{\gamma+1}$
to the set $\{False, True\}^\delta$. Such an injective function always exists,
as $2^\delta>B^{\gamma+1}$.  Suppose that the bags of the decomposition are
numbered $B_1,\ldots, B_t$, the decomposition is nice, and the first bag is
empty. The CSP variables are numbered $y_1,\ldots, y_n$ according to the order
of their introduction in the decomposition. 

\begin{enumerate}

\item (Variable groups): For each $i\in[k]$ and $j\in[t]$ we define $\delta$
boolean variables $x_{i,j,s}$ for $s\in[\delta]$, meant to represent the
assignment of the at most $\gamma+1$ CSP variables of the group $B_j\cap V_i$.

\item (Legal assignment): For each $i\in [k]$ and $j\in[t]$ consider every
assignment $\sigma$ to the $\delta$ variables $x_{i,j,s}$. If there does not
exist $\sigma'$ such that $\tau(\sigma')=\sigma$, then add a clause falsified
by this assignment.

\item (Consistency): For each $i\in [k]$ and $j\in[t-1]$, consider every
assignment $\sigma_1$ to the $\delta$ variables $x_{i,j,s}$ and every
assignment $\sigma_2$ to the $\delta$ variables $x_{i,j+1,s}$. If there exist
$\sigma_1', \sigma_2'$ such that $\tau(\sigma_1')=\sigma_1$,
$\tau(\sigma_2')=\sigma_2$, but the assignments $\sigma_1', \sigma_2'$ disagree
on a variable of $B_j\cap B_{j+1}\cap V_i$, add a clause falsified by
$\sigma_1,\sigma_2$.

\item (Satisfaction): For each constraint $c$ of the CSP instance, suppose that
$c$ uses $r'\le r$ variables from $r'$ groups $V_{i_1}, V_{i_2}, \ldots,
V_{i_{r'}}$.  All $r'$ variables appear in $B_{b(c)}$. For each $r'$
assignments $\sigma_1, \sigma_2,\ldots, \sigma_{r'}$ to the variables
$x_{i_1,b(c),s}$, $x_{i_2,b(c),s}$,\ldots, $x_{i_{r'},b(c),s}$ such that there
exist $\sigma_{\alpha}'$, with $\tau(\sigma_{\alpha}') = \sigma_\alpha$, for
$\alpha\in[r']$, and such that $\sigma_1', \sigma_2',\ldots, \sigma_{r'}'$
falsifies the constraint $c$, add a clause falsified by $\sigma_1,
\sigma_2,\ldots, \sigma_{r'}$.

\item (Weights): Let $w_{\max}$ be the maximum value of the weight function.
For each $j\in \{2,\ldots,t\}$ such that $B_j$ contains a variable $y_{i'}$,
with $i'\in[n]$, not present in $B_{j-1}$, we add to the instance
$(nt)^5w_{\max}$ new variables labeled $w_{i',s}$, for $s\in [(nt)^5w_{\max}]$.
Suppose $y_{i'}\in V_i$, for $i\in[k]$. Consider every possible assignment
$\sigma$ to the $\delta$ variables $x_{i,j,s}$, such that there exists
$\sigma'$ such that $\tau(\sigma')=\sigma$, where $\sigma'$ is an assignment to
the variables of $V_i\cap B_j$. If $\sigma'$ assigns to $y_{i'}$ a value with
weight $w'\in \{0,\ldots,w_{\max}\}$, construct $w_{\max}(nt)^5$ clauses which
are falsified by $\sigma$ and select a set of $w'(nt)^5$ variables $w_{i',s}$.
Add one of each selected variable to a distinct constructed clause, and add the
remaining non-selected variables \emph{negated} one each to each remaining
clause. Note that this ensures that in a satisfying assignment that agrees with
$\sigma$ exactly $w'(nt)^5$ of the new variables will be set to True.

\end{enumerate}

If the target weight value of the given CSP instance was $w_t$, we set the
target weight of the new instance to be $(nt)^5w_t$. This completes the
construction.

Correctness now is straightforward, because there is a one-to-one mapping
between satisfying assignments of $\phi$ and $\psi$ which also approximately
preserves the weight. In particular, because of clauses of step 2, the
assignment we give to each group $x_{i,j,s}$ must have a corresponding
assignment for the variables of $V_i\cap B_j$. Because of the clauses of step
3, the assignment we extract in this way is consistent. Because of the clauses
of step 4, the extracted assignment must be satisfying for both instances.
Finally, in the last step, for each variable $y$ of the CSP instance we
construct a group of $(tn)^5$ variables whose assignments are forced once we
fix the assignment to the group that contains $y$, and which contribute to the
weight exactly $(tn)^5$ times as much as $y$ contributes to the weight.  Since
the total contribution of all other variables is upper-bounded by their number,
which is at most $tn$, we can infer the maximum weight of any satisfying
assignment of the CSP instance by the maximum weight of the new instance
divided by $(nt)^5$.

Finally, let us argue why we can obtain a path decomposition of the claimed
width. For each bag $B_j$ of the original decomposition we construct a bag
containing all variables $x_{i,j,s}$, for all $i\in[k], s\in[\delta]$. This
covers clauses of steps 2 and 4. To cover clauses of step 3, take two
consecutive bags $B_j, B_{j+1}$ and insert between them a sequence of $k$ bags,
at each step inserting a group $x_{i,j+1,s}$ and in the next step removing the
group $x_{i,j,s}$. The maximum size of any bag up to this point is $(k+1)\delta
= k\delta+O(1)$. To cover the variables and clauses of the last step, find the
bag $B_j$ where $y_{i'}$ was introduced, and insert after it $(tn)^5$ copies,
each containing a distinct variable from the ones constructed in this step.
\end{proof}

\subsection{CSPs with Monotone Multi-assignments}\label{sec:cspmulti}

Even though our main motivation for proving \cref{thm:csp} was that we would
like to use \textsc{CSP} instances as starting points in our reductions of
\cref{sec:fpt}, one could argue that the question of whether \textsc{CSP}s over
alphabets of size $B> 2$ can be solved faster than $B^{\pw}$ is a natural one,
so it is good that this problem turns out to be equivalent to the \ppseth, for
all $B\ge 3$. However, we are now at a point where we will treat a
\textsc{CSP}-related problem which, a priori, seems very artificial. Therefore,
before we proceed, let us try to motivate the problem.

Most reductions which establish that a problem cannot be solved in time
$(c-\eps)^{\pw}n^{O(1)}$ (under SETH) typically start with a problem-specific
gadget: we want to reduce $n$-variable \textsc{CSP} over alphabets of size $c$
to an instance of pathwidth (roughly) $n$, so we design a gadget to represent
each variable of the \textsc{CSP}. This gadget has $c$ distinct locally optimal
configurations, one for each value of the alphabet. A common challenge now is
how to propagate these choices through the construction, that is, how to ensure
that if two distinct copies of the gadget are meant to represent the same
variable, they must select the same local configuration.

Unfortunately, for many problems the known gadgets \emph{do not} allow us to
obtain consistency in such a strong sense, without connecting the gadgets in a
way that blows up the target width.  However, they often do allow us to order
the gadgets in a linear fashion in a way that ensures that, even though it is
possible for the optimal solution to ``cheat'', flipping local configurations
between consecutive copies, the cheating needs to be \emph{monotone}. More
precisely, there is some ordering of the $c$ canonical configurations, so that
we can prove that any locally optimal solution which switches between distinct
configurations, must move towards a configuration that is higher in the
ordering. To understand what we mean by this, recall the (by now standard)
reduction proving that \textsc{Independent Set} cannot be solved in
$(2-\eps)^{\pw}$, given by Lokshtanov, Marx, and Saurabh \cite{LokshtanovMS18}
(see also \cite{CyganFKLMPPS15}). In this reduction, every variable of the
original instance is represented by a long path, where there are two canonical
solutions: either take all odd-numbered vertices (starting from the first one),
or all even-numbered vertices. It is possible for a solution to ``switch'' at
some point, but this can only happen from a configuration of the first type, to
the second type, and not vice-versa, as this would select adjacent vertices.
This monotonicity property allows us to obtain consistency simply by repeating
the construction a sufficiently large number of times. This simple but
ingenious trick has been reused in numerous SETH-based lower bounds
(e.g.~\cite{BorradaileL16,DubloisLP21,DubloisLP22,FockeMR22,HanakaKLOS20,KatsikarelisLP19,KatsikarelisLP22,LampisV23}).

The challenge we face now is that the repetition trick we sketched above cannot
be used to ensure consistency in our setting. We now want to reduce from a
\textsc{CSP} instance where we keep the pathwidth constant, so any section of
our construction can only represent the variables appearing in a section of the
path decomposition. Repeating just this part does not help us in ensuring that
a variable that appears there does not flip values in a constraint that appears
outside this section; however, once we exit the section we can no longer repeat
its construction, since some of the relevant variables have been forgotten.
This conundrum was dealt with by Iwata and Yoshida \cite{IwataY15} via the
following method: in the beginning of each section we add $O(c\log p)$
variables which count, for each of the $c$ configurations, how many of the
gadgets  (representing the $p$ variables of a bag of the \textsc{CSP}) are in
this configuration. At the end of each section we do the same, and we compare
the two counts to make sure they are identical. The monotonicity property is
crucially used here: because if some variable gadget ``cheats'' in the middle
of the section, flipping configurations, it must assume a configuration
representing a higher value, we know that if the count for the lowest value is
correct, then the variables that had this value did not cheat; but then we can
reapply the argument to the second value, and so on, proving that no cheating
can have occurred.

Our reductions of \cref{sec:fpt} will rely on the strategy  introduced by Iwata
and Yoshida \cite{IwataY15} sketched above. However, because it would be very
tedious to reproduce this trick with problem-specific gadgets each time, we
believe it is better to formalize this method through a \textsc{CSP} question.
We therefore arrive at the following definition:

\begin{definition} Suppose we have a \textsc{CSP} instance $\psi$ over alphabet
$[B]$, for $B\ge 2$, with variable set $V$, a path decomposition of its primal
graph $B_1,\ldots, B_t$, and an injective function $b$ mapping each constraint
to the index of a bag that contains all its variables. A
\emph{multi-assignment} is a function $\sigma$ that takes as input a variable
$x\in V$ and an index $j\in [t]$ such that $v\in B_j$ and returns a value in
$[B]$. We will say that:

\begin{enumerate}

\item A multi-assignment $\sigma$ is satisfying for $\psi$ if for each
constraint $c$, the assignment $\sigma_c(x) = \sigma(x,b(c))$, that is, the
restriction of the multi-assignment to $b(c)$, satisfies the constraint.

\item A multi-assignment $\sigma$ is monotone if for all $x\in V$ and $j_1<j_2$
with $x\in B_{j_1}\cap B_{j_2}$ we have $\sigma(x,j_1)\le \sigma(x,j_2)$.

\item A multi-assignment $\sigma$ is consistent for $x\in V$ if for all
$j_1,j_2\in [t]$ such that $x\in B_{j_1}\cap B_{j_2}$ we have
$\sigma(x,j_1)=\sigma(x,j_2)$.

\end{enumerate}

\end{definition} 

Informally, a multi-assignment to a \textsc{CSP} instance supplied with a path
decomposition is an assignment that is allowed to cheat, assigning a distinct
value to each appearance of a variable $x$ in a bag of the decomposition. In
order to call such an assignment satisfying, we only demand that for each
constraint $c$, if we look at the values we gave to the variables in the bag
$B_{b(c)}$, the constraint should be satisfied (even if some variables take
other values in other bags). However, if a multi-assignment is monotone, then
it cannot arbitrarily change values between distinct bags and it must instead
always increase values as we move later in the decomposition.

We now present a technical lemma that captures the technique of \cite{IwataY15}
and shows that looking for monotone multi-assignments is not any easier than
looking for consistent satisfying assignments.

\begin{lemma}\label{lem:weird}

For each $B\ge 2$, there exists a polynomial time algorithm which takes as
input a 4-\textsc{CSP} instance $\psi$ over alphabet $[B]$, as well as a path
decomposition of the primal graph of $\psi$ of width $p$. The algorithm
produces as output a new 4-\textsc{CSP} instance $\psi'$, over the same
alphabet, constructed by adding variables and constraints to $\psi$, a path
decomposition of the primal graph of $\psi'$, and an injective function $b$
mapping constraints of $\psi'$ to bags that contain all their variables, with
the following properties:

\begin{enumerate}

\item Any satisfying assignment for $\psi$ can be extended to a satisfying
assignment for $\psi'$.

\item Any monotone satisfying multi-assignment for $\psi'$ (where monotonicity
is for the constructed decomposition and $b$) which is consistent on the new
variables, is in fact consistent on all variables.  Hence, its restriction to
the variables of $\psi$ satisfies $\psi$.

\item Each bag of the path decomposition of the primal graph of $\psi'$
contains at most $p+1$ variables of $\psi$ and $O(B\log p)$ new variables.

\end{enumerate}

\end{lemma}

\begin{proof} We apply \cref{lem:nice} to the given path decomposition to make
it nice. The bags are numbered $B_1,\ldots, B_t$ and we have an injective
function $b$ that maps constraints to bags containing all their variables.  Let
$\rho = \lceil \log(p+2)\rceil$ be the number of bits needed to represent every
integer in $\{0,\ldots, p+1\}$.

For each $j\in [t]$ we introduce $2B(p+1)\rho$ variables, call them
$c_{i,j,v,s}^\alpha$, for $i\in[p+1]$, $v\in[B]$, $s\in[\rho]$, and
$\alpha\in\{L,R\}$.  The intuitive meaning of these variables is the following:
suppose that $B_j$ contains $|B_j|\le p+1$ variables $x_{j_1}, x_{j_2}, \ldots,
x_{j_{|B_j|}}$.  Then, for each $v\in [B]$, $\alpha\in\{L,R\}$, and $i\in
[p+1]$ the $\rho$ variables $c_{i,j,v,s}^\alpha$ encode a binary number which
tells us how many of the first $i$ variables of $B_j$ have taken value $v$.
Therefore, for $i=p+1$, the variables $c_{p+1,j,v,s}^\alpha$ are meant to tell
us for each $v\in[B]$ how many variables of $B_j$ have taken value $v$. We
introduce two such groups, one for $\alpha=L$ (left) and one for $\alpha=R$
(right), because we intend to measure this twice: once immediately to the left
of bag $B_j$ and once immediately to the right. This is supposed to account for
the possibility that a variable changes value in the middle of the counting
phase.

To simplify presentation, assume that the variables $c_{i,j,v,s}^\alpha$ take
values in $\{0,1\}$ rather than $[B]$. This can be achieved by adding
constraints ensuring that such variables only have two acceptable values in
$[B]$ and then renaming these two values to $\{0,1\}$. Furthermore, to unify
presentation, define for all $j\in[t], v\in[B], s\in[\rho]$,
$\alpha\in\{L,R\}$, a variable $c_{0,j,v,s}^\alpha$ and add constraints to
ensure it takes value $0$.  This is consistent with our intended meaning as
among the first $0$ variables of $B_j$, $0$ have taken value $v$.

Before we go on, let us state a useful (but trivial) technical tool for
handling variables that represent binary numbers.

\begin{claim}\label{claim:count}

For each $r\ge 1$, there exists a 3-CNF formula $\phi_{+1}^r$ with $2r$ input
variables $x_1,\ldots,x_r$ and $y_1,\ldots,y_r$ and at most $2r+1$ additional
(internal) variables, which can be constructed in time polynomial in $r$, and
which satisfies the following: $\phi_{+1}^r(x_1,\ldots,x_r,y_1,\ldots,y_r)$ is
satisfiable, if and only if $x_1,\ldots,x_r$ represent a binary number $X$,
$y_1,\ldots,y_r$ represent a binary number $Y$, and $X= Y+1$.

\end{claim}

\begin{claimproof} Suppose that bits are ordered so that $x_1,y_1$ are the least
significant bits.  The input variables satisfy the property $X=Y+1$ if and only
if there exists $i\in[r]$ which satisfies the following: (i) $x_i=1$ and
$y_i=0$ (ii) for all $i'<i, i'\in[r]$ we have $x_i=0$ and $y_i=1$ (iii) for all
$i'>i$ we have $x_i=y_i$. If we fix an $i$, we can express all these
constraints with $O(r)$ clauses of arity at most $2$, call such a collection of
clauses $\phi_i$.  Define now $r$ variables $z_i$, $i\in[r]$. For each $i$ we
add to all clauses of $\phi_i$ the literal $\neg z_i$. $\phi_{+1}^r$ contains
all clauses of each $\phi_i$ for $i\in[r]$.

Furthermore, we also add $r+1$ variables $w_i$ for $i\in\{0,\ldots,r\}$. We add
the clauses $(w_{r}), (\neg w_0)$ and for each $i\in[r]$ the clause $(\neg
w_{i} \lor w_{i-1} \lor z_i)$. Adding these clauses to $\phi_{+1}^r$ completes
the construction. For one direction, if $X=Y+1$, there exists $i\in[r]$ as
described in the previous paragraph. We satisfy the formula by setting $z_i$ to
True, all other $z_{i'}$ variables to False, and $w_{i'}$ to True whenever
$i'\ge i$. For the converse direction, in any satisfying assignment we must
have $z_i$ set to True for some $i\in[r]$. Then, $\phi_i$ is satisfied, which
implies that $X=Y+1$.  \end{claimproof}

We now add the constraints necessary to ensure that the new variables have
values consistent with their intended role:

\begin{enumerate}

\item For each $j\in[t]$, $v\in[B]$, $i\in[p+1]$, $\alpha\in\{L,R\}$,
construct a new variable $q_{i,j,v,\alpha}$, which takes value $1$ if the
$i$-th variable of $B_j$ takes value $v$, and takes value $0$ otherwise.  In
particular, if $B_j$ has fewer than $i$ variables, add a constraint forcing
$q_{i,j,v,\alpha}$ to value $0$.

\item For each $j\in[t]$, $v\in[B]$, $i\in[p+1]$, $\alpha\in\{L,R\}$,  add
constraints ensuring that if $q_{j,i,v,\alpha}=0$, then
$c_{i,j,v,s}^\alpha=c_{i-1,j,v,s}^\alpha$ for all $s\in[\rho]$.

\item For each $j\in[t]$, $v\in[B]$, $i\in[p+1]$, $\alpha\in\{L,R\}$, if $B_j$
has at least $i$ variables, add constraints ensuring that if
$q_{j,i,v,\alpha}=1$, then $c_{i,j,v,s}^\alpha$ encodes an integer that is one
more than the integer encoded by $c_{i-1,j,v,s}^\alpha$. This is done by
constructing the formula $\phi_{+1}^\rho$ of \cref{claim:count}, using input
variables $c_{i,j,v,s}^\alpha, c_{i-1,j,v,s}^\alpha$ for $s\in[\rho]$, and
fresh internal variables, and adding the literal $\neg q_{j,i,v,\alpha}$ to all
clauses.

\item For each $j\in[t]$, $v\in[B]$, and $s\in[\rho]$ we add a constraint
ensuring that $c_{p+1,j,v,s}^L = c_{p+1,j,v,s}^R$.

\item For each $j\in[t-1]$ such that $B_j=B_{j+1}$, for each $v\in[B]$ we add
constraints ensuring that for all $s\in[\rho]$ we have
$c_{p+1,j,v,s}^R=c_{p+1,j+1,c,s}^L$.

\item For each $j\in[t-1]$ such that $B_j\neq B_{j+1}$, let $x$ be the unique
element that appears in exactly one of $B_j, B_{j+1}$. Let $j',j''$ be such
that $j',j''\in\{j,j+1\}$, and $x\in B_{j'}$ but $x\not\in B_{j''}$. Let
$\alpha=R$ if $j'=j$ and $\alpha=L$ if $j'=j+1$, while
$\alpha'\in\{L,R\}\setminus\{\alpha\}$.  Suppose that $x$ is the $i$-th
variable of $B_{j'}$, for $i\in[p+1]$.  We add constraints that ensure that for
all $v\in[B]$, if $q_{i,j',v,\alpha}=0$, then
$c_{p+1,j,v,s}^{R}=c_{p+1,j+1,v,s}^{L}$, for all $s\in[\rho]$.  Furthermore,
for all $v\in[B]$, we add constraints ensuring that if $q_{i,j',v,\alpha}=1$,
then $c_{p+1,j',v,\alpha}$ encodes an integer one higher than the one encoded
by $c_{p+1,j'',v,\alpha'}$. Again this can be done by including a copy of the
formula $\phi_{+1}^\rho$ of \cref{claim:count}, using fresh internal variables,
where we add in all clauses the literal $\neg q_{i,j',v,\alpha}$.

\end{enumerate}

This completes the construction and it is clear that we can perform this in
polynomial time. In particular, we have added a polynomial number of new
constraints, and the arity of each constraint is at most $4$ (this is the case
for clauses of $\phi_{+1}^\rho$ to which we add a literal).  $\psi'$ is
obtained by adding variables and constraints to $\psi$, as promised.  Also, by
construction we have assumed that new variables take only two possible values. 

Let us also describe how we obtain a path decomposition of $\psi'$ by editing
the decomposition of $\psi$ and how to update the function $b$ so that it
remains injective and covers new constraints.  For each $j\in[t]$ we insert
immediately before $B_j$ a sequence of $p+1$ bags, call them $B^{L,1}_j,
B^{L,2}_j, \ldots, B^{L,p+1}_j$, and immediately after $B_j$ a sequence of
$p+1$ bags, call them $B^{R,1}_j, B^{R,2}_j, \ldots, B^{R,p+1}_j$.  We place
into each of these bags all the elements of $B_j$.  For $i\in[p+1]$,
$\alpha\in\{L,R\}$, bag $B^{\alpha,i}_j$ also contains all $3B\rho$ variables
$c_{i,j,v,s}^\alpha, c_{i-1,j,v,s}^\alpha, c_{p+1,j,v,s}^\alpha$, for $v\in[B],
s\in[\rho]$, all variables $q_{i,j,v,\alpha}$, for $v\in [B]$, and the internal
variables of the copy of $\phi_{+1}^\rho$ we used in step 3.  We also add the
$2B\rho$ variables $c_{p+1,j,v,s}^L, c_{p+1,j,v,s}^R$ to $B_j$. Furthermore, if
$B_j$ is an Introduce bag, and the element $x\in B_{j}\setminus B_{j-1}$ is the
$i$-th variable of $B_j$, add the $B$ variables $q_{i,j,v,L}$ to all bags in
$B^{L,1},\ldots,B^{L,p+1},B_j$. Symmetrically, if $B_j$ is a Forget bag, and
the element $x\in B_{j}\setminus B_{j+1}$ is the $i$-th variable of $B_j$, add
the $B$ variables $q_{i,j,v,R}$ to all bags in $B_j,B^{R,1},\ldots,B^{R,p+1}$.
We have now covered all constraints of steps 1,2,3, and 4, and each bag
contains the elements of an old bag plus at most $O(B\rho)$ new variables.

In order to cover the remaining constraints, for each $j\in[t-1]$ we insert
between $B_j^{R,p+1}$ and $B_{j+1}^{L,1}$ a bag containing $ B_j^{R,p+1}\cup
B_{j+1}^{L,1}$ plus the internal variables of the copy of $\phi_{+1}^\rho$ we
may have used in step 6. This covers the constraints of steps 5 and 6, and the
total number of new variables in any bag is $O(B\rho)$.  One can easily check
that the properties of path decompositions are satisfied. To update the
function $b$ to cover the new constraints, for each $j\in[t]$, we map the
constraints of step 1 that involve the $i$-th variable of $B_j$ to bags
$B_j^{L,i}$ and $B_j^{R,i}$. The function $b$ is now injective for the
constraints involving any of the old variables, as all other new constraints
involve only the new variables. For each remaining new constraint, find a bag
$B$ that contains all variables of $c$, insert after $B$ a copy of it, and set
$b$ to map $B$ to this copy. This ensures that $b$ is injective. In the
remainder, when we talk about the interval around a bag $B_j$ of the original
decomposition, we mean the bags $B_{j}^{L,1},\ldots, B_j^{L,p+1}, B_j,
B_j^{R,1},\ldots, B_j^{R,p+1}$, as well as the copies of these bags added to
make $b$ injective.

We now need to argue for correctness.  If $\psi$ is satisfiable, use the same
assignment for the variables of $\psi'$.  We can now assign to the
$c_{i,j,v,s}^\alpha$ variables values that correspond to their intended
meaning, that is, for fixed $i,j,v,\alpha$, the $\rho$ variables
$c_{i,j,v,s}^\alpha$ represent a binary number that tells us how many of the
first $i$ variables of $B_j$ have value $v$, and for each $i,j,v,\alpha$, the
variable $q_{i,j,v,\alpha}$ has value $1$ if the $i$-th variable of $B_j$ has
value $v$.  It is not hard to see that this satisfies $\psi'$, or more
precisely this can be extended to satisfy $\psi'$, including the clauses of all
copies of $\phi_{+1}^\rho$ we have used.

For the converse direction, suppose we have a satisfying multi-assignment
$\sigma$ which is, however, monotone for the new path decomposition and $b$ and
which is consistent for the new variables $c_{i,j,v,s}^\alpha$,
$q_{i,j,v,\alpha}$ and the internal variables of copies of $\phi_{+1}^\rho$.
We want to show that the assignment must also be consistent on the old
variables.  In order to achieve this, we first consider for each $j\in[t]$ the
interval of bags  around $B_j$, as we defined it above. We first want to
establish that every old variable $x$ is consistent in any such interval. Note
that if $x$ appears in one bag of such an interval, it must appear in $B_j$,
therefore it must appear in the whole interval. Suppose that $x$ is the $i$-th
element of $B_j$.  We observe that there are only three bags where the
assignment to $x$ possibly affects satisfaction: $B_j$ (where $b$ may map an
old constraint involving $x$), $B_j^{L,i}$, and $B_j^{R,i}$ (where the new
constraints of step 1 are mapped). Suppose that $\sigma$ assigns to $x$ values
$v_1, v_2, v_3\in [B]$ in bags $B_j^{L,i}, B_j, B_j^{R,i}$ respectively, with
$v_1\le v_2\le v_3$ because of monotonicity. If $v_1=v_3$, then we can actually
assign value $v_1$ to $x$ in the whole interval, without affecting satisfaction
of any constraint or the monotonicity of $\sigma$. If $v_1<v_3$, then $x$ is
problematic. Among all problematic variables, let $x$ be one with minimum
$v_1$. We will show that if such an $x$ exists, then some constraint is
violated. By the selection of $x$, any other variable of $B_j$ which is
assigned a value $v'<v_1$ anywhere in the interval around $B_j$ is actually
consistent in the whole interval, retaining a value strictly smaller than
$v_1$. Consider now the variables $c_{i,j,v_1,s}^L, c_{i,j,v_1,s}^R$ as two
counters, which are increased each time we encounter a variable with value
$v_1$.  Variable $x$ contributes to increasing this counter on the left,
because its value in $B^{L,i}_j$ is $v_1$, but it does not contribute on the
right.  Other variables of $B_j$ which fail to contribute on the left do so
either because their value at the moment they are considered is strictly
smaller than $v_1$, but then such variables do not contribute on the right
either, since they are consistently assigned; or because their value is
strictly larger than $v_1$, in which case they fail to contribute on the right
due to monotonicity.  We conclude that the number encoded by
$c_{p+1,j,v_1,s}^R$ must be strictly smaller than $c_{p+1,j,v_1,s}^L$,
violating the constraints of step 4.

Suppose then that the multi-assignment is locally consistent for old variables,
that is, all old variables in the interval around $B_j$ receive assignments
consistent for that interval, for all $j\in [t]$. We now want to argue that
constraints of steps 5 and 6 ensure that $\sigma$ must be consistent across
consecutive intervals. 

Fix a $j\in[t-1]$ and consider a variable $x\in B_j\cap B_{j+1}$ which is the
$i_1$-th variables of $B_j$ and the $i_2$-th variable of $B_{j+1}$.  In the
interval of bags strictly between $B_j$ and $B_{j+1}$, the value assigned to
$x$ affects satisfaction only in $B_j^{R,i_1}$ and $B_{j+1}^{L,i_2}$. If
$\sigma$ assigns to $x$ values $v_1<v_2$ in these two bags, we say that $x$ is
problematic. Select $x$ among the problematic variables so as to minimize
$v_1$. As a result, any old variable which is assigned a value $v'<v_1$
anywhere receives a consistent assignment in all intervals where it appears.
Since old variables are locally consistent around $B_j$ and $B_{j+1}$, the
variables $c_{p+1,j,v_1,s}^R$ and $c_{p+1,j+1,v_1,s}^L$ encode how many
variables of $B_j$ and $B_{j+1}$ respectively received value $v_1$ in the
interval around $B_j$, $B_{j+1}$ respectively. 

Recall that there is a unique vertex $y\in (B_{j+1}\setminus B_j) \cup
(B_j\setminus B_{j+1})$. Let $j',j''\in\{j,j+1\}$ such that $y\in B_{j'}$ and
$y\not\in B_{j''}$, and suppose $y$ is the $i$-th variable of $B_{j'}$. Let
$\alpha=R$ if $j'=j$ and $\alpha=L$ if $j'=j+1$, while
$\alpha'\in\{L,R\}\setminus\{\alpha\}$. Recall that we have inserted a unique
bag between $B_j^{R,p+1}$ and $B_{j+1}^{L,1}$ that contains
$q_{i,j',v,\alpha}$, for all $v\in [B]$. These variables encode whether $y$
received value $v$.

We now distinguish between several cases: 

\begin{itemize}

\item $q_{i,j',v_1,\alpha}=1$, that is, $y$ is assigned the same value as $x$,
according to $q_{i,j',v_1,\alpha}$. If $y\in B_j$, then $y$ contributes to the
counter $c_{p+1,j,v_1,s}^R$, to which $x$ also contributes. This counter is
supposed to decrease by one to satisfy the constraint of the last step, but it
decreases by two (due to forgetting $y$ and changing the value of $x$).
Similarly, if $y\in B_{j+1}$, the constraint of the last step stipulates that
the counter $c_{p+1,j+1,v_1,s}^L$ must be strictly higher than
$c_{p+1,j,v_1,s}^R$; however, the value of $c_{p+1,j+1,v_1,s}^L$ cannot achieve
this value, because $x$ no longer contributes to it, and other variables that
do not contribute to $c_{p+1,j,v_1,s}^R$, cannot contribute to
$c_{p+1,j+1,v_1,s}^L$ either. This is due to the selection of $v_1$ (variables
with lower values are consistent) or monotonicity (variables with higher
values, retain a higher value).

\item $q_{i,j',v_1,\alpha}=0$, that is, $y$ is assigned a different value from
$x$, according to $q_{i,j',v_1,\alpha}$. In this case the constraint of the
last step states that $c_{p+1,j,v_1,s}^R=c_{p+1,j+1,v_1,s}^L$ for all
$s\in[\rho]$. However, this cannot be achieved, as $x$ contributes to the
counters for $B_j, v_1$, but not to the counters for $B_{j+1}, v_1$, and other
variables which do not contribute to the former, also do not contribute to the
latter (again, by the selection of $v_1$ or monotonicity).

\end{itemize}

We have now established that a monotone multi-assignment must be consistent:
for any old variable $x$ and $j,j'$ such that $x\in B_{j}\cap B_{j'}$, $\sigma$
assigns the same value to $x$ in both $B_j$ and $B_{j'}$, and the intervals of
bags around them.  \end{proof}

By \cref{lem:weird} we obtain the following:

\begin{corollary}\label{cor:weird}

For all $\eps>0, B\ge 2$ we have the following. Suppose there is an algorithm
that takes as input a 4-\textsc{CSP} instance $\psi$, a partition of its
variables into two sets $V_1, V_2$, a path decomposition of its primal graph of
width $p$ where each bag contains at most $O(B\log p)$ variables of $V_2$, and
an injective function $b$ mapping each constraint to a bag that contains its
variables.  The algorithm decides if there exists a monotone satisfying
multi-assignment $\sigma$, which is consistent for the variables of $V_2$. If
the supposed algorithm runs in time $O((B-\eps)^p |\psi|^{O(1)})$, then the
\ppseth\ is false.

\end{corollary}

\begin{proof}

We begin with an instance $\phi$ of 4-\textsc{CSP} over alphabet $B$ (if $B=2$
consider simply an instance of \tsat) and a path decomposition of its primal
graph. If for some $\eps'$ we could decide $\phi$ in time
$O((B-\eps')^{\pw(\phi)}|\phi|^{O(1)})$, then the \ppseth\ would be falsified,
by \cref{thm:csp}. We first execute the algorithm of \cref{lem:weird} to obtain
in polynomial time an instance $\phi'$ whose variables are partitioned into
$V_1$ (variables of $\phi$) and $V_2$ (new variables). Deciding if a satisfying
multi-assignment with the desired conditions exists for $\phi'$ is equivalent
to deciding $\phi$ and this can be done by the supposed algorithm. The width of
the decomposition of the new instance is at most $p=\pw(\phi)+O(B\log
(\pw(\phi)))$. We therefore have $(B-\eps)^p<(B-\eps)^{\pw(\phi)}\cdot
B^{O(B\log(\pw(\phi)))} = (B-\eps)^{\pw(\phi)} |\phi|^{O(1)}$ obtaining an
upper bound on the running time of the algorithm on $\phi'$.  Setting
$\eps'=\eps$ gives the statement.  \end{proof}

\section{Evidence}\label{sec:evidence}

Given the results of \cref{thm:robust}, it is clear that the \ppseth\ is
implied by (and hence is at least as plausible as) the SETH: since for all CNF
formulas $\phi$ on $n$ variables, $\pw(\phi)\le n$, falsifying the \ppseth\
would imply the existence of an algorithm solving \textsc{SAT} (and indeed,
\textsc{Max-}$(t,w)$\textsc{-SAT}) in time $(2-\eps)^n$. It is, however, less
clear how much more believable one should consider the \ppseth. In this section
we give evidence that falsifying the \ppseth\ would require significantly more
effort than falsifying the SETH, because the \ppseth\ is implied by two other
well-established complexity assumptions, namely the $k$-Orthogonal Vectors
Assumption (\kova) and the Set Cover Conjecture (\scc).

Our results are summarized below in \cref{thm:evidence} and also in
\cref{fig:evidence}.

\begin{figure}
\centering
\input{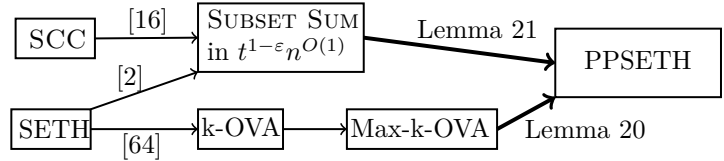} \caption{Summary of connections between other conjectures and
the \ppseth}\label{fig:evidence} \end{figure}

\begin{theorem}\label{thm:evidence} 

If the \ppseth\ is false, then both of the following hold:

\begin{enumerate}

\item There exist $k_0>1, \eps>0$, such that there exists an algorithm which
for all $k>k_0$ solves $k$-\textsc{Most-Orthogonal Vectors} in time
$O(n^{(1-\eps)k}d^{O(1)})$. In particular, this implies that the maximization
versions of the \seth\ and the \kova\ are false.

\item There exists $\eps>0$ and an algorithm that solves \textsc{Subset Sum} on
$n$ elements with target value $t$ in time $O(t^{1-\eps}n^{O(1)})$. In
particular, this implies that the \scc\ is false.

\end{enumerate}

\end{theorem}

Before we proceed, let us provide some further context. Recall that the \kova\
states that for all $k\ge2, \eps>0$, there is no algorithm which takes as input
$k$ arrays, each containing $n$ $d$-dimensional boolean vectors and decides if
there exists a collection of $k$ vectors, one from each array, with inner
product equal to $0$. The maximization version of this problem, which we denote
as $k$-\textsc{Most-Orthogonal Vectors}, asks whether there exists a collection
of $k$ vectors such that the weight of their inner product contains at least
$t$ zeroes. It is by now a classical result in fine-grained complexity that the
\kova\ is implied by the SETH, going back to Williams \cite{Williams05}. The
\kova, and especially its restriction to $k=2$, is among the most prominent
complexity assumptions in fine-grained complexity.  We refer the reader to the
survey of Vassilevska Williams for more information \cite{williams2018some}.
Note that, since maximizing the number of zeroes in the inner product is harder
than deciding if the inner product can be made equal to $0$, the assumption
that the optimization problem is hard is implied by the \kova. However, it is
not clear if this is actually a harder problem, and in fact the two problems
may have the same complexity for all $k\ge2$, as indicated by the work of Chen
and Williams for $k=2$ \cite{ChenW19}. Gao, Impagliazzo, Kolokolova, and
Williams have shown that the restriction of the \kova\ to $k=2$ is equivalent
to the statement that there exists a first-order formula with $k$ quantifiers
that cannot be model checked in time $O(m^{k-1-\eps})$.

Our results show that falsifying the \ppseth\ is harder than falsifying (even
the maximization version of) the \kova. This is established in \cref{lem:kova}.
Hence, the \ppseth\ should be thought of as more plausible than not just the
SETH, but also than some of its most prominent weakenings.

The Set Cover Conjecture (\scc) states that there is no algorithm which can
solve \textsc{Set Cover} on instances with $n$ elements and $m$ sets in time
$(2-\eps)^nm^{O(1)}$. This conjecture was introduced in \cite{CyganDLMNOPSW16},
where it was shown that the symmetric version of the conjecture, where we
replace $m$ with $n$ and vice-versa, is in fact equivalent to the SETH. Despite
much effort, finding a connection between the \scc\ and the SETH (in either
direction) was left as an open problem in \cite{CyganDLMNOPSW16} and remains a
major question to this day.   One part of the work of \cite{CyganDLMNOPSW16}
was devoted to exploring the consequences of the SCC and, among other results,
it was shown that if the SCC is true, then no algorithm can solve
\textsc{Subset Sum} in time $t^{1-\eps}n^{O(1)}$, where $n$ is the number of
given integers and $t$ the target sum. This important problem was later
revisited by Abboud, Bringmann, Hermelin, and Shabtay \cite{AbboudBHS22}, who
showed that such a running time can also be ruled out if we assume the SETH
instead of the SCC.

\scc\ is currently viewed as a plausible complexity assumption, somewhat
orthogonal to the SETH -- indeed this is what motivated the authors of
\cite{AbboudBHS22} to reprove a bound known under one conjecture, using the
other conjecture as their basis. Here we show that the \ppseth\ is also implied
by the \scc, giving some further, seemingly independent evidence that the
problem we consider is hard. This is established in \cref{lem:subsetsum}, where
we in fact show that falsifying the \ppseth\ would provide a \textsc{Subset
Sum} algorithm fast enough to falsify both the \scc\ and the SETH.

We now proceed to the two lemmas which together imply \cref{thm:evidence}.

\subsection{Orthogonal Vectors}

\begin{lemma}\label{lem:kova} If there exist $\eps>0, c>0$ such that there exists an algorithm
that solves \textsc{MaxW-SAT} in time $((2-\eps)^{\pw}|\psi|^c)$, where $\psi$
is the input formula, then there exist $\eps'>0, k_0>1, c'>0$ and an algorithm
which for all $k>k_0$ solves $k$-\textsc{Most-Orthogonal Vectors} in time
$O(n^{(1-\eps')k}d^{c'})$.  \end{lemma}

\begin{proof}

Fix the $\eps,c$ for which we have a fast algorithm solving \textsc{MaxW-SAT}
as described. We are given $k$ arrays $A_1,\ldots, A_k$, each containing $n$
boolean vectors of dimension $d$. We are also given a target weight $d_t$ and
are asked to select for each $i\in[k]$ an element $v_i\in A_i$ such that the
inner product $\Pi_{i=1}^k v_i$ contains at least $d_t$ zeroes (note that
maximizing zeroes is of course equivalent to minimizing ones).

We will construct a CNF formula $\psi$ and a weight target $w_t$ with the
following properties:

\begin{enumerate}

\item $\psi$ has a satisfying assignment setting at least $w_t$ variables to
True, if there exist $k$ vectors, one from each array, whose inner product has
at least $d_t$ zeroes.

\item $\psi$ can be produced in time $O(d^2n^2)$ and we can also produce a path
decomposition of its primal graph of width at most $k\log n + O(1)$.

\end{enumerate}

Before we proceed, let us argue why the above imply the lemma. We will decide
the $k$-\textsc{Most-Orthogonal Vectors} instance by running the above
reduction and feeding the resulting instance to the supposed algorithm for
\textsc{MaxW-SAT}. The running time is at most $O((2-\eps)^{k\log n+O(1)}
(d^2n^2)^c)$. Set $\delta>0$ so that  $2-\eps = 2^{1-\delta}$, in particular,
$\delta = 1-\log(2-\eps)$. The running time is then at most
$O(n^{(1-\delta)k+2c}d^{2c})$. Set $k_0 = 4c/\delta$ and $c'=2c$. If $k>k_0$
the running time is at most $O(n^{(1-\delta/2)k}d^{c'})$, so if we set
$\eps'=\delta/2$ we are done. Observe that $\eps', k_0, \delta$ are absolute
constants that depend only on $\eps, c$.

Assume that $n$ is a power of $2$ (we can add dummy all-1 vectors in each array
to achieve this) and that the elements of each array $A_i$ are $v_{i,0},
v_{i,1},\ldots, v_{i,n-1}$.  We construct a CNF formula starting with $k\log n$
variables $x_{i,s}$, for $s\in[\log n]$. The intended meaning is that for a
fixed $i\in [k]$, the value of the variables of the group $x_{i,s}$ encodes the
index of the selected vector of $A_i$.

We now add the following:

\begin{enumerate}

\item For each $j\in[d]$ we add a variable $z_j$, intended to represent whether
the $j$-th bit of the inner product is zero. We also add $n^2$ variables
$z_{j,s}$, $s\in[dn^2]$. For each $s\in[dn^2-1]$ we add the clause $(\neg
z_{j,s+1}\lor z_{j,s})$. We also add the clause $(\neg z_{j,1}\lor z_j)$. We
have constructed $O(n^2d^2)$ clauses of size $2$.

\item For each $i\in[k]$ and $j\in [d]$ we add a variable $w_{i,j}$. Its
intended meaning is to represent if the $j$-th bit of the $i$-th selected
vector is zero. For each assignment $\sigma$ to the variables of $x_{i,s}$,
with $s\in[\log n]$, read the assignment as a binary number $s_1$. If vector
$v_{i,s_1}$ has a $1$ in position $j$, we add a clause that is falsified by the
assignment $\sigma$ and add to it the literal $\neg w_{i,j}$. We have
constructed $O(nd)$ clauses of size $O(\log n)$.

\item For each $j\in [d]$ add the clause $(\bigvee_{i\in [k]} w_{i,j} \lor \neg
z_j)$. These are $d$ clauses of arity $k$.

\end{enumerate}

We set the target weight to be $d_tdn^2$. This completes the construction.
Correctness is straightforward by observing that because of the clauses of the
first step, setting any of the $z_{j,s}$ to True implies that $z_{j}$ is True,
and in such a case it is optimal to set all the $z_{j,s}$ to True. Because the
number of variables outside of the variables $z_{j,s}$ is at most $O(k\log n +
d +kd)$, the weight of the satisfying assignment is dominated by the number of
$z_{j,s}$ variables that we set to True. We now see that for each assignment to
the $k\log n$ variables $x_{i,s}$, if we read this assignment in the intended
way, $w_{i,j}$ can be set to True if and only if the $i$-th selected vector has
a $0$ in position $j$ and it is optimal to do so, as it allows us to set $z_j$
to True.

Let us also argue why we can construct a path decomposition of the desired
width. Start width $d$ bags $B_1,\ldots, B_d$ and place all $x_{i,s}$ for
$i\in[k], s\in[\log n]$ in all bags. For each $i\in[k], j\in [d]$ we place in
bag $B_j$ the variables $w_{i,j}$ and $z_j$. Finally, for each $j$, after $B_j$
we insert a sequence of $dn^2$ bags $B_j^1,\ldots, B_j^{dn^2}$, each of which
contains all the elements of $B_j$. For each $s\in[dn^2-1]$ we add to $B_j^s$
the variables $z_{j,s},z_{j,s+1}$.  \end{proof}

\subsection{Subset Sum and Set Cover Conjecture}

\begin{lemma}\label{lem:subsetsum} If there exist $\eps>0, c>0$ such that \textsc{SAT} can be solved
in time $O((2-\eps)^{\pw}|\psi|^c)$, where $|\psi|$ is the input formula, then
there exists $\eps'>0$ and an algorithm that solves \textsc{Subset Sum} on $n$
elements with target value $t$ in time $O(t^{1-\eps'}n^{O(1)})$.  \end{lemma}

\begin{proof}

Fix the $\eps,c$ for which a fast satisfiability algorithm exists.  We are
given a collection of $n$ positive integers $v_1,\ldots, v_n$ and a target
value $t$ and are asked if there exists a subset of the given integers whose
sum is exactly $t$. We will construct a CNF formula $\psi$ that satisfies the
following:

\begin{enumerate}

\item $\psi$ is satisfiable if and only if a subset of the given integers has
sum exactly $t$.

\item $\psi$ can be constructed in time $O(n\log t)$.

\item We can construct in the same amount of time a path decomposition of the
primal graph of $\psi$ of width $\log t + O(1)$.

\end{enumerate}

Before we proceed, let us explain why the above imply the lemma. We will decide
the given \textsc{Subset Sum} instance by producing $\psi$ and feeding it into
the supposed \textsc{SAT} algorithm. The running time of this procedure is at
most $O((2-\eps)^{\log t} (n\log t)^c) = O(t^{\log(2-\eps)}(\log^{c}t)n^c)$.
Set $\eps' = (1-\log(2-\eps))/2>0$ and for sufficiently large $t$ we have
$t^{\eps'}>\log^{c}t$, so we get the promised running time as
$t^{1-2\eps'}\log^{c}t<t^{1-\eps'}$.

Let us now explain how to construct $\psi$. Let $b=\lceil\log t\rceil$ be the
number of bits required to represent $t$. We construct the following variables:

\begin{enumerate}

\item $n$ variables $s_1,\ldots, s_n$, such that $s_i$ for $i\in[n]$ is meant
to indicate whether the solution selects integer $v_i$.

\item $nb$ variables $x_{i,s}$, for $i\in[n]$, $s\in[b]$. For a given $i\in[n]$
the variables $x_{i,s}$ are supposed to encode the sum of the elements that
have been selected among the integers $v_1,\ldots, v_i$. To unify the
construction of the rest of the formula, we also define $b$ variables
$x_{0,s}$, for $s\in[b]$ and add the clauses $(\neg x_{0,s})$ for all
$s\in[b]$. 

\item $nb$ variables $c_{i,s}$, for $i\in[n]$, $s\in[b]$. For a given
$i\in[n]$, the $c_{i,s}$ variable is supposed to calculate the $i$-th carry bit
we obtain when adding $v_i$ to the current sum, which is encoded in
$x_{i-1,s}$. Again, to unify presentation we also construct for each $i\in[n]$
a variable $c_{i,0}$ and add the clause $(\neg c_{i,0})$.

\end{enumerate}

Intuitively, the variables above are meant to capture the execution of a
non-deterministic algorithm which goes through the $n$ elements, at each step
guesses if the element is selected (this guess is encoded in $s_i$) and keeps
track of the current sum using $b$ bits. If the current element is selected the
algorithm needs to update the current sum by adding to it the current element.

In the following, when we view an integer $v_i$ as a binary number, we suppose
that the bits are numbered $\{1,\ldots,b\}$, where bit $1$ is the least
significant bit (that is, right to left). We now add the following clauses:

\begin{enumerate}

\item For all $i\in[n]$ and $s\in[b]$ we add the clauses $(s_i\lor \neg
x_{i-1,s}\lor x_{i,s})$ and $(s_i\lor \neg x_{i,s} \lor x_{i-1,s})$.
Intuitively, these ensure that if $s_i$ is false, then $x_{i,s-1}$ and
$x_{i,s}$ encode the same integer.

\item For all $i\in[n]$ we do the following. For all $s\in[b]$, if the $s$-th
bit of $v_i$ is $0$, then we add clauses encoding the constraint $x_{i,s} =
x_{i-1,s} \oplus c_{i,s-1}$, that is, we add the four clauses corresponding to
assignments that violate this constraint. Furthermore, we add the clause $(\neg
c_{i,s})$. In all the clauses constructed for some $i\in[n]$ in this step we
add the literal $\neg s_i$.

\item For all $i\in[n]$ we do the following. For all $s\in[b]$, if the $s$-th
bit of $v_i$ is $0$, then we add clauses encoding the constraint $x_{i,s} =
x_{i-1,s} \oplus c_{i,s-1}\oplus 1$, that is, we add the four clauses
corresponding to assignments that violate this constraint. We also add clauses
to encode the constraint $c_{i,s}\leftrightarrow (x_{i-1,s}\lor c_{i,s-1})$. In
all the clauses constructed for some $i\in[n]$ in this step we add the literal
$\neg s_i$.

\item For all $i\in[n]$ we add the clause $(\neg c_{i,b})$. Informally, this
clause ensures that we never have an overflow.

\item For $s\in [b]$, if the $s$-th bit of $t$ is $0$ we add the clause $(\neg
x_{n,s})$ otherwise we add the clause $(x_{n,s})$.

\end{enumerate}

We have constructed $O(nb)=O(n\log t)$ clauses of constant arity, and the
construction can be performed in time linear in the output size. Let us first
argue for correctness.  Suppose that a set of integers of sum $t$ exists.  We
set the $s_i$ variables to encode the integers included in this set, and the
variables $x_{i,s}$ to encode the sum of the integers selected among the first
$i$ integers of the input.  Clauses of the first step are already satisfied by
this assignment, as are the clauses of the last step.  For a value $i\in[n]$
for which $v_i$ was not selected, we set $c_{i,s}$ to False for all $s\in[b]$
and all clauses added in steps 2,3,4 for this $i$ are satisfied. For the
remaining values, we set $c_{i,s}$ to $1$ if the result of adding the $s$ least
significant bits of $v_i$ to the $s$ least significant bits of the sum we have
formed using the $i-1$ first integers is at least $2^{s}$, that is, if we have
a carry at bit $s$ when adding $v_i$ to the current sum.  This satisfies the
cluses of steps 2 and 3, and clauses of step 4 are satisfied if the sum of the
selected integers is indeed $t$, in which case we should never have an
overflow. For the converse direction, using the same interpretation of the
variables, we claim that any satisfying assignment that sets the variables
$x_{i,s}$ to encode some binary value $t_i$ implies that there is a subset of
$v_1,\ldots,v_i$ with sum exactly $t_i$. This can be shown via induction and
implies that, since by step 5 the variables $x_{n,s}$ encode $t$, we have a
solution for \textsc{Subset Sum}.  

Let us now argue about the pathwidth. Variables $x_{0,s}$ and $c_{i,0}$ can be
omitted: we have added these variables to simplify the presentation of the
construction, but since their values are forced, we can assign to them the
corresponding values and simplify the formula. Note that this simplification
procedure will also remove variables $c_{1,s}$, which is natural since when we
(possibly) add $v_1$ to the initial sum (which is $0$), all carries will be
$0$.  Construct $n$ bags $B_1,\ldots, B_n$ and place into each $B_i$ the $b$
variables $x_{i,s}$, as well as $s_i$.  For each $i\in[n-1]$ we insert a
sequence of $t-1$ bags between $B_i$ and $B_{i+1}$, call them $B_i^1,\ldots,
B_i^{t-1}$, all of which contain $s_i$. For $s\in[t-1]$ we add into $B_i^s$ the
variables $c_{i+1,s}, c_{i+1,s-1}, c_{i+1,s+1}$, for all $s'\in[b]$ with $s'\le
s$ we add to $B_i^s$ the variable $x_{i+1,s'}$ and for all $s'\in [b]$ with
$s'\ge s$ we add to $B_i^s$ the variables $x_{i,s'}$.  Informally, we are
exchanging at each step a variable $x_{i,s}$ with a variable $x_{i+1,s}$, while
making sure that these two variables appear together in some bag together with
$c_{i+1,s}, c_{i+1,s-1}$. Observe that all constructed bags have size $b+O(1)$,
as desired.  \end{proof}

\section{Single-exponential FPT problems}\label{sec:fpt}

In this section we consider a number of fundamental algorithmic problems,
parameterized by a linear structure (pathwidth or linear clique-width), such as
\textsc{Coloring}, \textsc{Independent Set}, and \textsc{Dominating Set}.  For
all the problems we consider, fine-grained tight results are already known, if
one assumes the SETH. For instance, for $k$-\textsc{Coloring}, it is known that
the problem can be solved in time $k^{\pw}n^{O(1)}$ and in time
$(2^k-2)^{\lcw}n^{O(1)}$, while obtaining algorithms with any improvement on
the bases of the exponentials is impossible under the SETH
(\cite{Lampis20,LokshtanovMS18}). Indeed, starting with the pioneering work of
Lokshtanov, Marx, and Saurabh \cite{LokshtanovMS18}, such fine-grained tight
bounds are today known for a plethora of standard problems parameterized by
pathwidth and other structural parameters, such as for example
\textsc{Independent Set} and its generalization to larger distances called
$d$-\textsc{Scattered Set} (\cite{LokshtanovMS18,KatsikarelisLP22}),
\textsc{Dominating Set} and its generalization to large distances
(\cite{BorradaileL16,KatsikarelisLP19}), \textsc{Hamiltonicity}
(\cite{CyganKN18}), \textsc{Max Cut} (\cite{LokshtanovMS18}), \textsc{Odd Cycle
Transversal} (\cite{LokshtanovMS18,HegerfeldK22}), \textsc{Steiner Tree},
\textsc{Feedback Vertex Set} (\cite{CyganNPPRW22}), \textsc{Bounded Degree
Deletion} (\cite{LampisV23}), and many others for various structural parameters
(\cite{BojikianCHK23,DubloisLP21,DubloisLP22,GanianHKOS22,HanakaKLOS20,HegerfeldK23,HegerfeldK23b,JaffkeJ23,OkrasaR21}).
We note in passing that one striking aspect of the current state of the art is
that for almost all these problems, the lower bound for the parameter that
represents a linear structure (pathwidth, linear clique-width) matches the
complexity of the best algorithm for the corresponding \emph{tree-like}
structure (treewidth, clique-width), giving the impression that the two
complexities coincide.

Our goal in this section is to revisit and strengthen these results. As
mentioned, one weakness of the state of the art is that all lower bound results
are known only as implications of the SETH. Therefore, if the SETH turns out to
be false, we know nothing about the complexity of all the aforementioned
problems, and if we obtain an improved algorithm for one of them, it would not
a priori be sufficient to improve the complexity of any other problem in the
list. This seems rather disappointing, because the best known algorithms for
most of these problems are essentially facets of the same algorithm: perform
straightforward DP on the path decomposition.

The solution we propose to this conundrum is to base our results not on the
SETH, but on the \ppseth. In addition to basing all lower bounds on a
hypothesis that seems more plausible (as indicated by the results of
\cref{sec:evidence}), this has the major benefit of allowing us to obtain
reductions in both directions. As a result, all lower bounds we present are not
mere implications of the \ppseth, but in fact equivalent reformulations. This
means that improving the standard DP algorithm for any of the problems we
consider would be \emph{both necessary and sufficient} for falsifying the
\ppseth\ and for improving the DP algorithms for \emph{all} the problems
considered. The end result is a much more solid understanding of the
fine-grained complexity of problems which are FPT parameterized by pathwidth or
linear clique-width.

As mentioned above, the list of problems on which we could attempt to apply
this approach is quite long. Our goal in this section is not to be exhaustive,
but rather to showcase the power and flexibility of the framework built on top
of the \ppseth, while investigating the limits of this approach.  To summarize
our results, we give the following list of problems which we show to be
equivalent to the \ppseth\ in this section.

\begin{enumerate}

\item For any $k\ge 3, \eps>0$, solving $k$-\textsc{Coloring} in time
$(k-\eps)^{\pw}n^{O(1)}$. (\cref{lem:coloringpw})

\item For any $k\ge 3, \eps>0$, solving $k$-\textsc{Coloring} in time
$(2^k-2-\eps)^{\lcw}n^{O(1)}$. (\cref{lem:coloringlcw1} and
\cref{lem:coloringlcw2})

\item For any $d\ge 2, \eps>0$, solving $d$-\textsc{Scattered Set} in time
$(d-\eps)^{\pw}n^{O(1)}$. (\cref{lem:scattered1} and \cref{lem:scattered2})

\item For any $\eps>0$, solving \textsc{Set Cover} in time
$(2-\eps)^{\pw}(n+m)^{O(1)}$, where $n$ is the size of the universe, $m$ the
number of sets, and $\pw$ denotes the pathwidth of the incidence graph.
(\cref{lem:setcover})

\item For any $r\ge 2, \eps>0$, solving $r$-\textsc{Dominating Set} in time
$(2r+1-\eps)^{\pw}n^{O(1)}$. (\cref{lem:domset1} and \cref{lem:domset2})

\end{enumerate}

It was already known that the SETH implies that no algorithm achieving any of
the above exists, so our work adapts previous reductions to the new setting,
strengthening these lower bounds by showing they are equivalent to each other
and to a more plausible assumption (the \ppseth), and that they are implied by
the \scc\ and the \kova\ (cf. \cref{thm:evidence}).  It is important to note,
however, that the work we invested in \cref{sec:csprobust} and
\cref{sec:cspmulti} pays off handsomely in this section.  Recall that we have
given equivalent reformulations of the \ppseth\ for \textsc{CSP}s with larger
alphabet, via a promise problem that only requires monotone assignments in the
negative case (\cref{cor:weird}). As a result, our reductions are significantly
simpler and shorter than the previous reductions, even though we mostly just
recycle existing gadgets.  We hope that the simplification we contribute in
this work will help to ``democratize'' the field of SETH and \ppseth-based
lower bounds and allow for the easier discovery of more \ppseth-equivalent
parameterized problems. Prior to this work, the only problem for which
equivalence to the \ppseth\ followed from previous work was
$d$-\textsc{Scattered Set} for $d=2$ (that is, \textsc{Independent Set}), given
in the work of Iwata and Yoshida (\cite{IwataY15}). Our results generalize
their work and allow us to go much further than is likely possible in their
setting, as we explain below.

\subparagraph*{Applicability} In this section we focus on only a few
fundamental problems to showcase the applicability of the \ppseth, mostly
because it would not be practical to try to be more exhaustive. The selection
is, however, not completely arbitrary and is meant to also trace the limits of
what we expect can be done. Our results are much more general and cover more
problems than those of Iwata and Yoshida \cite{IwataY15} who based their
investigation on the treewidth version of the \ppseth\ -- indeed the only
non-\textsc{SAT}-related problem they consider is \textsc{Independent Set}.
There are two main reasons why we are able to go farther, one that is mostly
practical and one that runs quite a bit deeper. The shallow practical reason is
that our preparatory work on \textsc{CSP}s makes the formulation of reductions
for problems with inconvenient complexity (that is, where the base is not a
power of $2$) much more tractable, and the formalization of the promise problem
of \cref{cor:weird} allows one to effortlessly generalize the technique used in
\cite{IwataY15} for \textsc{Independent Set} to other problems. In other words,
as far as we can see, there is no fundamental reason why it was not already
known from  \cite{IwataY15} that, for example, solving $d$-\textsc{Scattered
Set} in time $(d-\eps)^{\tw}n^{O(1)}$ is equivalent to the treewidth version of
the \ppseth\ -- no fundamental reason, that is, except for the fact that
formulating such a reduction would be very tedious and time-consuming, because
we would have to encode in the problem instance the machinery needed both to
obtain the right base, and to perform the counting part that turns monotone
assignments to consistent assignments. In this sense, one contribution of this
work is that we simplify the framework sufficiently to make it practical to
obtain many such results without superhuman effort.

There is, however, also a deeper reason why we are able to obtain more results
than \cite{IwataY15} and this is related to the fact that \cref{thm:robust} is
flexible enough to cover also the parameterization by the \emph{incidence}
pathwidth. In the case of treewidth, the ``obvious'' DP for parameter $\tw^I$
has complexity $3^{\tw^I}|\psi|^{O(1)}$ and obtaining an algorithm with the
(likely) optimal complexity $2^{\tw^I}|\psi^{O(1)}|$ requires advanced
techniques, such as fast subset convolution (see \cite{CyganFKLMPPS15}, Chapter
11.1).  As a result, it is not at all clear if the hypothesis that \tsat\
cannot be solved with dependence $(2-\eps)^{\tw}$ is equivalent to the
hypothesis that the dependence cannot be $(2-\eps)^{\tw^I}$. This may seem like
a minor inconvenience, but it implies that for many domination and covering
type problems, it would be very challenging (if not impossible) to obtain an
equivalence between a lower bound for their complexity parameterized by
treewidth and the treewidth version of the \ppseth. Hence, a reason we manage
to show that more problems are equivalent to the \ppseth\ is that covering
problems, which are indeed equivalent to the \ppseth\ in the pathwidth setting,
require advanced techniques which are a priori not easy to translate in the
treewidth setting.

On an intuitive level, a takeaway from this discussion is that the \ppseth\ is
expected to be equivalent to improving the complexity of problems for which the
optimal algorithm is believed to be ``simple DP''. Under this assumption, it is
natural that many domination-type problems drop out of the equivalence class of
the corresponding hypothesis for treewidth, because the best algorithm for this
parameter is not ``simple DP''. In other words, what fails is not the reduction
from \textsc{SAT} to the corresponding problem, but the reduction from the
problem to \textsc{SAT} parameterized by the width of the primal graph. This
unmasks a subtle fine-grained distinction between linear and tree-like
parameters, which exists for many problems and is hidden by the fact that the
running times of the best algorithms parameterized by pathwidth and treewidth
coincide.

\subparagraph*{Limits} Given the above discussion, we would in general expect
the \ppseth\ to give tight equivalent lower bounds for the complexity of any
problem for which, when parameterized by a linear structure width, the best
algorithm is simple DP. For pathwidth and linear clique-width, where fast
subset convolution techniques are usually not relevant, this covers a large
swath of natural problems, for example we see no major obstacle stopping us
from extending our results for $k$-\textsc{Coloring} to $k$-\textsc{Colorable
Deletion} (as in \cite{HegerfeldK22}), or from \textsc{Independent Set} to
\textsc{Bounded Degree Deletion} (as in \cite{LampisV23}).  It is, however,
important to discuss also a class of problems which \emph{do not} fall in this
category.  These are the connectivity problems, such as \textsc{Steiner Tree},
\textsc{Feedback Vertex Set}, and \textsc{Connected Dominating Set}, which are
treated using advanced techniques, such as Cut\&Count \cite{CyganNPPRW22}. If
we consider such a problem, the techniques of this section would likely easily
allow us to adapt the current reduction and show that, say, obtaining a
$(3-\eps)^{\pw}n^{O(1)}$ algorithm for \textsc{Steiner Tree} (to take an
example from \cite{CyganNPPRW22}) would falsify the \ppseth, rather than the
SETH. However, the blocking point to obtain \emph{equivalence} to the \ppseth,
would be reducing \textsc{Steiner Tree} back to \textsc{SAT}.  Adapting the
Cut\&Count algorithm would more naturally provide a reduction to
$\oplus$-\textsc{SAT}. Is the hypothesis that this problem is hard equivalent
to the \ppseth? This seems likely to be a very difficult question, so we
suspect that proving \ppseth-equivalence for connectivity problems
parameterized by linear width parameters is, if not impossible, at least
currently out of reach. Note that this is akin to the situation for the class
XNLP: \textsc{Feedback Vertex Set} parameterized by pathwidth$/\log n$ is known
to be XNLP-hard, but not known to be contained in XNLP \cite{BodlaenderGJJL22}.

\subsection{Coloring}\label{sec:coloring1}

In this section we consider the complexity of the $k$-\textsc{Coloring}
problem, for all fixed $k\ge 3$. Recall that it is already known that if this
problem admits algorithms of complexity $(k-\eps)^{\pw}n^{O(1)}$ or
$(2^k-2-\eps)^{\lcw}n^{O(1)}$, for any $k\ge 3, \eps>0$, then the SETH is false
\cite{Lampis20,LokshtanovMS18}, while algorithms essentially matching these
bounds are known for treewidth and clique-width respectively.  Our goal is to
strengthen these lower bounds by showing that they are \emph{equivalent} to the
\ppseth, that is, algorithms beating these bounds can be obtained if and only
if the \ppseth\ is false.

We will need to present reductions in both directions. Since \textsc{Coloring}
is a CSP problem, showing that if the \ppseth\ is false, then \textsc{Coloring}
admits a faster algorithm will be the easy direction. For the other direction,
we will rely on several simple tricks from previous reductions. 

First, we will actually reduce \textsc{CSP} (with an appropriate alphabet size)
to \textsc{List Coloring}, rather than \textsc{Coloring}, but using everywhere
lists which are subsets of $[k]$. If $k$ is fixed, \textsc{List Coloring} can
then be reduced to \textsc{Coloring} in a way that preserves the pathwidth (by
adding to the graph a $k$-clique, as in \cite{LokshtanovMS18}) and linear
clique-width (Lemma 4.1 of \cite{Lampis20}).

Second, we will use a simple gadget, called a weak edge in \cite{Lampis20}. In
a graph $G=(V,E)$ for two (not necessarily distinct) values $i_1,i_2\in[k]$, an
$(i_1,i_2)$-weak edge between two vertices $v_1,v_2\in V$ is a path consisting
of three new internal vertices linking $v_1,v_2$, where appropriate lists are
assigned to the internal vertices so that the only color combination forbidden
in $(v_1, v_2)$ is $(i_1, i_2)$. In other words, the internal vertex lists are
appropriately constructed so that if we assign colors $i_1,i_2$ to $v_1,v_2$
the coloring fails, while in all other cases the internal vertices can be
colored. We refer the reader to Definition 4.2 and Lemma 4.3 of \cite{Lampis20}
for the exact definition and proof of correctness of this gadget.

Armed with these basic tools, we are now ready to establish our results.

\begin{lemma}\label{lem:coloringpw} 

For all $k\ge3, \eps>0$ we have the following: there exists an algorithm
solving $k$-\textsc{Coloring} in time $O((k-\eps)^{\pw}n^{O(1)})$ if and only
if the \ppseth\ is false.

\end{lemma}

\begin{proof}

Fix $k,\eps$. For one direction we observe that $k$-\textsc{Coloring} is a
special case of  2-\textsc{CSP} with alphabet $[k]$. If the \ppseth\ is false,
by \cref{thm:csp} we thus have an algorithm solving $k$-\textsc{Coloring} in
$O((k-\eps)^{\pw}n^{O(1)})$.

For the converse direction, take an instance $\psi$ of 2-\textsc{CSP} with
alphabet $[k]$. By \cref{thm:csp}, if we can solve $\psi$ in time
$O((k-\eps)^{\pw(\psi)}|\psi|^{O(1)})$, then the \ppseth\ is false. We will
reduce this to $k$-\textsc{Coloring} in a straightforward way. For each
variable of $\psi$ we construct a vertex of a new graph $G$ and give it list
$[k]$; and for each constraint of $\psi$ involving two variables $x,y$, we
consider each of the at most $k^2$ possible assignments to $x,y$ that falsify
$c$. For each such $(i_1,i_2)\in[k]^2$ that falsifies $c$ we add an
$(i_1,i_2)$-weak edge between the vertices representing $x,y$. Correctness is
not hard to see. It is also easy to see that the pathwidth of the new graph is
at most $\pw(\psi)+1$, because we only added weak edges between $x$ and $y$
only if $x,y$ already had an edge in the primal graph of $\psi$, and weak edges
are just sub-divided edges.  \end{proof}

For linear clique-width, we need slightly (but not much) more work, so we break
down the result into two parts.

\begin{lemma}\label{lem:coloringlcw1}

For all $k\ge3$ we have the following: if the \ppseth\ is false, then
there exists $\eps>0$ and an algorithm solving $k$-\textsc{Coloring} in time
$O((2^k-2-\eps)^{\lcw}n^{O(1)})$.

\end{lemma}

\begin{proof}

Fix $k\ge 3$. Suppose we are given a graph $G$ and a linear clique-width
expression of $G$ with $\lcw(G)$ labels. We want to reduce the problem of
$k$-coloring $G$ to a \textsc{CSP} instance $\psi$ over alphabet $[2^k-2]$ such
that $\pw(\psi) = \lcw(G)+O(1)$. If the \ppseth\ is false, we will then be able
to solve this instance fast enough to obtain a $k$-\textsc{Coloring} algorithm
in the desired time. Suppose the given linear clique-width expression is a
string $\chi$ of length $t=|\chi|$. We construct at most $t\cdot\lcw(G)$
variables, call them $x_{i,j}$ for $i\in[\lcw(G)], j\in[t]$.  The intuitive
meaning of $x_{i,j}$ is the following: if we consider a coloring of $G$ and
look at the subgraph of $G$ associated with the length-$j$ prefix of $\chi$,
then $x_{i,j}$ encodes the subset of colors of $[k]$ which appear in a vertex
with label $i$. Note that we only allow the variables $x_{i,j}$ to take values
which correspond to non-empty proper subsets of $[k]$, hence we have $2^k-2$
choices.

Let us now be a little more precise. For $j\in[t]$ let $G_j$ be the graph
represented by the length-$j$ prefix of $\chi$. We denote $V_{i,j}$ the set of
vertices with label $i$ in $G_j$. We say that $V_{i,j}$ is \emph{live} if
$V_{i,j}\neq \emptyset$ and there is an edge of $G$ incident on some vertex of
$V_{i,j}$ which does not appear in $G_j$. It is now not hard to see that in any
valid coloring of $G$, the set of colors used in a live set $V_{i,j}$ is a
non-empty proper subset of $[k]$: the set is non-empty as
$V_{i,j}\neq\emptyset$, and it is not $[k]$ as any vertex that becomes a
neighbor of a vertex of $V_{i,j}$ later in the construction, must be connected
to all of $V_{i,j}$. We therefore construct a variable $x_{i,j}$ only for
\emph{live} sets $V_{i,j}$. In the remainder we will look at the values of
$x_{i,j}$ as subsets of $[k]$.

For the constraints we have the following:

\begin{enumerate}

\item If the $j$-th operation of $\chi$ is $I(i)$, for $i\in[k]$ and
$V_{i,j-1}, V_{i,j}$ are live, we add a constraint ensuring that
$x_{i,j}\supseteq x_{i,j-1}$ and $|x_{i,j}|\in [|x_{i,j-1}|, |x_{i,j-1}|+1]$.
If $V_{i,j}$ is live and $V_{i,j-1}=\emptyset$ we add the constraint that
$|x_{i,j}|=1$. For all $i'\in [k]\setminus\{i\}$ for which $V_{i',j}$ is live,
add the constraint $x_{i',j}=x_{i',j-1}$.

\item If the $j$-th operation of $\chi$ is $R(i_1\to i_2)$, for $i_1,i_2\in[k]$
and $V_{i_2,j}$ is live, add a constraint ensuring that
$x_{i_2,j}=x_{i_1,j-1}\cup x_{i_2,j-1}$. For all $i'\in[k]\setminus\{i_1,i_2\}$
for which $V_{i',j}$ is live, add the constraint $x_{i',j}=x_{i',j-1}$.  

\item If the $j$-th operation of $\chi$ is $J(i_1\to i_2)$, add the constraint
$x_{i_1,j-1}\cap x_{i_2,j-1} = \emptyset$. For all $i\in[k]$ for which
$V_{i,j}$ is live, add the constraint $x_{i,j}=x_{i,j-1}$.

\end{enumerate}

This completes the construction of the \textsc{CSP} instance. It is not hard to
show by induction that for all $j\in [t]$, $G_j$ has a $k$-coloring if and only
if there is an assignment to the variables $x_{i,j'}$, for $i\in[\lcw(G)],
j'\le j$ that satisfies all constraints induced by these variables.  The
pathwidth of the primal graph of the produced instance is at most $\lcw(G)+2$,
because the primal graph is essentially contained in a grid of dimensions
$\lcw(G)\times t$.  \end{proof}

\begin{lemma}\label{lem:coloringlcw2}

For all $k\ge3, \eps>0$ we have the following: If there exists an algorithm
solving $k$-\textsc{Coloring} in time $O((2^k-2-\eps)^{\lcw}n^{O(1)})$, then
the \ppseth\ is false.

\end{lemma}

\begin{proof}

We will make use of \cref{cor:weird}. Suppose we are given a 4-\textsc{CSP}
instance $\psi$ with alphabet size $B=2^k-2$, a partition of the variable set
into $V_a,V_b$, a path decomposition of the primal graph of width $p$ where
each bag contains $O(B\log p)$ variables of $V_b$, and an injective function
$b$ mapping constraints to bags that contain their variables. We want to reduce
the problem of deciding whether a multi-assignment that is monotone and
consistent for $V_b$ exists to $k$-\textsc{Coloring} on a graph $G$ with
$\lcw(G) = p +O(B\log p)$.  If we can do this, and if there exists a fast
$k$-\textsc{Coloring} algorithm as stated in the lemma, then by
\cref{cor:weird} we would falsify the \ppseth.  As mentioned, we will in fact
reduce to \textsc{List Coloring}, where all lists are subsets of $[k]$, and use
known reductions from this problem to $k$-\textsc{Coloring}.

Assume that the given decomposition is nice (or use \cref{lem:nice}) and
partition the variables of $\psi$ into $p+1$ sets $V_1,\ldots,V_{p+1}$, so that
each bag contains at most $1$ variable from each set, using \cref{lem:pwcolor}
for $k=p+1$.  Suppose the bags of the decomposition are numbered $B_1,\ldots,
B_t$.  Without loss of generality, assume also that $B_1=\emptyset$ (this can
be achieved while keeping the decomposition nice by adding a sequence of bags
in the beginning, removing one by one the elements of $B_1$). 

We will construct a graph $G$ and also give a corresponding linear clique-width
expression using $p+O(kB^4)$ labels, of which one special label $*$ will be
called a junk label (its intended meaning is that once we give a vertex this
label, the vertex receives no further edges).  For each $j\in\{2,\ldots,t\}$,
suppose we have already constructed a graph for the first $j-1$ bags and a
linear clique-width expression given by string $\chi_{j-1}$. Furthermore,
suppose that in $G_{j-1}$ all vertices have a label from $[p+1]$ or the junk
label. $G_1$ is the empty graph which represents $B_1$, which is an empty bag.
We do the following:

\begin{enumerate}

\item If $B_j=B_{j-1}\setminus\{x\}$, and $x\in V_i$, for some $i\in[p+1]$ we
append to $\chi_{j-1}$ the operation $R(i\to*)$, that is, we rename vertices
with label $i$ to the junk label.

\item If $B_j=B_{j-1}\cup\{x\}$, and $x\in V_i$, for some $i\in[p+1]$, we
append to $\chi_{j-1}$ $k-1$ operations $I(i)$. All new vertices constructed in
this step have list $[k]$.

\item If there exists a (unique) constraint $c$ that is mapped to $B_j$ by $b$
we consider the list of at most $B^4$ satisfying assignments of $c$, call it
$\mathcal{S}_c$, with $\ell_c=|\mathcal{S}_c|$. For each $\sigma\in
\mathcal{S}_c$ construct a vertex $t_{c,\sigma}$ with list $\{1,2,3\}$.
Construct another $\ell_c+1$ vertices with list $\{1,2\}$ and connect them to
the vertices $t_{c,\sigma}$ in a way that a cycle of length $2\ell+1$ is formed
in which no two $t_{c,\sigma}$ vertices are adjacent. For each
$\sigma\in\mathcal{S}_c$, construct $4k$ vertices and place weak edges between
these vertices and $t_{c,\sigma}$ ensuring that if $t_{c,\sigma}$ receives
color 3, then for each $i\in[k]$ exactly $4$ of these vertices receive color
$i$; otherwise the vertices can be colored in any way. Call the vertices that
receive color $i$ if $t_{c,\sigma}$ receives color $3$,
$t_{c,\sigma,i,\alpha}$, for $\alpha\in\{1,2,3,4\}$. What we have constructed
so far can easily be accomplished using $O(kB^4)$ labels in a way that each
$t_{c,\sigma,i,\alpha}$ has a distinct label. Suppose that $c$ involves
variables from $V_{i_1}, V_{i_2}, V_{i_3}, V_{i_4}$. For each
$\alpha\in\{1,2,3,4\}$, each assignment $\sigma\in\mathcal{S}_c$, and each
$i\in[k]$, if $\sigma$ dictates that the variable of $V_{i_\alpha}$ involved in
$c$ takes a value in $[2^k-2]$ which, when viewed as a set, does not contain
$i$, add a Join operation between label $i_1$ and the label of
$t_{c,\sigma,i,\alpha}$; if $\sigma$ dictates that the same variable takes a
value which (viewed as a set) contains $i$, add a Rename operation from the
label of $t_{c,\sigma,i,\alpha}$.  Conclude this part by renaming all vertices
that still have a label outside $[p+1]$ to the junk label $*$.

\end{enumerate}

This concludes the construction and we have already argued why the graph has
clique-width $p+O(kB^4) = p+O(1)$. 

To prove equivalence, if there is a satisfying assignment for $\psi$, we obtain
a coloring of $G$ as follows: when we introduce a group of $k-1$ vertices
representing a variable $x$, we give them colors that correspond to the value
of $x$ in $\{2^k-2\}$, viewed as a subset of $[k]$ (that is, if the $i$-th bit
is $1$, we color a vertex in the group $i$); when we introduce a vertex
representing $x\in V_b$ we give it color $1$ or $2$ corresponding to the value
of $x$ in the satisfying assignment. We now claim that for all gadgets
constructed in step 4 we can color the $O(kB^4)$ added vertices. In particular,
if constraint $c$ is satisfied via assignment $\sigma$, we give color $3$ to
$t_{c,\sigma}$ and use $1,2$ for the rest of the odd cycle; this eliminates all
weak edges not incident on $t_{c,\sigma}$. For vertices $t_{c,\sigma,i,\alpha}$
we use the forced color, but this does not create any improperly colored edge,
because any Join operation these vertices are involved with is with a set that
does not contain its color; furthermore, renaming these vertices does not
change the set of colors contained in any label class. For remaining vertices
$t_{c,\sigma',i,\alpha}$, for $\sigma'\neq\sigma$, if we performed a Join on
such a vertex we pick an arbitrary color that does not produce a monochromatic
edge, while if we renamed it into a label in $[p+1]$ we pick a color that is
already present in that label class. This preserves the invariant that label
class $i\in[p+1]$ contains a set of colors representing the value of $x\in
V_i\cap B_j$.

For the converse direction, we extract an assignment to $V_b$ variables by
looking at the $k-1$ vertices introduced for each such variable and reading the
set of colors used (which is a non-empty proper subset of $[k]$) as a value in
$[2^k-2]$. Note that it is crucial here that when the path decomposition
introduces a variable $x\in V_b\cap V_i$, and we introduce $k-1$ vertices with
label $i$ in the new graph, all pre-existing vertices of label $i$ have been
forgotten, since no other vertex of $V_i$ can appear in the current bag. 

For variables of $V_a$ we extract a multi-assignment as follows: for each $x\in
V_a\cap V_i$ and $j$ such that $x\in B_j$ we look at the set of colors used by
vertices with label $i$ in the graph $G_j$, that is, the graph constructed
after we have processed bag $B_j$.  We claim that this assignment must satisfy
any constraint $c$ with $b(c)=j$.  Indeed, because the $t_{c,\sigma}$ vertices
participate in an odd cycle, one of them must have color $3$, meaning that the
vertices $t_{c,\sigma,i,\alpha}$ receive forced colors. These ensure that the
colors present in label classes $i_1, i_2, i_3$ are consistent with $\sigma$,
as we Join such classes with all the colors they are not supposed to contain in
$\sigma$, and we add to them all the colors they are supposed to contain.
Finally, observe that the multi-assignment we have extracted is monotone, as
the set of colors contained in a label class can only increase, until the point
when we forget the corresponding variable in the decomposition. Therefore, if
there exists a list coloring of the constructed graph, we have a
multi-assignment as demanded by \cref{cor:weird}.  \end{proof}

\subsection{Independent Set and Scattered Set}

In this section we consider a generalization of \textsc{Independent Set} which
is sometimes called in the literature $d$-\textsc{Scattered Set}. In this
problem we are given a graph $G=(V,E)$ and are asked to select a maximum
cardinality $S\subseteq V$ such that all vertices of $S$ are at pairwise
distance at least $d$. Clearly, for $d=2$ this problem is exactly
\textsc{Independent Set}. We show that for all fixed $d\ge 2$, solving
$d$-\textsc{Scattered Set} with parameter dependence better than $d^{\pw}$ is
equivalent to the \ppseth. This strengthens the lower bounds given by
Katsikarelis et al. \cite{KatsikarelisLP22}, who proved that no algorithm can
solve this problem in time $(d-\eps)^{\pw}n^{O(1)}$, for any $d\ge2, \eps>0$,
assuming the SETH. Note that an algorithm with complexity $d^{\pw}n^{O(1)}$ is
obtained via standard DP in \cite{KatsikarelisLP22}, so this result is tight.

We note also that Iwata and Yoshida \cite{IwataY15} have shown that, for the
specific case of $d=2$, \textsc{Independent Set} is equivalent to the treewidth
version of the \ppseth, that is, the hypothesis that \textsc{SAT} cannot be
solved with parameter dependence better than $(2-\eps)^{\tw}$. Their reductions
can easily go through unchanged if we start from the \ppseth\ instead, so the
fact that solving \textsc{Independent Set} with this complexity is equivalent
to the \ppseth\ essentially follows from their work. Our result generalizes
theirs to larger values of $d$. However, we prefer to give a unified proof for
all cases $d\ge 2$, including $d=2$, to showcase the flexibility of our
framework.

We begin by giving a complexity upper bound.

\begin{lemma}\label{lem:scattered1} If the \ppseth\ is false, then for all
$d\ge2$ there exists $\eps>0$ such that there is an algorithm that solves
$d$-\textsc{Scattered Set} in time $O((d-\eps)^{\pw}n^{O(1)})$.  \end{lemma}

\begin{proof}

For $d=2$ the lemma is trivial, as \textsc{Independent Set} can be seen as a
special case of \textsc{MaxW-}$2$\textsc{-SAT}, which, if the \ppseth\ is
false, can be solved in the desired amount of time by \cref{thm:robust}. We
therefore assume $d>2$.

Given an instance $G=(V,E)$ of $d$-\textsc{Scattered Set} and a path
decomposition of width $\pw(G)$,  we will reduce it to an instance of
2-\textsc{CSP} over an alphabet of size $d$, while increasing the pathwidth by
at most an additive constant. If the \ppseth\ is false, then the resulting
instance can be solved fast enough to obtain the lemma by \cref{thm:csp}.

Suppose we have a nice path decomposition of $G$ with the bags numbered
$B_1,\ldots, B_t$. Suppose also that for each $x,y\in V$ we have calculated the
shortest path distance $d(x,y)$ in $G$. For $j\in[t]$ we will denote by $G_j$
the edge-weighted graph whose vertex set is $V(G_j)=\bigcup_{j'=1}^jB_j$ and
where for each $x,y\in V(G_j)$ we have an edge $xy$ with weight $d(x,y)$. In
other words, $G_j$ is constructed by first taking the metric closure of $G$ and
then keeping only vertices of the first $j$ bags.

Let us now describe the construction of our \textsc{CSP} instance. For each
$x\in V$ we construct $t$ \textsc{CSP} variables, $x_j$, for $j\in[t]$. For a
vertex $x$ appearing in a bag $B_j$ the informal meaning of $x_j$ is that it
represents the distance between $x$ and its closest selected vertex in the
graph $G_j$. We will allow these variables to take $d$ possible values, in the
set $\{0,\ldots, d-1\}$. We now add the following constraints:

\begin{enumerate}

\item For each $j\in[t]$ and $x,y\in B_j$, we add the constraints $x_j\le
y_j+d(x,y)$ and $y_j\le x_j+d(x,y)$.

\item For each $j\in[t-1]$ and $x\in B_{j}\cap B_{j+1}$ we add the constraint
$x_{j+1}\le x_j$. Furthermore, we add a constraint ensuring that if
$x_{j+1}=0$, then $x_{j}=0$.

\item For each $j\in[t-1]$ such that there exists $x\in B_{j+1}\setminus B_j$,
for each $y\in B_{j}$, we add the constraint ensuring that if $x_j=0$, then
$y_{j-1}+d(x,y)\ge d$.

\item For each $x\in V$ such that $x$ appears in bags with indices in
$[j_1,j_2]$, we add the constraints $x_{j'}=x_{j'+1}$ for all
$j'\in[j_1-1]\cup\{j_2,\ldots,t-1\}$.

\end{enumerate}

To complete the construction we define a weight function that gives weight $1$
to value $0$ and weight $0$ to every other value. If the target value for the
scattered set was $k$, then we set the target weight to $tk$.

To prove correctness, suppose first that a $d$-scattered set $S$ of size $k$
exists in $G$. We assign value $0$ to all $x_j$, such that $x\in S, j\in[t]$,
achieving the target weight. For all $x\not\in S$ and $j\in[t]$ such that $x\in
B_j$, let $y$ be the vertex of $S\cap V(G_j)$ that is closest to $x$. Then,
assign to $x_j$ value $d(x,y)$. If $y$ does not exist, or $d(x,y)\ge d-1$,
assign to $x_j$ value $d-1$. Extend this assignment for all $x$ and $j$ for
which $x\not\in B_j$ in a way that satisfies constraints of the last step (give
the same value as $x_{j_1}$ to $x_{j'}$ for $j'<j$, and similarly for
$j'>j_2$). The assignment we have constructed is satisfying for the following
reasons: for constraints of the first step, observe that the distance from $x$
to $S$ in $G_j$ is at most $d(x,y)$ plus the distance from $y$ to $S$; for the
second step, the distance from $x$ to $S$ can only decrease from $G_j$ to
$G_{j+1}$; for the third step, if $x\in S$ and there existed a vertex $y\in
V(G_j)$ at distance at most $d-1-d(x,y)$ from $S$ in $G_j$, then $x$ would be
too close to another element of $S$.

For the converse direction, suppose we find a satisfying assignment of the new
instance that sets as many variables as possible to $0$. Because of the
constraints of step $2$, the number of such variables is a multiple of $t$, and
if we have weight at least $tk$, we can extract a set $S\subseteq V$ of at
least $k$ vertices of $G$, which we claim is a $d$-scattered set. To see this,
we will prove that for all $j\in[t]$ and $x\in B_j$, if $S\cap
V(G_j)\neq\emptyset$, then the variable $x_j$ must take value at most as large
as the distance from $x$ to $S\cap V(G_j)$. We can establish this by induction
on $j$. For $j=1$, if $S\cap B_1=\emptyset$, then the statement is vacuous,
while if $S\cap B_1\neq\emptyset$, then the constraints of the first step
assure the property because for $y\in S$ we have $y_1=0$. Since by constraints
of step 2 the values assigned to $x_j$ can only decrease as we increase $j$, if
$B_{j+1}\cap S=\emptyset$ and the statement is true for $G_j$, then the
statement is true for $G_{j+1}$. The remaining case is when $B_{j+1}$
introduces a new vertex in $S$, call this vertex $y$, with $y_{j+1}=0$.
However, for all $x\in B_j$ the constraints of step 1 ensure that if $x$ is
closer to $y$ than to $S\cap V(G_j)$, then the value $x_{j+1}$ will decrease
from $x_j$ to (at most) $d(x,y)$. Having established the property that for
$x\in B_j$ we have $x_j\le d(x,S\cap V(G_j))$, constraints of step 3 ensure
that the first time we introduce a vertex of $S$, we check that this vertex is
not too close to any previously introduced vertex of $S$. This is true, because
in this case the remaining vertices of the bag form a separator, so any
shortest path to a previous vertex of $S$ would have to go through them.

Finally, we need to argue that the new instance has small pathwidth. Take the
path decomposition of $G$ and replace in each $B_j$ each vertex $x\in B_j$ with
the variable $x_j$. This covers constraints of step 1. For each $j$ such that
there exists $x\in B_{j+1}\setminus B_j$, also place $x_{j+1}$ in $B_j$. Since
the decomposition is nice, this can only increase the width by $1$, and we
cover constraints of step 3. For each $j\in[t-1]$, in order to cover the
constraints of step 2, we insert between $B_j$ and $B_{j+1}$ a sequence of bags
which start with $B_j$ and at each step add a variable $x_{j+1}$ and then
remove $x_j$, one by one. Finally, to cover the remaining variables and
constraints of step 4 it suffices for each $x\in V$ that appears in the
interval $B_{j_1},\ldots, B_{j_2}$ to insert to the left of $B_{j_1}$ a
sequence of bags that contain all of $B_{j_1}$ and a path decomposition of the
path formed by $x_1,\ldots, x_{j_1}$, and similarly after $B_{j_2}$.
\end{proof}

We now proceed to the reduction in the converse direction.

\begin{lemma}\label{lem:scattered2} If for some $d\ge2, \eps>0$ there is an
algorithm that solves $d$-\textsc{Scattered Set} in time
$O((d-\eps)^{\pw}n^{O(1)})$, then the \ppseth\ is false.  \end{lemma}

\begin{proof}

We present a reduction from the 4-\textsc{CSP} problem of \cref{cor:weird} with
alphabet size $d$ to $d$-\textsc{Scattered Set}. We are given a \textsc{CSP}
instance $\psi$, a partition of its variables into two sets $V_1, V_2$, a path
decomposition $B_1,\ldots, B_t$ of $\psi$ of width $p$ such that each bag
contains at most $O(d\log p)$ variables of $V_2$, and an injective function $b$
mapping each constraint to a bag that contains its variables. We want to
construct a $d$-\textsc{Scattered Set} instance $G$ with $\pw(G)=p+O(d\log
p+d^6)$ and a target size $k$, such that if $\psi$ is satisfiable, then $G$ has
a $d$-scattered set of size at least $k$, while if $G$ has such a $d$-scattered
set, $\psi$ admits a monotone satisfying multi-assignment which is consistent
for $V_2$.

Suppose that the variables of $\psi$ are numbered $x_1,\ldots, x_n$ and that
the given path decomposition is nice. We construct the following graph:

\begin{enumerate}

\item For each $x_i\in V_2$ we construct a path $P_i$ on $d$ vertices, call
them $p_i^1,\ldots,p_i^d$.

\item For each $x_i\in V_1$ and $j\in[t]$ with $x_i\in B_j$ we construct a path
$P_{i,j}$ on $d$ vertices, call them $p_{i,j}^1,\ldots,p_{i,j}^d$.

\item For each $j\in[t-1]$ and $x\in B_j\cap B_{j+1}\cap V_1$ we add an edge
$p_{i,j}^dp_{i,j+1}^1$.

\item For each constraint $c$ let $j=b(c)$ and suppose $c$ involves variables
$x_{i_1},x_{i_2},x_{i_3}, x_{i_4}\in B_j$. Consider the set $\mathcal{S}_c$ of
the at most $d^4$ satisfying assignments of $c$. Construct two sets of vertices
$X_c = \{ x_{c,\sigma}\ |\ \sigma\in\mathcal{S}_c\}$ and $Y_c = \{
y_{c,\sigma}\ |\ \sigma\in\mathcal{S}_c\}$ and add edges so that each of the
two sets induces a clique.  For each $\sigma\in \mathcal{S}_c$ and
$\alpha\in\{1,2,3,4\}$, if $\sigma(x_{i_\alpha})= v_{\alpha}$, then connect
$y_{c,\sigma}$ with $p_{i_\alpha,j}^{v'}$ (or $p_{i_\alpha}^{v'}$, if
$x_{i_\alpha}\in V_2$), for all $v'\in[d]\setminus\{v_{\alpha}\}$, using paths
of length $\lceil \frac{d}{2}\rceil$ going through new vertices.  For each
$\sigma\in\mathcal{S}_c$, connect $x_{c,\sigma}$ to $y_{c,\sigma}$ using a path
of length $\lfloor\frac{d}{2}\rfloor-1$, through new vertices if necessary.
For $d\le2$ this means that we contract the vertex $x_{c,\sigma}$ into
$y_{c,\sigma}$.

\end{enumerate}

We set the target $d$-scattered set size to be equal to be $k=L_1+L_2+m$, where
$L_1$ is the number of paths constructed in step 1, $L_2$ is the number of
paths constructed in step 2, and $m$ is the number of constraints of $\psi$.

For correctness, suppose we have a satisfying assignment $\sigma$ for $\psi$.
For each $x_i\in V_2$ we select into our solution $p_{i}^{\sigma(x_i)}$. For
each $x_i\in V_1$ and $j\in[t]$ such that $x_i\in B_j$ we select
$p_{i,j}^{\sigma(x_i)}$. Finally, for each constraint $c$, we select
$x_{c,\sigma_c}$, where $\sigma_c\in\mathcal{S}_c$ is the restriction of
$\sigma$ to the variables of $c$. The set we selected has the required size. To
see that it is a $d$-scattered set, we first note that any path between a
vertex of $P_i$ or $P_{i,j}$ and a vertex of the cliques $X_c$ has length at
least $\lceil \frac{d}{2}\rceil$. This implies that selected vertices from
cliques $X_c$ are at pairwise distances at least $d$; and that if two selected
vertices from the paths $P_i$ and $P_{i,j}$ are at distance at most $d-1$, then
the shortest path connecting them must use only edges added in steps 1,2, and
3. However, it is not hard to see that if we restrict ourselves to such edges,
the distances between selected vertices are at least $d$. Finally, to check
that no selected vertex from a clique $X_c$ is at distance at most $d-1$ to a
selected vertex from a path, we observe that each $x_{c,\sigma_c}\in X_c$ is at
distance at most $d-1$ only from the vertices of the paths corresponding to
vertices of $c$ which represent assignments that disagree with $\sigma_c$.
Since no such vertex has been selected, we have a valid scattered set.

For the converse direction we first observe that for each $c$, the vertices
added in step 4, namely $X_c, Y_c$ and possibly the internal vertices of paths
we constructed, are all at pairwise distances at most $d-1$. Hence, from each
such set any scattered set selects at most one vertex. Similarly, for each path
$P_i$ and $P_{i,j}$ the scattered set may select at most one vertex. Therefore,
a solution of the desired size must select exactly one vertex from each path
and from each constraint gadget. From the vertices selected in paths we extract
a multi-assignment: for $x_i\in V_1$ and $j\in [t]$ with $x_i\in B_j$ we assign
to $x_i$ value $v$ in $j$ if $p_{i,j}^v$ is selected; for $x_i\in V_2$ we
assign $x_i$ value $v$ if $p_i^v$ is selected. This assignment is monotone,
because if for $x_i\in B_j\cap B_{j+1}\cap V_1$ we assign values $v_1$ in $j$
and $v_2<v_1$ in $j+1$, this means we have in the scattered set $p_{i,j}^{v_1},
p_{i,j+1}^{v_2}$, which are at distance $v_2-v_1+d<d$, contradiction. Let us
also argue why the assignment is satisfying. Consider a constraint $c$ with
$b(c)=j$, involving four variables from $V_{i_1}, V_{i_2}, V_{i_3}, V_{i_4}$.
If the assignment we have given to the four variables is not satisfying, then
for each $\sigma\in\mathcal{S}_c$, the vertex $y_{c,\sigma}$ is at distance at
most $\lceil\frac{d}{2}\rceil$ from a selected vertex. This means that all
vertices of $X_c$ (and therefore all vertices of the gadget) are at distance at
most $\lfloor\frac{d}{2}\rfloor +\lceil\frac{d}{2}\rceil-1 = d-1$ from a
selected vertex, meaning that we cannot have selected any vertex from this
constraint gadget, contradiction.

Finally, to bound the pathwidth of the graph $G$, start with the decomposition
of the primal graph of $\psi$ and in each $B_j$ replace each $x_i\in V_2\cap
B_j$ with all the vertices of $P_i$. For each constraint $c$, let $b(c)=j$ and
add into $B_j$ all the at most $O(d^6)$ vertices we constructed in step 4 for
$c$ (that is, $X_c, Y_c$ and the internal vertices of the paths). So far each
bag contains $O(d\log p + d^6)$ vertices. To cover the remaining vertices, for
each $j\in [t]$ we replace the bag $B_j$ with a sequence of bags such that all
of them contain the vertices we have added to $B_j$ so far, the first bag
contains $\{p_{i,j}^1\ |\ x_i\in V_1\cap B_j\}$ and the last bag contains
$\{p_{i,j}^d\ |\ x_i\in V_1\cap B_j\}$. We insert a sequence of $O(pd)$ bags
between these two, at each step adding a vertex $p_{i,j}^{v+1}$ and then
removing $p_{i,j}^v$ in a way that covers all edges of paths of step 2. To
cover the edges of step 3 it suffices to add a similar sequence of bags between
the bag containing $\{p_{i,j}^d\ |\ x_i\in V_1\cap B_j\}$ and the bag
containing $\{p_{i,j+1}^1\ |\ x_i\in V_1\cap B_{j+1}\}$.  \end{proof}

We note that in \cref{thm:robust} we skipped the case of \textsc{MaxW-SAT} for
instances of arity $2$, as hardness for such instances follows from the
hardness of \textsc{Independent Set}, which we just established. We therefore
have the following corollary.

\begin{corollary}

There exists an algorithm solving \textsc{MaxW-}$2$\textsc{-SAT} in time
$O((2-\eps)^{\pw}|\psi|^{O(1)})$, where $\psi$ is the input formula, if and
only if the \ppseth\ is false. 

\end{corollary}

\subsection{Set Cover}

We present a tight \ppseth\ equivalence for the complexity of solving
\textsc{Set Cover}. Recall that in this problem we are given a universe $U$ of
$n$ elements and a collection $\mathcal{C}$ of $m$ subsets of $U$ and are asked
to select a minimum cardinality collection of sets from $\mathcal{C}$ whose
union covers $U$. The pathwidth of such an instance is the pathwidth of the
bipartite incidence graph that represents the set system, that is, the graph
that has $U$ on one side, $\mathcal{C}$ on the other, and edges $xC$, for $x\in
U, c\in\mathcal{C}$ whenever $x\in C$.

We recall that it is already known that an algorithm solving \textsc{Set Cover}
in time $(2-\eps)^{\pw}(n+m)^{O(1)}$ would refute both the SETH and the \scc.
This is because $\pw\le \min\{n,m\}$, and it was shown in
\cite{CyganDLMNOPSW16} that an algorithm for \textsc{Set Cover} running in time
$(2-\eps)^mn^{O(1)}$ contradicts the SETH, while for the \scc\ this follows
from the definition of the conjecture. As shown in \cref{sec:evidence}, the
\ppseth\ is weaker than both of these conjectures and in this section we
establish equivalence with the \ppseth, strengthening the lower bound.

\begin{lemma}\label{lem:setcover}

For all $\eps>0$ we have the following: \textsc{Set Cover} with $n$ elements
and $m$ sets can be solved in time $O((2-\eps)^{\pw}(n+m)^{O(1)})$ if and only
if the \ppseth\ is false.

\end{lemma}

\begin{proof}

One direction is easy: if the \ppseth\ is false, by \cref{thm:robust} we have
an algorithm solving \textsc{MaxW-SAT} in
$((2-\eps)^{\pw^I(\psi)}|\psi|^{O(1)})$ time, where $\psi$ is the input
formula. We construct from the \textsc{Set Cover} instance a formula that has a
variable $x_C$ for each $C\in\mathcal{C}$. For each element $u\in U$ we
construct the clause $(\bigvee_{C: u\in C} \neg x_C )$. An assignment of
maximum weight to this formula corresponds to a minimum cardinality set cover
(select sets $C$ for which $x_C$ is False). The incidence graph of $\psi$ is
exactly the same as the incidence graph of the \textsc{Set Cover} instance.

For the converse direction, we perform a reduction from the \textsc{4-CSP}
problem of \cref{cor:weird} for $B=2$. Given a \textsc{CSP} instance $\psi$ and
a path decomposition  $B_1,\ldots, B_t$ of its primal graph, with the variables
partitioned into $V_a, V_b$, we do the following:

\begin{enumerate}

\item For each $x\in V_b$ we construct an element $u_x$ and two sets $C_{x,1},
C_{x,2}$. The element $u_x$ is placed in both of these sets and no other set.

\item For each $x\in V_a$ and $j\in[t]$ such that $x\in B_j$ we construct an
element $u_{x,j}$ and two sets $C_{x,j,1}, C_{x,j,2}$. The element $u_{x,j}$ is
placed in both of these sets and no other set.

\item For each $j\in[t-1]$ and variable $x\in B_j\cap B_{j+1}\cap V_a$, we
construct a new element $u_{x,j,j+1}$ and place it into $C_{x,j,1}$ and
$C_{x,j+1,2}$.

\item For each constraint $c$ such that $b(c)=j$, suppose that $c$ involves
four variables $x_1,x_2,x_3,x_4\in B_j$ and has the set of satisfying
assignments $\mathcal{S}_c$. For each $\sigma\in\mathcal{S}_c$ we construct
four elements $y_{c,\sigma,\alpha}$, for $\alpha\in\{1,2,3,4\}$. For each
$\sigma$, if $\sigma$ assigns to $x_1,x_2,x_3,x_4$ values $v_1,v_2,v_3,v_4$
respectively, then for each $\alpha\in\{1,2,3,4\}$ we place
$y_{c,\sigma,\alpha}$ in the set $C_{x_{i_\alpha},j,v_\alpha}$. We add
$|\mathcal{S}_c|$ further sets: for each $\sigma\in\mathcal{S}_c$, the set
$R_{c,\sigma}$ contains all elements $y_{c,\sigma',\alpha}$, for
$\sigma'\in\mathcal{S}\setminus\{\sigma\}$, $\alpha\in\{1,2,3,4\}$. Add a new
element $w_c$ and place it into all sets $R_{c,\sigma}$ and into no other set.

\end{enumerate}

If the number of elements constructed in steps $1$ and $2$ is $L_1,L_2$
respectively, and the number of constraints of $\psi$ is $m$, we set the target
cover cardinality to $L_1+L_2+m$. This completes the construction. 

To argue for correctness, assume first that we have a satisfying assignment
$\sigma$ to $\psi$. We transform it into a set cover by selecting for each
$x\in V_a$ and $j$ such that $x\in B_j$, the set $C_{x,j,\sigma(x)}$, and for
each $x\in V_b$ the set $C_{x,\sigma(x)}$. For each constraint $c$ we select
the set $R_{c,\sigma_c}$, where $\sigma_c$ is the restriction of $\sigma$ to
the variables of $c$. This selection clearly has the desired cardinality and
covers elements of steps 1 and 2. Elements of step 3 are also covered, because
the selection is consistent. Elements of step 4 are covered, because $\sigma_c$
is satisfying for $c$, so the four elements $y_{c,\sigma_c,\alpha}$ which are
not covered by $R_{c,\sigma_c}$ are coverd by the other sets.

For the converse direction, if we have a set cover of the desired size we
observe that (i) we must select at least one set for each element of step 1
(ii) at least one set for each element of step 2 (iii) at least one set for
each constraint. If we achieve the desired target all three requirements become
tight, so we can extract an assignment for $V_b$ and a multi-assignment for
$V_a$. The multi-assignment must be monotone, because if for some $j$ and $x\in
B_j\cap B_{j+1}\cap V_a$ we have sets $C_{x,j,2}, C_{x,j+1,1}$, the element
added in step 3 is not covered. The multi-assignment must be satisfying because
for each $c$ we have selected a unique $R_{c,\sigma}$, for $\sigma$ a
satisfying assignment to $c$ and the elements $y_{c,\sigma,\alpha}$, which are
not covered by $R_{c,\sigma}$ must therefore be covered by our selection for
the variables.

Finally, let us bound the pathwidth of the new instance. For each $j\in[t]$ we
replace the bag $B_j$ of the original decomposition with a sequence of
$|B_j\cap V_a|\le p+1$ bags $B_j^1,\ldots,B_j^{|B_j\cap V_a|}$. For $\beta\in
[|B_j\cap V_a|]$ the bag $B_j^\beta$ contains: (i) for all $x\in V_b\cap B_j$
the elements $C_{x,1}, C_{x,2}, u_x$; (ii) if there exists $c$ with $b(c)=j$,
then $B_j^\beta$ contains all the at most $O(1)$ elements and sets of step 4
constructed for $c$. We now add elements to this sequence of bags to cover the
construction of step 2, namely, for each $x\in V_a\cap B_j$ we place
$C_{x,j,2}$ in $B_j^1$ and $C_{x,j,1}$ in $B_j^{|B_j\cap V_a|}$; for each $x\in
B_j\cap V_a$ we select a distinct bag in the sequence, we add to this bag
$u_{x,j}$, we add $C_{x,j,2}$ to all bags up to and including this bag, and we
add $C_{x,j,1}$ to all bags from this bag to the last. Finally, to cover the
edges of step 3 it suffices to add a sequence of bags between $B_j^{|V_a\cap
B_j|}$ and $B_{j+1}^1$ to cover the edges between the new elements
$u_{x,j,j+1}$ and sets $C_{x,j,2}, C_{x,j+1,1}$, again by exchanging elements
one pair at a time. The width of the resulting decomposition is $p+O(\log p)$.
\end{proof}

\subsection{Dominating Set}

For a positive integer $r\ge 1$, the $r$-\textsc{Dominating Set} problem is the
problem of selecting in a graph $G$ a set of vertices of minimum cardinality
such that all remaining vertices are at distance at most $r$ from a selected
vertex. For $r=1$ this is just \textsc{Dominating Set}. In this section we
present a tight bound on the complexity of solving $r$-\textsc{Dominating Set},
proving that, for all $r\ge 1$, obtaining an algorithm with parameter
dependence better than $(2r+1)^{\pw}$ is equivalent to the \ppseth. This
strengthens a result of Borradaile and Le \cite{BorradaileL16} who proved that
such an algorithm cannot be obtained under the SETH. We remark that obtaining
an algorithm with running time $(2r+1)^{\pw}n^{O(1)}$ is easy using dynamic
programming, and it is possible to achieve this performance even if we
parameterize by treewidth, by using fast subset convolution techniques
\cite{BorradaileL16}.

\begin{lemma}\label{lem:domset1} If the \ppseth\ is false, then for all $r\ge
1$ there exists $\eps>0$ and an algorithm that solves $r$-\textsc{Dominating
Set} in time $O((2r+1-\eps)^{\pw}n^{O(1)})$.  \end{lemma}

\begin{proof}

Fix some $r\ge 1$. We will reduce a given $r$-\textsc{Dominating Set} instance
$G$ to \textsc{MaxW-CSP} over an alphabet of size $2r+1$, while increasing the
pathwidth by at most an additive constant. If the \ppseth\ is false, then by
\cref{thm:robust} there is a fast algorithm for solving this \textsc{MaxW-CSP}
instance, which will allow us to obtain the desired running time for
$r$-\textsc{Dominating Set}.

Suppose we have a nice path decomposition of $G=(V,E)$ with the bags numbered
$B_1,\ldots, B_t$. Vertices of $G$ are numbered $V=\{v_1,\ldots, v_n\}$. We can
assume that there exists an injective function $b:E\to [t]$ which maps each
edge of $G$ to the index of a bag that contains both endpoints ($b$ can be made
injective by repeating bags if necessary).  Furthermore, we can assume that $b$
maps edges to bags which are not introducing new vertices, again by repeating
bags.

We will construct a \textsc{CSP} instance $\psi$ over an alphabet $\{-r,\ldots,
r\}$ of size $2r+1$ as follows:

\begin{enumerate}

\item For each $v_i\in V$ we construct $t$ variables $x_{i,j}$, for $j\in[t]$.

\item For each $v_i\in V$ and $j\in[t-1]$,  we add a constraint which states
that $|x_{i,j}|=|x_{i,j+1}|$. If no edge $e$ has $b(e)=j+1$ or for some edge
$e$ we have $b(e)=j+1$ but $e$ is not incident on $v_i$, then we also add the
constraint $x_{i,j+1}\le x_{i,j}$.

\item For each $v_i\in V$, if the last bag containing $v_i$ is $B_j$, then for
all $j'\in\{j,\ldots, t-1\}$ we add the constraint $x_{i,j'}=x_{i,j'+1}$.

\item For each $e\in E$ let $j=b(e)$ and $e=v_{i_1}v_{i_2}$ with
$v_{i_1},v_{i_2}\in B_j$. We add a constraint dictating that if $x_{i_1,j}\ge
0$, then $x_{i_1,j-1}\ge0$ or $|x_{i_1,j}|\ge |x_{i_2,j}|+1$.  We also add the
symmetric constraint for $x_{i_2,j}$.

\item For each $v_i\in V$ we add the constraint $x_{i,t}\ge 0$ and $x_{i,1}\le
0$.

\end{enumerate}

We define a weight function which gives weight $1$ to all values except value
$0$, which receives value $0$. The target weight is $(n-k)t$, where $k$ is the
target value for the $r$-dominating set.

For correctness, suppose we have an $r$-dominating set of $G$, call it $S$, of
size at most $k$. For each $v_i\in S$ we assign to $x_{i,j}$, for $j\in[t]$ the
value $0$. We will not reuse value $0$, so this will achieve the target weight.
For each other $v_i\in V\setminus S$, we assign all $x_{i,j}$ values such that
$|x_{i,j}|=d(v_i,S)\le r$.  More precisely, for each $v_i\in V\setminus S$, let
$e$ be the edge incident on $v_i$ that is the first edge of a shortest path
from $v_i$ to $S$. Let $j=b(e)$. We set $x_{i,j'}=-d(v_i,S)$ for all $j'<j,
j'\in[t]$ and $x_{i,j'}=d(v_i,S)$ for all remaining $j'$. This satisfies
constraints of steps 2,3 and 5. To see that constraints of step 4 are
satisfied, we observe that if $x_{i_2}$ is the other endpoint of $e$, then it
must be the case that $d(v_{i_2},S) = d(v_i,S)-1$ so $|x_{i,j}|=|x_{i_2,j}|+1$.

For the converse direction, suppose that we have a satisfying assignment with
the desired weight. By constraints of step 2, each group of $t$ variables is
either assigned $0$ everywhere or nowhere, therefore if we select the set $S=\{
v_i\ |\ x_{i,j}=0\ \textrm{for some }j\}$, then $|S|\le k$. We want to argue
that $S$ is an $r$-dominating set. For this we will prove that for all $v_i\in
V$, and $j\in[t]$, if $x_{i,j}\ge 0$, then $d(v_i,S)\le x_{i,j} $. Since
$x_{i,t}\ge 0$ for all $v_i\in V$, this will prove correctness. We prove the
statement by induction on $|x_{i,j}|$. For $|x_{i,j}|=1$, the constraints
ensure that $|x_{i,j}|$ is contant throughout, $x_{i,1}<0$, and $x_{i,t}>0$, so
there must be some $j$ such that $x_{i,j-1}=-1$ and $x_{i,j}=1$. Then, to avoid
violating the constraints of step 2, it must be the case that there exists $e$,
with $b(e)=j$, such that $e$ is incident on $v_i$. This implies that we have to
satisfy the constraint of step 4, which can only be satisfied if the other
endpoint of $e$ is a vertex of $S$. For the general case, if we have
$x_{i,j}>0$ and $x_{i,j-1}<0$, there must exist an edge $e$ with $b(e)=j$, $e$
is incident on $v_i$, and if $i_2$ is the other endpoint of $e$, then
$|x_{i_2,j}|\le x_{i,j}-1$. Since $x_{i_2,t}=|x_{i_2,t}|$ and we have
established by induction that $d(v_{i_2},S)=|x_{i_2,j}|$, we conclude that
$d(x_i,S)\le x_{i,j}$, as desired.

Finally, let us argue about the pathwidth of $\psi$.  Take the path
decomposition of $G$ and replace in each $B_j$ each vertex $v_i\in B_j$ with
the variable $x_{i,j}$. For each $e\in E$ such that $e=v_{i_1}v_{i_2}$ also
place $x_{i_1,j}, x_{i_2,j}$ in $B_{j-1}$. This covers constraints of step 4,
increasing the width by at most $2$ (since $b$ is injective). For each
$j\in[t-1]$ we insert between $B_j$ and $B_{j+1}$ a sequence of bags which
start with $B_j$ and at each step add a variable $x_{i,j+1}$ and then remove
$x_{i,j}$, one by one.  Finally, to cover the remaining variables and
constraints of steps 2,3 and 5, it suffices for each $v_i\in V$ that appears in
the interval $B_{j_1},\ldots, B_{j_2}$ to insert to the left of $B_{j_1}$ a
sequence of bags that contain all of $B_{j_1}$ and a path decomposition of the
path formed by $x_{i,1},\ldots, x_{i,j_1}$, and similarly after $B_{j_2}$.
\end{proof}

\begin{lemma}\label{lem:domset2} If for some $r\ge 1$ there exists $\eps>0$ an
algorithm that solves $r$-\textsc{Dominating Set} in time
$O((2r+1-\eps)^{\pw}n^{O(1)})$, then the \ppseth\ is false.  \end{lemma}

\begin{proof}

We present a reduction from the 4-\textsc{CSP} problem of \cref{cor:weird} with
alphabet size $(2r+1)$ to $r$-\textsc{Dominating Set}. We are given a
\textsc{CSP} instance $\psi$, a partition of its variables into two sets $V_1,
V_2$, a path decomposition $B_1,\ldots, B_t$ of $\psi$ of width $p$ such that
each bag contains at most $O(d\log p)$ variables of $V_2$, and an injective
function $b$ mapping each constraint to a bag that contains its variables. We
want to construct an $r$-\textsc{Dominating Set} instance $G$ with
$\pw(G)=p+O(r\log p+r^9)$ and a target size $k$, such that if $\psi$ is
satisfiable, then $G$ has an $r$-dominating set of size at most $k$, while if
$G$ has such an $r$-dominating set, $\psi$ admits a monotone satisfying
multi-assignment which is consistent for $V_2$. To ease presentation, we will
be looking for a multi-assignment that is monotonically \emph{decreasing}
rather than increasing. This is equivalent to what is required by
\cref{cor:weird}, as we can just flip the ordering of the bags of the given
decomposition.

Before we proceed, let us define a slightly more general version of
\textsc{Dominating Set}. Suppose that a subset of vertices of $G$ is given to
us marked as ``optional'', meaning that we are seeking a dominating set which
may select some of these vertices, but it is not mandatory to dominate them.
This more general version of the problem is actually equivalent to
$r$-\textsc{Dominating Set}. Indeed, given an instance of the problem with a
set of optional vertices $V_o$, we can add a new vertex $x$, connect it with
paths of length $r$ to all vertices of $V_o$, attach to $x$ a path of length
$r$, and increase the budget by $1$. Since without loss of generality any
$r$-dominating set of the new graph selects $x$ and no other new vertex, there
is a one-to-one correspondence between optimal dominating sets that dominate
non-optional vertices in the original graph and optimal dominating sets in the
new graph. Observe that this construction only increases the pathwidth by an
additive constant. Because of this, we will in the remainder allow ourselves to
mark some vertices as optional.

Suppose that the variables of $\psi$ are numbered $x_1,\ldots, x_n$ and that
the given path decomposition is nice. We construct the following graph:

\begin{enumerate}

\item For each $x_i\in V_2$ we construct a path $P_i$ on $2r+1$ vertices, call
them $p_i^1,\ldots,p_i^{2r+1}$. Construct a vertex $t_i$, connected to each
$p_i^v$, for $v\in[2r+1]$ via a path of length $r$ that goes through new
optional vertices.

\item For each $x_i\in V_1$ and $j\in[t]$ with $x_i\in B_j$ we construct a path
$P_{i,j}$ on $2r+1$ vertices, call them $p_{i,j}^1,\ldots,p_{i,j}^{2r+1}$.
Construct a vertex $t_{i,j}$, connected to each $p_i^v$, for $v\in[2r+1]$ via a
path of length $r$ that goes through new optional vertices.

\item For each $x_i\in V_1$ if the interval of bags containing $x_i$ is
$B_{j_1},\ldots,B_{j_2}$, we mark the vertices $p_{i,j_1}^1,\ldots,p_{i,j_1}^r$
and $p_{i,j_2}^{r+1},\ldots,p_{i,j_2}^{2r+1}$ as optional.

\item For each $j\in[t-1]$ and $x\in B_j\cap B_{j+1}\cap V_1$ we add an edge
$p_{i,j}^{2r+1}p_{i,j+1}^1$.

\item For each constraint $c$ let $j=b(c)$ and suppose $c$ involves variables
$x_{i_1},x_{i_2},x_{i_3}, x_{i_4}\in B_j$. Consider the set $\mathcal{S}_c$ of
the at most $(2r+1)^4$ satisfying assignments of $c$. For each
$\sigma\in\mathcal{S}_c$ construct four vertices $x_{c,\sigma,\alpha}$, for
$\alpha\in\{1,2,3,4\}$ and a vertex $y_{c,\sigma}$. Add edges so that vertices
$y_{c,\sigma}$, for $\sigma\in\mathcal{S}_c$ form a clique. Construct a vertex
$w_c$.  Connect $w_c$ to each $y_{c,\sigma}$, for $\sigma\in\mathcal{S}_c$, via
paths of length $r$, going through $r-1$ new vertices. For each
$\sigma,\sigma'\in\mathcal{C}$ with $\sigma\neq\sigma'$, connect $y_{c,\sigma}$
with each $x_{c,\sigma',\alpha}$, for $\alpha\in\{1,2,3,4\}$, through paths of
length $r$ through new vertices. For each $\sigma\in\mathcal{S}_c$ and
$\alpha\in\{1,2,3,4\}$, if $\sigma(x_{i_\alpha})=v_{\alpha}$, then connect
$x_{c,\sigma,\alpha}$ to $p_{i,j}^{v_\alpha}$ (or to $p_i^{v_\alpha}$ if
$x_i\in V_2$) using a path of length $r$ through new vertices. All internal
vertices of paths of length $r$ constructed in this step are marked as
optional.

\end{enumerate}

We set the target $r$-dominating set size to be equal to be $k=L_1+L_2+m$,
where $L_1$ is the number of paths constructed in step 1, $L_2$ is the number
of paths constructed in step 2, and $m$ is the number of constraints of $\psi$.

For correctness, suppose we have a satisfying assignment $\sigma$ for $\psi$.
For each $x_i\in V_2$ we select into our solution $p_{i}^{\sigma(x_i)}$. For
each $x_i\in V_1$ and $j\in[t]$ such that $x_i\in B_j$ we select
$p_{i,j}^{\sigma(x_i)}$. Finally, for each constraint $c$, we select
$y_{c,\sigma_c}$, where $\sigma_c\in\mathcal{S}_c$ is the restriction of
$\sigma$ to the variables of $c$. The set we selected has the required size. To
see that it is an $r$-dominating set for all non-optional vertices, we observe
that for each $x_i\in V_2$ all the vertices of $P_i$ and $t_i$ are dominated;
for each $x_i\in V_1$ and $j\in[t]$ such that $x_i\in B_j$ all the vertices of
$P_{i,j}$ and $t_{i,j}$ are dominated, except possibly the optional vertices of
$P_{i,j_1}$ or $P_{i,j_2}$, where $j_1,j_2$ are as in step 3; for each
constraint $c$, $w_c$ is dominated, all vertices  $y_{c,\sigma}$ are dominated,
and for all $\sigma'\in\mathcal{S}_c\setminus\{ \sigma_c\}$ and
$\alpha\in\{1,2,3,4\}$, $x_{c,\sigma',\alpha}$ is dominated, while the vertices
$x_{c,\sigma_c,\alpha}$ are dominated by the vertices selected in the paths.

For the converse direction, we first observe that if we consider a set made up
of the middle vertex of each path $P_i$ or $P_{i,j}$ and the vertex $w_c$ for
each constraint $c$, we have a set of $k$ vertices which are at pairwise
distance at least $2r+1$. This implies that a solution of size $k$ must
dominate each of these vertices exactly once, as no vertex can dominate two of
them. In particular, from each path $P_i$ or $P_{i,j}$, at most one vertex can
be selected in the solution. We can now extract a multi-assignment for $x_i\in
V_1$ and an assignment for $x_i\in V_2$: look at the path $P_{i,j}$ or $P_i$
and locate the vertex $p_i^v$ or $p_{i,j}^v$ which is in the shortest path from
the middle vertex of the path to the vertex that dominates it, and assign value
$v$ to $x_i$. This assignment is monotonically decreasing, because if for
$x_i\in B_j\cap B_{j+1}\cap V_1$ we assign values $v_1$ in $j$ and $v_2>v_1$ in
$j+1$, this means there is a sequence of $2r+1$ consecutive vertices of the
path $P_{i,j},P_{i,j+1}$ which contain no vertex of the dominating set, hence
there is a non-dominated vertex.  Let us also argue why the assignment is
satisfying. Consider a constraint $c$ with $b(c)=j$, involving four variables
from $V_{i_1}, V_{i_2}, V_{i_3}, V_{i_4}$. Without loss of generality, we have
selected a vertex $y_{c,\sigma}$ to dominate $w_c$, as any other vertex that
can dominate $w_c$ covers fewer vertices. We claim that $\sigma$ must be
consistent with our assignment. To see this, observe that $x_{c,\sigma,\alpha}$
are not covered by $y_{c,\sigma}$, therefore, they must be covered by the
vertices we have selected in the paths. If $\sigma(x_{i_\alpha})=v_\alpha$,
this can only happen if the path vertex closes to a selected vertex is
$p_{i_\alpha,j}^{v_\alpha}$.

Finally, to bound the pathwidth of the graph $G$, start with the decomposition
of the primal graph of $\psi$ and in each $B_j$ replace each $x_i\in V_2\cap
B_j$ with all the vertices of $P_i$. For each constraint $c$, let $b(c)=j$ and
add into $B_j$ all the at most $O(r^9)$ vertices we constructed in step 5 for
$c$. So far each bag contains $O(r\log p + r^9)$ vertices. To cover the
remaining vertices, for each $j\in [t]$ we replace the bag $B_j$ with a
sequence of bags such that all of them contain the vertices we have added to
$B_j$ so far, the first bag contains $\{p_{i,j}^1\ |\ x_i\in V_1\cap B_j\}$ and
the last bag contains $\{p_{i,j}^{2r+1}\ |\ x_i\in V_1\cap B_j\}$. We insert a
sequence of $O(pr)$ bags between these two, at each step adding a vertex
$p_{i,j}^{v+1}$ and then removing $p_{i,j}^v$ in a way that covers all edges of
paths of step 2. To cover the edges of step 3 it suffices to add a similar
sequence of bags between the bag containing $\{p_{i,j}^{2r+1}\ |\ x_i\in
V_1\cap B_j\}$ and the bag containing $\{p_{i,j+1}^1\ |\ x_i\in V_1\cap
B_{j+1}\}$.  \end{proof}

\section{Super-exponential FPT problems}\label{sec:fpt2}

\subsection{Coloring}\label{sec:coloring2}

In this section we consider the complexity of \textsc{Coloring} parameterized
by pathwidth, when the number of colors is part of the input. It is well known
(and easy to see) that on a graph of pathwidth $\pw$ it is always possible to
produce a proper coloring using $\pw+1$ colors in polynomial time, hence
\textsc{Coloring} is FPT and admits an algorithm with complexity
$\pw^{\pw}n^{O(1)}$ via the standard DP algorithm. The question we focus on is
whether this algortihm is best possible. 

It is already known that if we disregard constants in the exponent, then the
algorithm we mentioned cannot be significantly improved: Lokshtanov, Marx, and
Saurabh showed that an algorithm with running time of $\pw^{o(\pw)}n^{O(1)}$
would falsify the ETH (\cite{LokshtanovMS18b}). However, we are asking a more
fine-grained question here, which sets a more modest goal: is it possible to
improve upon this algorithm even slightly, obtaining a running time of the form
$\pw^{(1-\eps)\pw}n^{O(1)}$? Our results in this section show that the answer
is likely to be negative, and in fact that obtaining such an algorithm is both
necessary and sufficient for falsifying the \ppseth. In other words, we obtain
a lower bound result that is sharper than that of \cite{LokshtanovMS18b},
because we not only show that the dependence on treewidth has to be exponential
in $\Theta(\pw\log\pw)$, but we precisely determine that the coefficient of the
leading term must be equal to $1$. 

To establish equivalence we need a reduction from and to \textsc{SAT}. We give
these reductions in \cref{lem:coloring2a} and \cref{lem:coloring2b}.

\begin{lemma}\label{lem:coloring2a}

If the \ppseth\ is false, then there exists an $\eps>0$ and an algorithm
solving \textsc{Coloring} in time $\pw^{(1-\eps)\pw}n^{O(1)}$.

\end{lemma}

\begin{proof}

Suppose that the \ppseth\ is false, so there exists $\delta>0$ and an algorithm
that decides the satisfiability of a CNF formula $\phi$ in time
$2^{(1-\delta)\pw(\phi)}|\phi|^{O(1)}$. Suppose also that we are given a graph
$G$ on $n$ vertices and a tree decomposition of width $\pw(G)$, as well as an
integer $k$, and we want to decide if $G$ can be colored with $k$ colors. 

We will assume without loss of generality that $\pw(G)>2^{4/\delta}$
(otherwise, $\pw(G)$ is bounded by an absolute constant and we can decide the
$k$-colorability of $G$ in linear time).  We can also assume without loss of
generality that $k\le \pw(G)$, because if $k\ge \pw(G)+1$ a $k$-coloring can
always be found by a simple greedy algorithm, using the fact that every induced
subgraph of $G$ contains a vertex of degree at most $\pw(G)$.  We will
construct in polynomial time an instance of \textsc{SAT} $\phi$ that satisfies
the following:

\begin{enumerate}

\item $\phi$ is satisfiable if and only if $G$ is $k$-colorable.

\item  $\pw(\phi) \le \pw(G)\log(\pw(G)) + 2\pw(G)$

\item $|\phi|=n^{O(1)}$

\end{enumerate}

Observe that if we have the above, then we obtain the lemma. This is because we
can decide $k$-colorability of $G$ by invoking the supposed algorithm that
falsifies the \ppseth.  We note that $(1-\delta)\pw(\phi) \le
(1-\delta)(\pw(G)\log(\pw(G)) + 2 \pw(G))\le
(1-\frac{\delta}{2})\pw(G)\log(\pw(G)) -
\pw(G)\left(\frac{\delta}{2}\log(\pw(G)) - 2 \right) \le
(1-\frac{\delta}{2})\pw(G)\log(\pw(G))$, where we use the fact that
$\pw(G)>2^{4/\delta}$. But then the running time of our procedure is
$2^{(1-\delta)\pw(\phi)}|\phi|^{O(1)} \le
(\pw(G))^{(1-\frac{\delta}{2})\pw(G)}n^{O(1)}$, so it suffices to set
$\eps=\delta/2$.

We now describe the construction of $\phi$. Let $t=\lceil\log k \rceil \le \log
\pw(G) +1$ be a number of bits sufficient to write $k$ distinct integers (from
the set $\{0,\ldots,k-1\}$) in binary. For each vertex $v$ of $G$ we construct
$t$ variables $v_i, i\in[t]$, which informally are meant to encode the color of
$v$ in binary. If $k$ is not a power of $2$, then some of the assignments to
the variables $v_i$ encode integers outside of $\{0,\ldots,k-1\}$; in this
case, for each $v$, and for each such assignment, we construct a clause of size
$t$ that is falsified by this assignment. Now, for each edge $e=uv$ in $G$ we
construct $t$ variables $e_i, i\in [t]$. Informally, $e_i$ is set to True if
the binary representations of the colors of $u$ and $v$ differ in the $i$-th
bit. We add the clause $(e_1\lor e_2\lor\ldots\lor e_t)$. Also, for each $i\in
[t]$, we add clauses encoding the constraint $e_i\leftrightarrow (u_i\neq
v_i)$. This completes the construction, which can clearly be carried out in
polynomial time.

It is now not hard to see that $\phi$ is satisfiable if and only if $G$ is
$k$-colorable. For one direction, if $G$ is $k$-colorable, we use the variables
$v_i$ to encode the color of $v$ in binary and for each edge $e=uv$ we set
$e_i$ to True if and only if the binary encodings of the colors of $u,v$ differ
in the $i$-th bit. This satisfies all clauses of $\phi$; in particular, clauses
of the form $(e_1\lor\ldots\lor e_t)$ are satisfied because if we have a proper
coloring, the encodings of the colors of the endpoints of $e$ must differ in
some bit. For the converse direction, we have added clauses ensuring that from
a satisfying assignment we can extract a $k$-coloring of $V(G)$. To see that
this coloring is proper, for each $e=uv$ the clause $(e_1\lor\ldots\lor e_t)$
ensures that at least one $e_i$ is True, but this is only possible if the
encoding of the colors of $u,v$ differ in the $i$-th bit and are therefore
encodings of distinct colors.

Finally, to obtain the bound on the pathwidth, we start with a decomposition of
$G$ and replace each vertex $v$ with all the $t$ variables $v_i$. Suppose we
have an injective mapping $f$ from $E(G)$ to the bags of the decomposition,
such that $f(e)$ contains both endpoints of $e$ (this can be obtained by
repeating bags if necessary). For each $e$ we add to $f(e)$ the $t$ variables
$e_i, i\in[t]$. It is not hard to see that this covers all clauses and all bags
contain now at most $(\pw(G)+1)t + t \le (\pw(G)+1)(\log(\pw(G))+1) +
\log(\pw(G))+1 \le \pw(G)\log(\pw(G)) + \pw(G) + 2\log(\pw(G)) + 2 \le
\pw(G)\log(\pw(G)) + 2\pw(G)$ variables, where we use the fact that $\pw(G)$ is
sufficiently large to have $\pw(G)\ge 2\log(\pw(G))+2$ (for this it suffices
that $\pw(G)\ge 8$, but we have already assumed that $\pw(G)>2^{4/\delta}$).
\end{proof}

For the converse reduction, we will go through \textsc{List Coloring}.
Similarly to the reductions of \cref{sec:coloring1}, we will use the simple
gadget that is called a weak edge (\cite{Lampis20}).  We recall again that in a
graph $G=(V,E)$ for two (not necessarily distinct) values $i_1,i_2\in[k]$, an
$(i_1,i_2)$-weak edge between two vertices $v_1,v_2\in V$ is a path consisting
of three new internal vertices linking $v_1,v_2$. Appropriate lists are
assigned to the internal vertices so that the color combination $(i_1, i_2)$ is
forbidden and all other combinations are allowed.

\begin{lemma}\label{lem:coloring2b}

If there exists $\eps>0$ and an algorithm solving \textsc{Coloring} in time
$\pw^{(1-\eps)\pw}n^{O(1)}$, then the \ppseth\ is false.

\end{lemma}

\begin{proof}

Suppose that for some $\eps\in (0,1/2)$, we have an algorithm solving
\textsc{Coloring} on graphs of $n$ vertices in time
$\pw^{(1-\eps)\pw}n^{O(1)}$. We are given a \tsat\ instance $\phi$ of pathwidth
$p=\pw(\phi)$, as well as a corresponding path decomposition. Define $s=\lceil
(1-\frac{\eps}{2})\log p\rceil$ and $k=2^s\le 2p^{1-\frac{\eps}{2}}$. We will
construct from $\phi$ in polynomial time a graph $G$ satisfying the following:

\begin{enumerate}

\item $G$ is $k$-colorable if and only if $\phi$ is satisfiable

\item $\pw(G) \le (1+\eps) \frac{p}{\log p}$

\end{enumerate}

If we have the properties above, the supposed algorithm for \textsc{Coloring}
would allow us to decide $\phi$ in time $(\pw(G))^{(1-\eps)\pw(G)}|G|^{O(1)}
\le p^{(1-\eps^2)\frac{p}{\log p}}|\phi|^{O(1)} =
2^{(1-\eps^2)p}|\phi|^{O(1)}$, which falsifies the \ppseth.

Start with a path decomposition of $\phi$ on which we have applied
\cref{lem:nice} and suppose the bags are numbered $B_1,\ldots, B_t$. Let
$g=\lceil\frac{p+1}{s}\rceil$ and note that $g\le \frac{p+1}{s}+1 \le
\frac{p}{s}+2 \le (1+2\eps/3)\frac{p}{\log p}+2$.  Here, we are using the fact
that $\frac{1}{1-\eps/2}<1+2\eps/3$ for $\eps<1/2$.  We invoke
\cref{lem:pwcolor} to partition the variables of $\phi$ into $g$ groups,
$V_1,\ldots,V_g$ so that each group contains at most $\lceil\frac{p+1}{g}\rceil
\le s$ variables of each bag.

We will describe a reduction to \textsc{List Coloring} where the list of each
vertex will be a subset of $[k]$. This can in turn easily be reduced to
$k$-\textsc{Coloring} by adding to the graph a clique of size $k$ and
connecting each vertex with the vertices of the clique that correspond to
colors missing from its list. This will add to the pathwidth of the graph at
most $k\le 2p^{1-\eps/2}$.

Let us now describe how we construct our \textsc{List Coloring} instance.
First, for each bag $B_i$ and each group $V_j, j\in [g]$ we construct a vertex
$v_{i,j}$ whose color is supposed to represent the assignment to the variables
of $B_i\cap V_j$. More precisely, we fix an arbitrary mapping $\tau_{i,j}$
which for each assignment $\sigma$ to the variables of $B_i\cap V_j$ gives a
distinct color in $[k]$; note that since $|B_i\cap V_j|\le s$ and $k=2^s$ this
is always possible. We assign the vertex $v_{i,j}$ as list the range of
$\tau_{i,j}$, that is, the set of all values $x$ such that there exists an
assignment $\sigma$ such that $\tau_{i,j}(\sigma)=x$.

We now need to add gadgets to ensure two properties: consistency and
satisfaction. For consistency, for each $i\in[t-1]$ and $j\in[g]$, consider the
sets of variables $B_i\cap V_j$ and $B_{i+1}\cap V_j$, which, since the path
decomposition is nice, have a symmetric difference that contains at most one
variable. We construct $|B_i\cap B_{i+1}\cap V_j|$ vertices, where each vertex
represents a distinct variable of $B_i\cap B_{i+1}\cap V_j$, and each vertex
has list $\{1,2\}$. For each assignment $\sigma$ to $B_i\cap V_j$ let
$\tau_{i,j}(\sigma)$ be the corresponding color of $v_{i,j}$. For each variable
$x\in B_i\cap B_{i+1}\cap V_j$ we place a weak edge between its representative
vertex and $v_{i,j}$ ensuring that if $v_{i,j}$ has color $\tau_{i,j}(\sigma)$,
then the vertex representing $x$ must take color $\sigma(x)+1$.  Similarly, we
place weak edges between $v_{i+1,j}$ to ensure that the color of $v_{i+1,j}$
encodes an assignment that is consistent with the $2$-coloring of the vertices
representing $B_i\cap B_{i+1}\cap V_j$.

Satisfaction gadgets are similar. Recall that we have an injective function $b$
such that for each clause $c$, the bag $B_{b(c)}$ contains all variables of
$c$. Consider each clause $c$ of $\phi$, and suppose that $c$ has size $3$
(repeating a literal if necessary). We construct a vertex $w_c$ and $3$
vertices $w_c^1, w_c^2, w_c^3$. The vertex $w_c$ has list $\{1,\ldots,7\}$,
informally encoding the $7$ possible satisfying assignments, while the vertices
$w_c^1, w_c^2, w_c^3$ represent the three literals of $c$ and have lists
$\{1,2\}$. We now add weak edges between $w_c$ and $w_c^\alpha$, for
$\alpha\in\{1,2,3\}$ ensuring that the assignment encoded by the color of $w_c$
is consistent with the one encoded by the colors of the three other vertices.
Furthermore, for $\alpha\in\{1,2,3\}$, if the $\alpha$-th variable of $c$
belongs in $V_{j_\alpha}$, then we add weak edges between $v_{b(c),j_\alpha}$
and $w_c^\alpha$, ensuring that the assignment of $B_{b(c)}\cap V_{j_\alpha}$
encoded by the color of $v_{b(c),j_\alpha}$ is consistent with the assignment
to the $\alpha$-th variable of $c$ encoded by $w_c^\alpha$.

This completes the construction and it is straightforward to check that the
graph has a valid list coloring if and only if $\phi$ is satisfiable. Regarding
the pathwidth of the new graph, we will treat weak edges as edges; a path
decomposition of the resulting graph can be transformed into a path
decomposition of the original graph while only increasing the width by adding
at most $2$. Now we start with the same decomposition as we have for $\phi$ and
in bag $B_i$ place the $g$ vertices $v_{i,j}$, for $j\in [g]$. If there exists
$c$ such that $b(c)=i$, then we also place in the same bag $w_c, w_c^1, w_c^2,
w_c^3$. What remains is to cover the consistency gadgets. Between bags $B_i,
B_{i+1}$ we insert a sequence of $g$ bags. In the first bag we place all
$v_{i,j}$, for $j\in[g]$ and $v_{i+1,1}$. In each subsequent bag we remove the
vertex $v_{i,j}$ with smallest $j$ that appeared in the previous bag and
replace it with $v_{i+1,j+1}$. Now, for each $j\in [g]$ we have a bag
containing both $v_{i,j}$ and $v_{i+1,j}$; we add to this bag the at most $s$
vertices representing the variables of $B_i\cap B_{i+1}\cap V_j$. It is not
hard to see that we have covered all edges and no bag contains more than
$g+s+1$ vertices.

To complete the proof, if we put back the weak edges and reduce \textsc{List
Coloring} to \textsc{Coloring} we obtain a path decomposition of the final
graph of width at most $g+k+s+4$. We have $g+k+s+4 \le
(1+\frac{2\eps}{3})\frac{p}{\log p} + 2p^{1-\frac{\eps}{2}} + \log p + 7
<(1+\eps)\frac{p}{\log p}$. To obtain the last inequality we can assume that
$p$ is sufficiently large (more precisely, that it is larger than some value
that depends only on $\eps$), as otherwise the satisfiability of $\phi$ can be
decided in linear time.  \end{proof}

\subsection{Hitting 4-cycles}

In this section we consider the following problem: we are given a graph $G$ and
are asked to find a minimum cardinality set of vertices $S$ of $G$ such that
$G-S$ contains no cycle of length $4$ ($C_4$) as a subgraph. This problem is a
particular specimen in a large class of problems where one seeks to hit all
copies of a given subgraph in a large graph, so the reader may be wondering why
we are interested in hitting $C_4$'s rather than something else. The answer is,
to some extent historical. After the discovery of the famous Cut\&Count method,
it was posed as a natural question to understand what are the problems for
which this method fails to deliver single-exponential algorithms. Within a
wider effort to trace the limits of this method (via a meta-theorem), Pilipczuk
gave a natural example of such a problem by showing that hitting all cycles of
length at least $5$ cannot be solved in time $2^{o(\pw^2)}n^{O(1)}$, under the
ETH \cite{Pilipczuk11}. Pilipczuk left it as an open problem whether the same
holds for cycles of length $4$, but this was more recently resolved in
\cite{SauS21}, where it was shown that this problem is also unsolvable in time
$2^{o(\pw^2)}n^{O(1)}$.  Therefore, hitting all $C_4$'s in a graph using the
minimum number of vertices is one of the simplest natural problems for which
the complexity is expected to be exponential in the square of the pathwidth. We
focus on this problem because we want to show that the \ppseth\ is applicable
also to FPT problems where the parameter dependence is ``unusual'' in this way.

We ask a more fine-grained question than that asked in \cite{Pilipczuk11}: what
is the smallest $c$ such that finding the minimum set hitting all $C_4$'s can
be done in time $c^{\pw^2}n^{O(1)}$? We first show that the correct value of
$c$ is not $2$, but in fact at most $\sqrt{2}$. This is because a
straightforward DP algorithm needs to store one bit of information for each
pair of vertices in the bag, which means that the dependence is at most
(roughly) $2^{\pw \choose 2}$ (\cref{thm:c4alg}).

Our main result is then to show that this complexity is very likely to be
correct: obtaining an algorithm with running time
$(\sqrt{2}-\eps)^{\pw^2}n^{O(1)}$ is \emph{equivalent} to falsifying the
\ppseth. Formally, we prove the following:

\begin{theorem}\label{thm:c4} There exists $\eps>0$ and an algorithm that takes
as input an $n$-vertex graph $G$ and a path decomposition of width $p$ and
computes a minimum-size set $S$ such that $G-S$ has no $C_4$ subgraph in time
$(\sqrt{2}-\eps)^{\pw^2}n^{O(1)}$ if and only if the \ppseth\ is false.
\end{theorem}

The proof of \cref{thm:c4} is given in two separate lemmas (\cref{lem:c4-1} and
\cref{lem:c4-2}).  We believe that these results make the problem of hitting
all $C_4$'s an interesting specimen, because not only are we able to determine
that the correct dependence is exponential in $\Theta(\pw^2)$ (as was done in
\cite{SauS21}), but in fact we determine precisely the coefficient of the
leading term, and this coefficient is equal to $\frac{1}{2}$. In particular, we
give an explicit algorithm with complexity that exactly matches our lower bound
(\cref{thm:c4alg}). 

Compared to previous reductions, we have to overcome several technical
obstacles to achieve our results.  The most basic one is that, because our
reduction needs to be much more efficient, we cannot afford to encode a
\textsc{SAT} instance by using one gadget for a variable $x$ and a different
gadget for its negation $\neg x$, as is done in \cite{SauS21} and
\cite{Pilipczuk11}. We overcome this difficulty by using an intermediate
version of \textsc{SAT} where we are only looking for assignments that set
exactly half of the variables to True, in a certain controlled way.

We begin with the algorithmic part.

\begin{theorem}\label{thm:c4alg} There is an algorithm that takes as input a
graph $G$ on $n$ vertices  and a path decomposition of width $p$ and outputs a
minimum set  of vertices $S$ such that $G-S$ contains no $C_4$ subgraph in time
$2^{\frac{p^2}{2}+O(p)}n^{O(1)}$. \end{theorem}

\begin{proof}

The algorithm uses standard dynamic programming, so we sketch some of the
details. We are looking for a deletion set $S$ such that $G-S$ contains no
$C_4$ subgraph and $S$ has minimum size. For each bag $B$ of a path
decomposition let $B^{\leftarrow}$ be the set of vertices of $G$ which appear
in $B$ or in a bag to the left of $B$ (that is, a bag already processed by the
algorithm). We define the signature of a partial solution as a pair consisting
of two parts: (i) a subset of $B$, meant to represent $S\cap B$ (ii) a bit for
each pair of distinct vertices $u,v\in B\setminus S$, meant to indicate whether
$u,v$ have a common neighbor in $G-S$ which appears in $B^{\leftarrow}\setminus
B$.  For each signature we store the minimum size of a set $S\subseteq
B^{\leftarrow}$ that respects the signature and that hits all $C_4$'s induced
by $G[B^{\leftarrow}]$.

It is now not hard to update the DP table as we process the path decomposition
of $G$. When introducing a new vertex $u$, we consider the cases $u\in S$ and
$u\not\in S$. In the former case, we take every valid signature of the previous
bag and increase the cost by $1$. In the latter case, we retain all signatures
with the same cost, but discard signatures which have one of the following: $u$
is adjacent to two vertices $v_1,v_2\in B\setminus S$ which already have a
common neighbor in $B^{\leftarrow}\setminus (B\cup S)$; $u$ is adjacent to two
vertices $v_1,v_2\in B\setminus S$ which have a common neighbor $v_3\in
B\setminus S$. For all pairs involving the new vertex we set the corresponding
bit to $0$. In both of these cases, discarding these signatures is justified as
we have a $C_4$ in $G[B^{\leftarrow}\setminus S]$.  When forgetting a vertex
$u$, for each signature where $u\in S$ we produce a new signature where we have
removed the information about $u$; while for each signature with $u\not\in S$
we produce a signature where we set to $1$ the bit for each pair $v_1,v_2\in
B\setminus S$ if $u$ is adjacent to $v_1,v_2$.  This algorithm will always
produce a valid deletion set, because if there was a $C_4$ in $G-S$, we can
consider its last vertex to be introduced, say $u$, which is adjacent to
$v_1,v_2$ in the $C_4$. By the properties of path decompositions, $v_1,v_2$
must be in the bag where $u$ is introduced. Then, if the fourth vertex of the
cycle is not in $B$ we must have set the bit for the pair $v_1,v_2$ to $1$
(hence we discard this solution); while if the fourth vertex is in $B$ we will
again discard the solution.

To compute the running time we mainly need to bound the size of the DP table,
which is at most $2^{{p+1\choose 2}+p}=2^{\frac{p^2}{2}+O(p)}$. Since we spend
polynomial time per DP entry, the total running time is then
$2^{\frac{p^2}{2}+O(p)}n^{O(1)}$.  \end{proof}

\begin{lemma}\label{lem:c4-1} If the \ppseth\ is false, then there exists an
$\eps>0$ and an algorithm that takes as input a graph $G$ and a path
decomposition of $G$ of width $p$ and solves $C_4$-\textsc{Hitting Set} in time
$2^{(1-\eps)\frac{p^2}{2}}n^{O(1)}$.  \end{lemma}

\begin{proof} The idea of this proof is to encode the workings of the algorithm
of \cref{thm:c4alg} as a \textsc{MaxW-SAT} instance and then find a satisfying
assignment of maximum weight. Suppose that there is an algorithm for
\textsc{MaxW-SAT} running in time $2^{(1-\delta)\pw}|\phi|^{O(1)}$, where
$\phi$ is the input formula.  We are given a graph $G$ on $n$ vertices,
together with a path decomposition of width $p$. We will construct a formula
$\phi$ satisfying the following:

\begin{enumerate}

\item $\phi$ has size $n^{O(1)}$, and its variables are partitioned into two
sets $V_0,V_1$. 

\item $\phi$ can be constructed in time $n^{O(1)}$ and a path decomposition of
$\phi$ of width $\frac{p^2}{2}+O(p)$ can be constructed in the same time.

\item For all $k\ge0$, there is a satisfying assignment setting at least $k$ of
the variables of $V_0$ to True if and only if there is a deletion set $S$ such
that $G-S$ has no $C_4$ subgraph and $|V(G)\setminus S|\ge k$.

\end{enumerate}

We first note that we are using a slightly more general version of
\textsc{MaxW-SAT}, because we partition the variables of $\phi$ into two sets,
but only care about the weight of the assignment on the variables of $V_0$.
This does not change much, though, because we can take our $\phi$, and assuming
$\phi$ has $N$ variables, introduce for each $x\in V_0$, $N^2$ copies, $x_i,
i\in[N^2]$, with the clauses $(\neg x\lor x_i)\land(\neg x_i\lor x)$, which
ensures that the copies take the same value as $x$. Solving standard
\textsc{MaxW-SAT} on the new instance is equivalent to solving our instance
where we only care about the weight in $V_0$. The new variables only increase
the pathwidth by at most $1$, since we can find a distinct bag for each $x\in
V_0$ and insert after it $N^2$ copies, inserting into each copy a distinct
$x_i$.

Now, let us explain why if we achieve the above, then we obtain the lemma.
Suppose that the decomposition of $\phi$ we produce has width at most
$\frac{p^2}{2}+cp$, for some constant $c$. We assume without loss of generality
that $p$ is sufficiently large to have
$(1-\delta)(\frac{p^2}{2}+cp)<(1-\frac{\delta}{2})\frac{p^2}{2}$ (this happens
when $p>\frac{4c}{\delta}$). This assumption is without loss of generality,
because if $p$ is bounded by a constant we can solve the problem outright in
polynomial time. But then, the $2^{(1-\delta)\pw}|\phi|^{O(1)}$ algorithm for
\textsc{MaxW-SAT} gives an algorithm for our problem with running time at most
$2^{(1-\frac{\delta}{2})\frac{p^2}{2}}n^{O(1)}$, as required. 

Suppose that the bags of the path decomposition of $G$ are numbered
$B_1,\ldots,B_t$, and assume without loss of generality that $t\le 2n$, because
each bag introduces or forgets a vertex (otherwise the bag is redundant). We
construct for each $v\in V(G)$ a variable $x_v\in V_0$. Furthermore, for each
$i\in[t]$ and for each pair of distinct $v_1,v_2\in B_i$ we construct a
variable $x_{v_1,v_2,i}\in V_1$. The total number of variables is at most
$n+t{p+1\choose 2} = O(n^3)$.

The intuitive meaning of the variables $x_v$ is that $x_v$ is True if and only
if $v\in V\setminus S$. The meaning of $x_{v_1,v_2,i}$ is that the variable is
set to True if $v_1,v_2$ have a common neighbor in
$\bigcup_{i'\in[i-1]}B_{i'}\setminus S$. We now describe the clauses of $\phi$. 

\begin{enumerate}

\item For each $v\in V(G)$ let $B_i$ be the last bag that contains $v$. Then we
add the clause $(\neg x_v \lor x_{v_1,v_2,i+1})$ for all pairs of distinct
$v_1,v_2\in B_i$ for which $v_1,v_2\in N(v)\cap B_{i+1}$.

\item For each $i\in[t-1]$ and pair of distinct $v_1,v_2\in B_i\cap B_{i+1}$ we
add the clause $(\neg x_{v_1,v_2,i} \lor x_{v_1,v_2,i+1})$.

\item For each $C_4$ of $G$ which has all its vertices, say $v_1,v_2,v_3,v_4$,
contained in a bag $B_i$, we construct the clause $(\neg x_{v_1}\lor \neg
x_{v_2} \lor \neg x_{v_3} \lor \neg x_{v_4})$. 

\item For each $C_4$, say $v_1,v_2,v_3,v_4$, that is not fully contained in any
bag, suppose without loss of generality that the last vertex to be introduced
in the decomposition is $v_2$, which first appears in bag $B_i$ (so
$v_1,v_3,v_4$ all appear in bags $B_i'$ with $i'\le i$). By the properties of
path decompositions, $B_i$ must also contain $v_1,v_3$. We construct the clause
$(\neg x_{v_1}\lor \neg x_{v_2} \lor \neg x_{v_3} \lor \neg x_{v_1,v_3,i})$.

\end{enumerate} 

We now show how a path decomposition of the primal graph of $\phi$ of width
$\frac{p^2}{2}+O(p)$ can be constructed from the decomposition of $G$. We start
with the decomposition of $G$, replace each $v\in V(G)$ in each bag with
$x_v\in V_0$ and add inside each bag $B_i$ all the variables $x_{v_1,v_2,i}\in
V_1$.  Furthermore, for each $i\in[t-1]$, add all variables of $V_0$ which are
in $B_i$ also to bag $B_{i+1}$. Each bag now contains at most $2p+2$ variables
of $V_0$. We observe that we have covered all clauses of steps 1,3, and 4. To
cover the clauses of step 2, we observe that for each $i\in[t-1]$ these edges
form a matching between variable contained in $B_i$ and variables contained in
$B_{i+1}$. We insert a sequence of bags between $B_i$ and $B_{i+1}$, where we
place into all copies all variables of $B_i\cap V_0$. Consider now the (at most
$p+1\choose 2$) variables of $B_i\cap V_1$. We order the pairs $v_1,v_2\in B_i$
in some arbitrary way, add in the first new bag all variables $x_{v_1,v_2,i}$
and the variables $x_{v_1,v_2,i+1}$ for the first pair $v_1,v_2$. Then we
gradually swap variables $x_{v_1,v_2,i}$ with variables $x_{v_1,v_2,i+1}$
covering all edges of the matching. The width of the new decomposition is
therefore ${p+1\choose 2}+O(p) = \frac{p^2}{2}+O(p)$. 

To show correctness, first suppose that a satisfying assignment to $\phi$ sets
at least $k$ of the variables of $V_0$ to True. We construct a deletion set $S$
by setting $v\in S$ if and only if $x_v$ is False in the satisfying assignment.
We claim $S$ hits all $C_4$'s. $C_4$'s which are fully contained in a bag are
hit due to clauses of step 3. For a $C_4$, say $v_1,v_2,v_3,v_4$ that is not
contained in any bag, suppose without loss of generality that $v_2$ is the last
vertex to be introduced in the decomposition, in bag $B_i$, which also contains
$v_1,v_3$. We claim that if $v_1,v_3,v_4\not\in S$ (so
$x_{v_1},x_{v_3},x_{v_4}$ are True), then the satisfying assignment must have
$x_{v_1,v_3,i}$ set to True, therefore, $x_{v_2}$ set to False (due to the
clause of step 4), therefore $v_2\in S$. To see that $x_{v_1,v_3,i}$ must be
True, consider the last bag $B_{i'}$ that contains $v_4$ (we have $i'<i$).
$B_{i'}$ must also contain $v_1,v_3$, therefore, since $v_4\not\in S$ implies
$x_{v_4}$ is True, the clause of step 1 implies that $x_{v_1,v_2,i'+1}$ is
True. Then, the clauses of step 2 imply that $x_{v_1,v_2,i''}$ is True for all
$i''\in\{i'+1,\ldots,i\}$.  We conclude that $S$ is a valid hitting set.

For the converse direction, suppose that $S$ is a valid hitting set and we
obtain a satisfying assignment by setting $x_v$ to True whenever $v\in
V\setminus S$. We claim that this can be extended to a satisfying assignment
for $\phi$. For each variable $x_{v_1,v_2,i}$ we give it value True if and only
if there exists $v\in \bigcup_{i'\in[i-1]} B_{i'}$ such that $v$ is adjacent to
both $v_1,v_2$ and $v\not\in S$ (therefore $x_v$ is True). This satisfies
clauses of step 2 (as our assignment is non-decreasing) and step 1 (such
clauses are trivially satisfied if $x_v$ is False). Clauses of step 3 are
satisfied by the assumption that $S$ is a hitting set for all $C_4$'s. For
clauses of step 4, suppose $x_{v_1},x_{v_2},x_{v_3}$ are True (otherwise the
clause is satisfied), and $x_{v_1,v_3,i}$ is also True. However, if
$x_{v_1,v_3,i}$ is True, this means there exists
$v_4\in\bigcup_{i'\in[i-1]}B_{i'}$ such that $v_4$ is adjacent to $v_1,v_2$ and
$v_4\not\in S$. But then we have $v_1,v_2,v_3,v_4\not\in S$ and we have a $C_4$
that is not hit, which gives a contradiction. Hence, $\phi$ is satisfied by
this assignment. \end{proof}

We now move on to the main result of this section, which is a reduction from
\textsc{SAT} to the problem of hitting all $C_4$'s that establishes a tight
lower bound on the parameter dependence on pathwidth. As mentioned, our
reduction needs to be significantly more efficient than previous reductions,
because we want to reduce a \textsc{SAT} instance of pathwidth $p$ to a graph
of pathwidth $\sqrt{2p}+o(\sqrt{p})$. In particular, this means that every pair
of vertices in a bag will represent a single variable, meaning that we cannot
afford to use a distinct gadget for $x$ and for $\neg x$, as is done in
previous reductions. To work around this difficulty, we will need an
intermediate satisfiability problem, where the variables are given to us
partitioned into groups, and we are only interested in satisfying assignments
that set half of the variables of each group to True. 

\begin{lemma}\label{lem:satforc4} There exists an algorithm that takes as input
a \tsat\ formula $\phi$ on $n$ variables and a path decomposition of its primal
graph where each bag contains at most $p$ variables and outputs a \textsc{SAT}
formula $\psi$ and a path decomposition of its primal graph such that we have
the following (where $q=\lceil \sqrt{2p}\ \rceil$, $r=\lceil \sqrt{q}\
\rceil$):

\begin{enumerate}

\item The algorithm runs in time $2^{O(\sqrt[4]{p})}n^{O(1)}$.

\item The algorithm also outputs a partition of the variables of $\psi$ into $
qr +1$ sets $V_0$ and $V_{i_1,i_2}$ where $i_1\in [q], i_2\in [r]$.

\item Each bag of the decomposition of $\psi$ contains at most $r = O(p^{1/4})$
variables of $V_0$

\item For each $i_1\in[q]$, $i_2\in[r]$, and each bag $B$ of the decomposition
of $\psi$ we have that $|B\cap V_{i_1,i_2}|\le\frac{q-i_1}{r} + 6\sqrt{r}$ and
$|B\cap V_{i_1,i_2}|$ is even. Hence, for each $i_1\in[q]$, each bag contains
at most $q-i_1+6r\sqrt{r}$ variables of $\bigcup_{i_2\in[r]} V_{i_1,i_2}$.

\item For each $i_1\in[q]$, $i_2\in[r]$ and consecutive bags $B, B'$ of the
decomposition of $\psi$ we have either $B\cap V_{i_1,i_2}=B'\cap V_{i_1,i_2}$
or $B\cap V_{i_1,i_2}\cap B'=\emptyset$. Furthermore, for each pair $B,B'$ of
consecutive bags, there exists at most one pair $i_1,i_2$ such that $B\cap
V_{i_1,i_2}\cap B'=\emptyset$.

\item Every clause of $\psi$ contains at most $O(r)$ literals.

\item $\phi$ is satisfiable if and only if there exists a satisfying assignment
of $\psi$ which also has the property that for each $i_1\in [q]$, $i_2\in[r]$
and each bag $B$ of the decomposition exactly $\frac{|B\cap V_{i_1,i_2}|}{2}$
variables of $B\cap V_{i_1,i_2}$ are set to True.

\end{enumerate}

\end{lemma}

\begin{proof}

We begin with some basic preprocessing of the given path decomposition. First,
assume that all bags contain exactly $p$ variables (this can be achieved by
adding dummy fresh variables in smaller bags). Then, assume that the
decomposition is nice, that is, for each consecutive bags $B,B'$ we have either
$B=B'$, or $B'$ is obtained from $B$ by removing a variable and replacing it
with a new one. This can be achieved by inserting between bags whose symmetric
difference is larger than two a sequence of bags that gradually swap vertices
one by one. Finally, assume we have an injective function that for each clause
of $\phi$,  maps the clause to a bag $B$ that contains all its variables, such
that the bags adjacent to $B$ in the decomposition contain the same variables
as $B$ (this can be achieved by repeating bags if necessary). All these steps
can be carried out in polynomial time.

We now group the variables of $\phi$ into $q$ sets, $W_1,\ldots, W_{q}$ so that
each set $W_{i_1}$, for $i_1\in[q]$, contains at most $q-{i_1}+1$ variables
from each bag.  This is done as follows: In the first bag, place for each
$i_1\in[q]$, $q-{i_1}+1$ arbitrarily chosen variables of $B$ into each
$W_{i_1}$.  Observe that $\sum_{i_1\in[q]}(q-i_1+1) = \frac{q(q+1)}{2} \ge
\frac{\sqrt{2p}(\sqrt{2p}+1)}{2}\ge p$, so this process partitions all
variables of the first bag without placing too many variables in any $W_{i_1}$.
Now we move along the decomposition and for each bag that is identical to the
previous bag we do nothing; while for a bag $B'$ that is after bag $B$ such
that $B'\setminus B =\{x'\}$ and $B\setminus B' = \{x\}$, we place $x'$ in the
same set $W_{i_1}$ that contains $x$. We now have that for each bag $B$,
$i_1\in[q]$, $|B\cap W_{i_1}|\le q-i_1+1$. Furthermore, for each $i_1\in[q]$
and two bags $B,B'$, we have $|B\cap W_{i_1}|=|B'\cap W_{i_1}|$.

We now further break down the groups as follows: for each $i_1\in [q]$, we want
to partition $W_{i_1}$ into $r$ groups $W_{i_1,i_2}$, $i_2\in [r]$, so that for
each bag $B$ we have $|W_{i_1,i_2}\cap B| \le \lceil \frac{q-i_1+1}{r} \rceil$.
We achieve this in a similar way as previously: In the first bag, we
arbitrarily partition the at most $q-i_1+1$ variables of $W_{i_1}$ into the $r$
sets $W_{i_1,i_2}$, for $i_2\in [r]$, so that each set contains at most
$\lceil\frac{q-i_1+1}{r}\rceil$ variables. Then, for each bag $B'$ that
contains a variable $x'$ which replaces the variable $x$ of the preceding bag
$B$, we place $x'$ in the same set $W_{i_1,i_2}$ as $x$. We now have that for
each bag $B$, $i_1\in[q]$, $i_2\in[r]$, $|B\cap W_{i_1,i_2}|\le \lceil
\frac{q-i_1+1}{r}\rceil$.

We are now ready to construct $\psi$ and we will also along the way construct a
path decomposition by editing the decomposition of $\phi$. For the first bag
$B$ of the decomposition of $\phi$, for each $i_1\in[q]$, $i_2\in[r]$, we
construct a set of variables $V_{i_1,i_2}$ of size
$\lceil\frac{q-i_1+1}{r}\rceil + 2\lceil \sqrt{r}\ \rceil = O(r)$. The informal
meaning is that the truth assignments of $W_{i_1,i_2}\cap B$ should be
injectively mapped into truth assignments to $V_{i_1,i_2}\cap B$ that set
exactly $\lfloor\frac{|B\cap  V_{i_1,i_2}|}{2}\rfloor$ variables to True. We
claim that this is always possible:

\begin{claim}\label{claim:many}

For all $i_1\in[q]$ we have  ${ \lceil\frac{q-i_1+1}{r}\rceil +
2\lceil\sqrt{r}\rceil \choose \lfloor \frac{\lceil\frac{q-i_1+1}{r}\rceil +
2\lceil\sqrt{r}\rceil}{2} \rfloor }\ge 2^{\lceil\frac{q-i_1+1}{r}\rceil}$.

\end{claim}

\begin{proof} Let $x= \lceil\frac{q-i_1+1}{r}\rceil$, $y=\lceil \sqrt{r}\
\rceil$, and we note that $x\le \lceil \frac{q}{r}\rceil \le r \le y^2$.  We
will show that for all positive integers $x,y$, such that $x\le y^2$ we have
${x+2y \choose \lfloor \frac{x+2y}{2}\rfloor} \ge 2^x$. We have that for all
$z\in\{0,1,\ldots,x+2y\}$, ${x+2y\choose \lfloor \frac{x+2y}{2}\rfloor} \ge
{x+2y\choose z}$, and that $\sum_{z=0}^{x+2y} {x+2y\choose z} = 2^{x+2y}$.
Therefore, ${x+2y\choose \lfloor \frac{x+2y}{2}\rfloor} \ge
\frac{2^{x+2y}}{x+2y+1}\ge 2^x (\frac{2^{2y}}{y^2+2y+1})$.  To obtain the claim
it is therefore sufficient to have $2^{2y}\ge y^2+2y+1$, which is true for all
$y\ge 2$. For $y=1$, since $x\le y^2$ we must have $x=1$, and the inequality
trivially holds.  \end{proof}

Armed with \cref{claim:many} we do the following for each $i_1\in[q],
i_2\in[r]$ starting from the first bag $B$: place all the (new) variables of
$V_{i_1,i_2}$ into $B$ and fix an injective mapping from each truth assignment
to $B\cap W_{i_1,i_2}$ to a truth assignment of $B\cap V_{i_1,i_2}$ that sets
exactly $\lfloor \frac{|B\cap V_{i_1,i_2}|}{2} \rfloor$ variables to True; for
each assignment to $B\cap V_{i_1,i_2}$ that is not in the range of this mapping
construct a clause in $\psi$ that is falsified by this assignment. If $|B\cap
V_{i_1,i_2}|$ is odd, add a new (dummy) variable to $B\cap V_{i_1,i_2}$ and a
clause that forces this variable to be True. Remove from $B$ the variables of
$W_{i_1,i_2}$ (since they do not appear in $\psi$).

We now continue to the next bag of the decomposition of $\phi$. For each bag
$B'$, let $B$ be the previous bag, and for each $i_1\in[q], i_2\in[p]$ we do
the following: if $B\cap W_{i_1,i_2} = B'\cap W_{i_1,i_2}$, then we do nothing;
otherwise we construct $|B\cap V_{i_1,i_2}|\le\lceil \frac{q-i_1+1}{r} \rceil +
2\lceil\sqrt{r}\ \rceil+1$ new variables and place them into $B'\cap
V_{i_1,i_2}$. Again, we construct an arbitrary injective mapping from
assignments to $B'\cap W_{i_1,i_2}$ to assignments that set exactly half the
variables of $B'\cap V_{i_1,i_2}$ to True (observe that we now know that
$|B'\cap V_{i_1,i_2}|$ is even) and add clauses ensuring that we must select an
assignment in the range of the mapping to satisfy $\psi$. We now need to ensure
consistency, meaning that the assignment we select for $B'\cap V_{i_1,i_2}$ and
the assignment we select for $B\cap V_{i_1,i_2}$ agree on the variables of
$B\cap B'\cap W_{i_1,i_2}$.  In order to achieve this, we make a fresh copy of
each variable of $W_{i_1,i_2}$ and add it to $V_0\cap B\cap B'$.  For each such
variable $x$ and each assignment $\sigma$ to $B\cap V_{i_1,i_2}$ that is in the
range of our injective mapping, there is a unique assignment $\sigma'$ to
$W_{i_1,i_2}$ which maps to $\sigma$. We add a clause that ensures that if we
set the variables of $B\cap V_{i_1,i_2}$ to $\sigma$ and $x$ to $\neg
\sigma'(x)$, $\psi$ is falsified. This ensures that the assignment we give to
the new variable $x$ is consistent with what the selected assignment to $B\cap
V_{i_1,i_2}$ encodes. We add similar clauses for $B'\cap V_{i_1,i_2}$.  Observe
that if we add the (at most $r$) new variables of $V_0$ to $B, B'$, both bags
satisfy condition 3, since there is only one pair $i_1,i_2$ for which they
differ. We continue in this way, until we have processed the whole
decomposition, and it is not hard to see that we have a decomposition of the
formula $\psi$ which we have constructed so far that satisfies conditions 4 and
5.

What remains is to add some clauses to $\psi$ so that only encodings of
satisfying assignments to $\phi$ can satisfy $\psi$. To do this, for each
clause of $\phi$ we find a bag $B$ that contains all its literals using the
injective mapping we assumed previously. We add fresh copies of the three
variables of the clause into $B\cap V_0$. Note that $B$ currently contains no
other variable of $V_0$, since it is the same as its neighboring bags. We use
these fresh variables to make a copy of the original clause. Now, to ensure
that the assignment to the new variables is consistent with the assignments to
$B\cap V_{i_1,i_2}$, we decode the assignment as in the previous paragraph. In
other words, for each $x\in V_0\cap B$ we added, if $x\in W_{i_1,i_2}\cap B$,
for each assignment $\sigma$ to $V_{i_1,i_2}\cap B$ that is in the range of the
injective mapping we have used, there is a unique assignment $\sigma'$ to
$W_{i_1,i_2}\cap B$ that maps to $\sigma$. We add a clause that is falsified if
we use simutaneously the assignment $\sigma$ for $V_{i_1,i_2}\cap B$ and
$\neg\sigma'(x)$ for $x$.

This completes the construction and it is not hard to see that $\psi$ is
satisfied only by assignments that set half of each $B\cap V_{i_1,i_2}$ to
True. Furthermore, such a satisfying assignment can be found if and only if
$\phi$ has a satisfying assignment. The construction can be carried out in
polynomial time, except for the part where we have to consider all assignments
to each $W_{i_1,i_2}$, but since each such set has size at most $O(r)$, we get
the promised running time.  \end{proof} 

Before we proceed, let us give some more intuition about why a problem such as
the one given by \cref{lem:satforc4} will be useful. A straightforward attempt
to reduce \textsc{SAT} to hitting $C_4$'s would, for each bag of the
decomposition of the formula containing $p$ variables, construct $\sqrt{2p}$
vertices. We would then place, for each pair of such vertices, a vertex of
degree $2$ connected to this pair. The idea would be that deleting this vertex
or not would represent the value of one variable, and we would have exactly as
many pairs as we need to represent all the variables. Before adding gadgets to
implement the rest of the reduction, though, it is important to observe an
important flaw in this strategy: one could always simply choose to delete the
$\sqrt{2p}$ vertices, rather than any of the vertices representing variables.
Since $\sqrt{2p}$ is significantly smaller than $p$, this looks like a serious
problem. The way this is handled in previous reductions is by constructing a
degree $2$ vertex for each \emph{literal}, rather than variable, and then
adding a $C_4$ that ensures one of the two literal vertices must be deleted.
This is simple, but will not work in our case, because we cannot afford to
waste a factor of $2$ in the efficiency of our reduction. 

The strategy we will use to overcome the obstacle outlined above is the
following: suppose we have some (large but not too large) groups of variables
for which we know the weight $w$ of a satisfying assignment. We can then add an
appropriate gadget to ensure that in any feasible solution at least $w$ of the
degree $2$ vertices representing variables from this group must be deleted. By
placing such gadgets for all groups of variables (since in the problem of
\cref{lem:satforc4} we know the desired weight for all groups), we can adjust
the total deletion budget so that it is exactly equal to the sum of the weights
dictated locally in each group. This ensures that there is never any budget
left to delete the $\sqrt{2p}$ vertices of a bag that ensure connectivity, and
the reduction works as intended.

With this intuition in mind, let us know define the budget-setting gadget we
will use in our reduction. We first mention a more basic gadget which we call a
\emph{strong edge}. A strong edge with endpoints $u,v$ is made up of four new
vertices $x_1,x_2,x_3,x_4$ which are adjacent to $u,v$ and not adjacent to any
other vertex. Informally, a strong edge ensures that at least one of $u,v$ must
be in the deletion set. Indeed, if a solution does not delete neither $u$ nor
$v$ it must delete at least three of the new vertices; however, in that case
exchanging these vertices with $u,v$ gives a better solution.

\begin{definition} Let $G=(V,E)$ be a graph, $X\subseteq V$ a set of vertices,
and $b\le |X|$ a positive integer. A budget-setter for $X$ with budget $b$ is
the following gadget: let $f={|X|\choose b}$ and we construct $3f$ sets of
$2|X|$ vertices each, call them $A_i, B_i, C_i$ for $i\in[f]$. For each
$i\in[f]$ we place strong edges between all vertices of $A_i$ and $B_i\cup
C_i$, and between all vertices of $B_i$ and $C_i$. For each $i\in[f-1]$ we
place strong edges between all vertices of $A_i$ and all vertices $A_{i+1}$,
and strong edges between all vertices of $B_i$ and all vertices of $B_{i+1}$.
If $f$ is odd, we place strong edges between all vertices of $A_f$ and all
vertices of $A_1$; and all vertices of $B_f$ and all vertices of $B_1$.  If $f$
is even, we place strong edges between all vertices of $A_f$ and $B_1$; and
strong edges between all vertices of $B_f$ and $A_1$. Finally, for each $i\in
[f]$ we place strong edges between all vertices of $C_i$ and all vertices of a
distinct subset of $X$ of size $b$. \end{definition}

We now show that the budget-setter has the required effect, namely that we
force the solution to delete at least $b$ of the vertices of $X$, and that if
we use any set of $b$ vertices the cost of the budget-setter is the same.

\begin{lemma}\label{lem:budget-setter} Let $G=(V,E)$ be a graph, $X$ a set of
vertices, and $b\le |X|$ a positive integer. Let $G'$ be the graph obtained by
adding a budget-setter for $X$ with budget $b$ to $G$. Then, we have the
following (where $f={|X|\choose b}$):

\begin{enumerate}

\item The optimal $C_4$-hitting set of $G'$ contains at least $b$ vertices of
$X$ and at least $4f|X|$ vertices of the budget-setter.

\item Any $C_4$-hitting set of $G$ of size $k$ that includes at least $b$
vertices of $X$ can be extended to a hitting set of $G'$ of size $k+4f|X|$.

\end{enumerate}

\end{lemma}

\begin{proof}

For the first statement, it is not hard to see that for each $i\in[f]$ a
$C_4$-hitting set of $G'$ must completely contain at least two of the sets
$A_i, B_i, C_i$, because we have added strong edges between their vertices
forming, essentially, a complete tripartite graph. This means that the hitting
set must have size at least $4f|X|$, since each set $A_i, B_i, C_i$ has size
$2|X|$. Suppose now that a hitting set $S$ contains at most $b-1$ vertices of
$X$. For all $i\in [f]$ we have that the neighborhood of $C_i$ cannot be
contained in $S\cap X$, as each $C_i$ is connected (via strong edges) to $b$
vertices of $X$. Then, for all $i\in[f]$ we have $C_i\subseteq S$. We claim
that in this case we are including at least $2(2f+1)|X|$ vertices of the
budget-setter in $S$. To see this, consider an auxiliary graph that contains a
vertex for each $A_i, B_i$ and an edge when two sets are completely joined via
strong edges. By construction, this graph has maximum independent set $f-1$.
Therefore, in the budget-setter there are at most $f-1$ sets among the
$A_i,B_i$ which contain a vertex outside $S$, hence at least $2f+1$ sets which
are contained in $S$. We now observe that the current solution can be improved
by placing all of $X$ in the hitting set, and all the sets $A_i,B_i$ for
$i\in[f]$. The previous solution was deleting at least $2(2f+1)|X|$ vertices,
while the new solution is deleting $4f|X|+|X|$ vertices and satisfies the first
condition of the lemma.

For the second condition, suppose that we have a hitting set $S$ of $G$ such
that $|S\cap X|\ge b$. Find a $C_i\in[f]$ such that $C_i$ is connected via
strong edges only to vertices of $S\cap X$. We place all of $A_i,B_i$ in $S$.
We also place all $C_{i'}$ for $i'\neq i$ into $S$ and for each $i'\neq i$
place one of $A_i,B_i$ completely inside $S$. This is always possible, because
if we consider the auxiliary graph formed by the $A_{i'},B_{i'}$ for $i'\neq i$
it is a $2\times(f-1)$ grid, which has an independent set of size $f-1$. The
hitting set of $G'$ we have thus constructed has size $k+4f|X|$. Since for any
strong edge we have deleted one of its two endpoints no $C_4$ can go through
the vertices of the budget-setter, so we have a feasible solution.  \end{proof}

\begin{lemma}\label{lem:budget-setter2} Let $G=(V,E)$ be an $n$-vertex graph,
$X_1,\ldots, X_s$ be disjoint sets of vertices, $b_1,\ldots,b_s$ be positive
integers with $b_i\le |X_i|$ for all $i\in[s]$. Let $x=\max_{i\in[s]}|X_i|$ and
suppose we have a path decomposition of $G$ of width $p$ such that for all
$i\in[s]$ there is a bag that contains all of $X_i$. Then,  if $G'$ is the
graph obtained by adding a budget-setter for each $X_i$ with budget $b_i$ to
$G$, for $i\in[s]$, we can construct $G'$ in time $2^xn^{O(1)}$ and in the same
time edit the decomposition of $G$ to obtain a decomposition of $G'$ of width
$p+O(x)$.  \end{lemma}

\begin{proof} Constructing the graph $G'$ in the stated time is
straightforward, as ${|X_i|\choose b_i}\le 2^x$ for all $i$. We therefore focus
on the construction of the path decomposition. For each $i\in[s]$ find a
distinct bag $B_i$ that contains all of $X_i$ (this can be achieved by copying
bags if necessary). We observe that the graph induced by the budget-setter
vertices for $X_i$ has pathwidth $O(|X_i|) = O(x)$. We take such a path
decomposition and add $B_i$ to all its bags, then insert the whole
decomposition after $B_i$ in our decomposition of $G'$. Repeating this for all
$i\in[s]$ concludes the lemma. \end{proof}

\begin{figure}

\centering

\includegraphics[width=\textwidth]{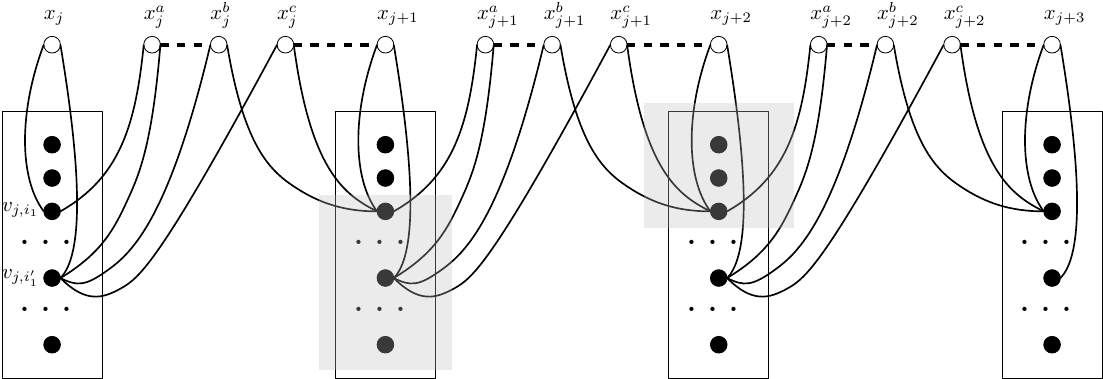}\caption{The consistency
gadget used in the reduction of \cref{lem:c4-2}. The four boxes contain the
backbone vertices $v_{j,i_1}$ and we concentrate on a specific variable $x$,
which is represented by the vertices $x_j, x_{j+1}, x_{j+2}, x_{j+3}$.  Dashed
lines represent strong edges. The intended assignment is that either we place
in the hitting set $x_j, x_j^b, x_{j+1}, x_{j+1}^b, x_{j+2}, \ldots$, or the
vertices $x_j^a, x_j^c, x_{j+1}^a, x_{j+1}^c, \ldots$. This is achieved because
we give the construction sufficient budget to delete half of the vertices of
the top row, and for any two consecutive vertices we have to delete at least
one (as they have either a strong edge or two common neighbors).  To see that
the construction has the claimed pathwidth consider the sequence of bags that
resemble the shaded area, that is, take $v_{j,i_1}, v_{j,i_1+1},\ldots$ and
also $v_{j+1,1},\ldots,v_{j+1,i_1}$. For each strong edge in the top row, there
exists such a bag such that the endpoints of the strong edge have all their
neighbors in the bag.}\label{fig:c4red}

\end{figure}

\begin{lemma}\label{lem:c4-2} If there exists an $\eps>0$ and an algorithm
solving $C_4$-\textsc{Hitting Set} in time
$2^{(1-\eps)\frac{\pw^2}{2}}n^{O(1)}$, then the \ppseth\ is false.  \end{lemma}

\begin{proof}

Suppose we are given a \tsat\ instance $\phi$ and a path decomposition of its
primal graph of width $p$. We execute the algorithm of \cref{lem:satforc4} so
now we have an equivalent CNF instance $\psi$ with $q=\lceil\sqrt{2p}\ \rceil$,
$r=\lceil\sqrt{q}\ \rceil$, the variables of $\psi$ partitioned into sets $V_0$
and $V_{i_1,i_2}$ for $i_1\in[q]$ and $i_2\in[r]$, and a path decomposition of
its primal graph with the requirements stipulated in \cref{lem:satforc4}. We
would like to construct a graph $G$ and a target budget $k$ such that:

\begin{itemize}

\item $G$ has a $C_4$-hitting set of size at most $k$ if and only if $\psi$ has
a satisfying assignment such that for all bags of the given decomposition and
for each $i_1\in [q], i_2\in[r]$ exactly $\frac{ |B\cap V_{i_1,i_2}|}{2}$
variables are set to True.

\item $G$ can be constructed in time $|\psi|^{O(1)}$

\item A path decomposition of $G$ of width $q+O(q^{3/4})$ can be constructed in
the same time.

\end{itemize}

Note that if we achieve the above, we obtain the lemma, because we can solve
the original \tsat\ instance by deciding $C_4$-\textsc{Hitting Set} on $G$.
Indeed, we can assume without loss of generality that $p$ (and therefore $q$)
is large enough so that the path decomposition of $G$ which has width
$q+O(q^{3/4})$ actually has width at most $(1+\frac{\eps}{4})q$. The running
time then becomes $2^{(1-\eps)(1+\frac{\eps}{4})^2\frac{q^2}{2}}
|\phi|^{O(1)}$, which becomes at most $2^{(1-\frac{\eps}{4})p}|\phi|^{O(1)}$.

We now describe our construction of $G$. Suppose that the bags of the given
decomposition of $\psi$ are numbered $B_1,\ldots, B_t$. Furthermore, suppose we
have an injective mapping of clauses of $\psi$ to bags that contain all their
literals, such that we always maps clauses to bags which contain the same
vertices as the previous and next bag (this can easily be achieved by inserting
before and after each bag, copies of the bag).  We do the following:

\begin{enumerate}

\item For each variable $x\in V_0$ we construct two vertices, $x$ and $\neg x$
and place a strong edge connecting them.

\item For each $j\in[t]$ we construct $q+6\lceil r\sqrt{r}\rceil$ vertices,
call them $v_{j,i_1}$ for $i_1\in[q+6\lceil r\sqrt{r}\rceil]$. We call these
vertices the \emph{backbone} vertices.

\item For each $j\in[t]$ and $i_1\in[q], i_2\in[r]$ we consider the variables
of $V_{i_1,i_2} \cap B_j$. For each such variable $x$ we select the minimum
index $i_1'$ which has not yet been selected for any variable of
$\bigcup_{i_2\in[r]}V_{i_1,i_2}\cap B_j$, such that $i_1<i_1'\le q+6\lceil
r\sqrt{r}\rceil$ and construct a vertex adjacent to $v_{j,i_1}$ and
$v_{j,i_1'}$.  We say that this new vertex represents the copy of $x$ in bag
$B_j$ and call this vertex $x_j$.  Note that because
$|\bigcup_{i_2\in[r]}V_{i_1,i_2} \cap B_j|\le q-i_1 + 6r\sqrt{r}$, there exist
sufficiently many indices $i_1'$ to perform this procedure assigning a distinct
$i_1'$ to each variable. Though the choice of $i_1'$ is in part arbitrary
(because we have not specified the order in which the variables of
$V_{i_1,i_2}$ are considered), we remain consistent, that is for each $x\in
V_{i_1,i_2}$, if in bag $B_j$ we represent $x$ via a vertex adjacent to
$v_{j,i_1}$ and $v_{j,i_1'}$, if $x$ appears in bag $B_{j+1}$ we represent it
via a vertex adjacent to $v_{j+1,i_1}$ and $v_{j+1,i_1'}$. Furthermore, note
that the indices $i_1'$ used for the variables of each $V_{i_1,i_2}$ form an
interval.

\item For each $j\in[t]$, $i_1\in[q]$, and $i_2\in[r]$, let $X_{j,i_1,i_2}$ be
the set of vertices representing copies of variables of $V_{i_1,i_2}$ in bag
$B_j$. We add a budget-setter on $X_{j,i_1,i_2}$ with budget $
\frac{|X_{j,i_1,i_2}|}{2}$.

\item (Consistency) For each $j\in[t-1]$ and $i_1\in[q]$, $i_2\in[r]$ such that
$B_j\cap V_{i_1,i_2} = B_{j+1}\cap V_{i_1,i_2}$, we add the following
consistency gadget: for $x\in B_j\cap V_{i_1,i_2}$, let $x_j$ be the vertex
representing the copy of $x$ in $B_j$ and $x_{j+1}$ be the vertex representing
$x$ for $B_{j+1}$. Recall that $x_j$ is adjacent to $v_{j,i_1},v_{j,i_1'}$ and
$x_{j+1}$ is adjacent to $v_{j+1,i_1},v_{j+1,i_1'}$, for some $i_1'>i_1$.  We
add three verticex $x_j^a, x_j^b, x_j^c$, such that (i) $x_j^a$ has a strong
edge to $x_j^b$ (ii) $x_j^c$ has a strong edge to $x_{j+1}$ (iii) $x_j^a$ is
adjacent to $v_{j,i_1}$ and $v_{j,i_1'}$ (iv) $x_j^b$ and $x_j^c$ are adjacent
to $v_{j,i_1'}$ and $v_{j+1,i_1}$.  We refer the reader to \cref{fig:c4red}.

\item For each $j\in[t-1]$ and $i_1\in[q]$, $i_2\in[r]$ such that $B_j\cap
V_{i_1,i_2} = B_{j+1}\cap V_{i_1,i_2}$, let $X^a_{j,i_1,i_2}, X^b_{j,i_1,i_2},
X^c_{j,i_1,i_2}$ be the sets of vertices $x^a_j, x^b_j, x^c_j$ respectively
added in the previous step for these values of $j,i_1,i_2$.  We add a
budget-setter for each of $X^a_{j,i_1,i_2}, X^b_{j,i_1,i_2}, X^c_{j,i_1,i_2}$
with budget half the cardinality of each set (all three sets have size
$|B_j\cap V_{i_1,i_2}|$).

\item For each clause $c$ of $\psi$, suppose that the mapping assigns $c$ to
bag $B_j$, that contains all its variables. Suppose that the clause contains
$\rho=O(r)$ literals. We construct a set of $\rho$ vertices, one for each
literal and connect all pairs via strong edges. Consider a literal, and suppose
that it involves a variable $x\in V_{i_1,i_2}$. If $x$ appears positive in $c$,
then we place a strong edge between the vertex representing the literal and
$x_j$, otherwise we place a strong edge between the vertex representing the
literal and $x_{j-1}^c$ (which exists, because $B_{j-1}=B_j$). Similarly, if
the literal involves a variable $x\in V_0$, if $x$ appears positive in $c$,
then we place a strong edge between $x$ and the vertex representing the
literal; otherwise we place the strong edge between the vertex representing the
literal and $\neg x$.

\end{enumerate}

We now need to set a budget for our hitting set. To compute this budget we add
the total cost of all the budget setters we have used in steps 4 and 6
(including the budget itself and the internal cost of the budget-setter as
given in \cref{lem:budget-setter}). We add to this budget for each clause $c$
of arity $\rho$ a budget of $\rho-1$; and further add one more vertex for each
variable of $V_0$.  This sum gives an integer $k$ that is the total number of
vertices that we seek to delete to hit all $C_4$'s in the new graph.

Before we argue for correctness, we refer the reader again to \cref{fig:c4red},
which illustrates the consistency gadget, which is the most involved part of
the construction. The rest of the construction is straightforward, as clause
gadgets can easily be seen to force the deletion of all but one literals of a
clause, so we only need to check that the remaining literal is set to True. A
variable $x$ is set to True if we include in the hitting set a vertex $x_j$
that represents it in some bag $B_j$. 

The crucial part now is how to ensure that this assignment stays consistent
throughout the construction, without increasing the pathwidth by much. To see
what we mean by this, one could consider a construction where we have strong
edges forming a path $x_j, x_j^a, x_{j+1}, x_{j+1}^a, x_{j+2},\ldots$. By
setting an appropriate budget this could easily ensure that either we delete
all $x_j$ vertices or none. However, the pathwidth would be linear in the
original pathwidth $p$, while we seek a pathwidth of roughly
$q=\lceil\sqrt{2p}\rceil$.  It is therefore necessary to pass on some of this
information indirectly. This is done in \cref{fig:c4red} between vertices $x_j$
and $x_j^a$ as well as vertices $x_j^b$ and $x_j^c$, because these vertices
share two common neighbors. Setting things up carefully in this way, we are
able to propagate information about the $p$ variables of bag $B_j$ to the $p$
variables of bag $B_{j+1}$, via only roughly $\sqrt{2p}$ vertices, which gives
us the desired pathwidth.

Let us now formally argue that this construction achieves the desired
properties. First, suppose that $\psi$ is satisfiable by an assignment that
sets exactly half of the variables of each $B_j\cap V_{i_1,i_2}$ to True. For
each $x\in V_0$ we add to the hitting set the variable that corresponds to the
literal involving $x$ that is set to True by the assignment; for each clause
$c$ we pick a literal that is set to True by the assignment and place in the
hitting set the vertices representing all other literals; for each $x\in
V_{i_1,i_2}$, for some $i_1\in[q], i_2\in[r]$ we consider each bag $B_j$ such
that $x\in B_j$: if $x$ is set to True we place in the hitting set $x_j$ and
$x_j^b$ (if it exists), otherwise we place in the hitting set $x_j^a$ and
$x_j^c$. Because the assignment sets exactly half of the variables of each
$B_j\cap V_{i_1,i_2}$ to True, this places in the hitting set vertices which
allow us to use \cref{lem:budget-setter} to extend the solution to all
budget-setters while using exactly the available budget. We claim that the set
we have formed indeed hits all $C_4$'s. We first note that for each strong edge
we have selected at least one of its endpoints. If we remove from the graph the
selected set and all strong edges, what remains is a bipartite graph with
backbone vertices $v_{j,i_1}$ on one side and vertices of consistency gadgets
on the other. Therefore, a $C_4$ would require two vertices from consistency
gadgets that have two common backbone neighbors. However, a vertex $x_j$ only
has two common neighbors with the unique vertex $x_j^a$ and one of the two is
in the hitting set, and the argument is similar for $x_j^b$ and $x_j^c$.
Hence, we have a feasible solution of the desired size.

For the converse direction, if we have a hitting set, for each $x\in V_0$ it
must contain one of the two literals and for each clause $c$ of arity $\rho$ it
must contain at least $\rho-1$ vertices representing literals, due to the
strong edges. Furthermore, by \cref{lem:budget-setter} the solution must
contain at least the expected budget of vertices from each budget-setter. Since
we have set the total budget to be the sum of these values, all these
inequalities are in fact equalities. In particular, this means that no backbone
vertex $v_{j,i_1}$ is in the hitting set. Consider now the sets of vertices
$X_{j,i_1,i_2}, X_{j,i_1,i_2}^a, X_{j,i_1,i_2}^b, X_{j,i_1,i_2}^c$, as
described in steps 4 and 6, for some specific values of $j,i_1,i_2$. These sets
have the same cardinality and the budget-setters ensure that we include exactly
half of each in the hitting set. This implies that if there is some $x\in
B_j\cap V_{i_1,i_2}$ such that $x_j,x_j^a$ are both in the hitting set, there
exist some other $y\in B_j\cap V_{i_1,i_2}$ with neither $y_{j},y_{j}^a$ in the
set.  This would cause a contradiction, as $y_{j},y_{j}^a$ form a $C_4$ with
their two common backbone neighbors. With similar reasoning, we conclude that
for each $x\in B_j\cap V_{i_1,i_2}$, exactly one from each pair
$\{x_j,x_j^a\}$, $\{x_j^a,x_j^b\}$, $\{x_j^b, x_j^c\}$, $\{x_j^c,x_{j+1}\}$ is
in the hitting set. This implies that $x_j$ is in the hitting set if and only
if $x_{j+1}$ is.  We can now use this fact to extract an assignment to the
variables of $\bigcup_{i_1\in[q], i_2\in[r]} V_{i_1,i_2}$, setting $x$ to True
if for some vertex $x_j$ representing $x$, $x_j$ is in the hitting set. We
extend this assignment to $V_0$ in the natural way, setting a value for each
$x\in V_0$ that sets to True the literal whose corresponding vertex is in the
hitting set.  We now claim that this assignment is satisfying, because if we
take a clause $c$, the literal corresponding to the vertex that is not in the
hitting set must be set to True. This is the case, because the literal has a
strong edge to a vertex that is in the hitting set if and only if the literal
is True, and since the literal vertex is not in the hitting set this vertex is.

Finally, since it is not hard to see that the new instance can be constructed
in the promised time, we need to argue about the pathwidth of the constructed
instance and show how a path decomposition can be built. We first note that
strong edges can be dealt with as if they were edges, as this will increase
pathwidth by at most $1$.  We now start from the path decomposition of $\psi$
and in each bag $B_j$ we place all backbone vertices $v_{j,i_1}$, for $i_1\in
[q+6\lceil r\sqrt{r}\rceil]$. At the moment each bag has the desired width. We
now explain how between $B_j$ and $B_{j+1}$ we can place a sequence of bags so
that we cover all the consistency gadgets and so that the sets $X_{j,i_1,i_2},
X_{j,i_1,i_2}^a, X_{j,i_1,i_2}^b, X_{j,i_1,i_2}^c$ all appear fully inside a
bag (this will allow us to invoke \cref{lem:budget-setter2} to take care of the
budget-setters). For each $i_1\in[q]$ we insert in the decomposition a bag
$B_j^{i_1}$ that contains $\{ v_{j,i'}\ |\ i'\ge i_1\} \cup \{ v_{j+1,i''}\ |\
i''\le i_1\}$. In other words we gradually exchange the backbone vertices of
$B_j$ with those of $B_{j+1}$. We place, for each $i_2\in[r]$,
$X_{j,i_1,i_2}^a\cup X_{j,i_1,i_2}^b$ inside a fresh copy of $B_j^{i_1}$.
Furthermore, for each $i_1\in[q], i_2\in [r]$, find the minimum
$i_1'\in[q+6\lceil r\sqrt{r}\rceil]$ such that some $x_j^c\in X_{j,i_1,i_2}^c$
is adjacent to $v_{j,i_1'}$ and place $X_{j,i_1,i_2}^c\cup X_{j+1,i_1,i_2}$
inside a fresh copy of $B_j^{i_1'}$. This covers everything in the consistency
gadgets, except sets $X_{j,i_1,i_2}$ where $B_j\cap V_{i_1,i_2}\cap
B_{j-1}=\emptyset$, but such sets are added to a fresh copy of $B_j$. After
performing the above, the size of all bags has increased by at most
$O(\sqrt{q})$ and we have covered all the consistency gadgets. It now remains
to cover variables of $V_0$, but these can be handled by adding, for each $x\in
V_0$, both vertices $x,\neg x$ to all bags that contained $x$ and since no bag
contains more than $r$ such variables, this also increases width by
$O(\sqrt{q})$; to cover the gadget for a clause $c$ of arity $\rho=O(r)$, where
$c$ is mapped to $B_j$, we add the $\rho$ vertices representing the literals to
all bags $B_{j-1}, B_j, B_{j+1}$ and the bags we have introduced between them,
and this further increases the width of the decomposition by $O(r)=O(\sqrt{q})$
(because the mapping is injective).  This completes the construction of a path
decomposition of the new graph of the desired width. \end{proof}

\section{XNLP-complete problems}\label{sec:xnlp}

In this section we show that the phenomenon of \ppseth-equivalence is not
confined to the class FPT. More specifically, we focus on three parameterized
problems which are XNLP-complete and hence highly unlikely to be FPT, namely
\textsc{List Coloring} parameterized by pathwidth, $k$-\textsc{DFA
Intersection} parameterized by $k$, and \textsc{Short}-$k$-\textsc{Independent
Set Reconfiguration} parameterized by $k$.  All these problems admit algorithms
with complexity $n^{k+O(1)}$, for $k$ the relevant parameter.  We show that
obtaining even slightly faster algorithms is equivalent to the \ppseth, that
is, algorithms with complexity $n^{(1-\eps)k}$ exist \emph{if and only if}
there exists a satisfiability algorithm falsifying the \ppseth. Notice that
this implies that the XNLP-complete problems of this section are  in a sense
equivalent to the FPT problems of previous sections, since obtaining an
algorithmic improvement for one implies an algorithmic improvement for all,
despite the fact that the problems in question belong in classes strongly
believed to be distinct.

It is worth noting that the lower bounds we obtain for the three XNLP-complete
problems are in fact even sharper than what we described above. More precisely,
we can rule out not just algorithmic improvements which improve the coefficient
of $k$ in the exponent, but even improvements which shave off a small additive
constant in the exponent. This is achieved while still maintaining the property
that improving upon the lower bound is equivalent to the \ppseth. We discuss
this in more detail below.

Before we present our results, let us briefly give some relevant context. The
class XNLP was recently defined in \cite{BodlaenderGNS21}, though it had been
previously studied (under a different name) in \cite{ElberfeldST15} and in a
slight variation in \cite{Guillemot11}. In a nutshell, XNLP is the class of
problems that can be solved in FPT time and space $f(k)\log n$ by a
\emph{non-deterministic} Turing machine.  The main contribution of
\cite{BodlaenderGNS21}, as well as the follow-up work of
\cite{BodlaenderGJJL22}, was to show that XNLP is the natural home for
intractable problems parameterized by a linear structure, such as pathwidth or
linear clique-width. In retrospect, this is logical: typical intractable
problems parameterized by pathwidth (and similarly for related widths) admit XP
dynamic programming algorithms, which for each bag of the decomposition store
$n^{f(\pw)}$ entries. A non-deterministic algorithm should be able to speed up
this computation by guessing in each node one relevant entry of the table
(hence, we need space $f(\pw)\log n$) and then verifying that this is
consistent with an entry in the preceding table, continuing in this way until
the end. XNLP-completeness expresses a strong notion of intractability, because
XNLP-complete problems are hard for all classes $W[t]$ of the $W$-hierarchy.
This fruitful line of research has more recently led to attempts to
characterize other complexity classes with respect to graph width structures,
such as XALP for tree-like structures (\cite{BodlaenderGJPP22}) and XSLP for
tree-depth (\cite{BodlaenderGP23}).

At this point we can remark a similarity between the aforementioned line of
work on XNLP, and the Turing machine characterization of the \ppseth\ given by
Iwata and Yoshida (\cite{IwataY15}). Recall that it was shown in
\cite{IwataY15} that the \ppseth\ is false if and only if it is possible to
decide in time $(2-\eps)^kn^{O(1)}$ if a given non-deterministic machine with
space $k+O(\log n)$ accepts a certain input. In other words, the TM machine
problem that Iwata and Yoshida consider is a variation of the model defining
XNLP machines except that we replace $f(k)\log n$ with $k+O(\log n)$ and we ask
a significantly more fine-grained question. In a way, the motivation behind
both definitions is the assumption that it should not be possible to do better
to predict the behavior of a bounded-space NTM than enumerating all possible
configurations. Our work shows that this similarity is not superficial: there
exist XNLP-complete problems which can be solved in space $k\log n$ and for
which obtaining an algorithm faster than $n^k$ would be equivalent to
falsifying the \ppseth.

Our results rely exactly on the intuition explained above. In order to reduce
the XNLP-complete problems we consider to the \ppseth, we essentially build a
CNF formula that describes the workings of a non-deterministic algorithm that
guesses a solution; while in the converse direction, one can think that we
start the reduction from a formula with pathwidth $p=k\log n$ and reduce to an
instance with parameter value (roughly) $k$, so that $n^{(1-\eps)k}$ would be
less than $(2-\eps')^p$ for some appropriate $\eps'$, so a fast algorithm for
our problem would imply a fast satisfiability algorithm.

\subparagraph*{Sharpness of lower bounds.} If we look back at the results of
\cref{sec:fpt} we can remark that all lower bound results could be formulated
in two ways: we could for example state that 3-\textsc{Coloring} cannot be
solved in time $(3-\eps)^{\pw}n^{O(1)}$, or in time $3^{(1-\eps)\pw}n^{O(1)}$.
Clearly, this is just a cosmetic decision, as the two formulations are
equivalent. Things become less clear if we are considering a problem solvable
in time $n^k$. If we believe that this running time is optimal, does this mean
that it is impossible to do $n^{(1-\eps)k}$, or that it is impossible to do
$n^{k-\eps}$? Clearly, the latter lower bound is at least as strong, but the
two statements are not obviously equivalent. Indeed, proving that the two
statements are equivalent would entail showing that an algorithm running in
$n^{k_0-\eps}$, for some fixed $k_0,\eps$, implies an algorithm running in time
$n^{(1-\eps')k}$, at least for $k$ sufficiently large.

Despite this seeming inequivalence, the lower bounds we present are
particularly sharp because we prove that \ppseth-equivalence implies that we
can freely exchange between the two types of lower bounds. To be more concrete,
for the $k$-\textsc{DFA Intersection} problem we prove the following: there is
an algorithm solving the problem in $n^{(1-\eps)k}$ if and only if there is an
algorithm solving the problem in $n^{k-\eps}$ (both of which happen if and only
if the \ppseth\ is false). Intuitively, what is happening here is that we prove
that an algorithm solving the intersection problem in $n^{k-\eps}$ is fast
enough to refute the \ppseth; however, if we refute the \ppseth\ we can reduce
the intersection problem to a \textsc{SAT} instance with pathwidth $k\log n$,
which a $(2-\eps)^{\pw}$ algorithm would decide in $n^{(1-\eps)k}$ (modulo some
additive constants in the exponent which disappear for large $k$). For the DFA
problem the $O(n^k)$ bound is tight, as there is an algorithm of this running
time solving the problem.

The situation is more intriguing for the reconfiguration problem we consider.
We show that deciding if it is possible to reconfigure an independent set of
size $k$ to another with at most $\ell$ moves can be done with a complexity
that is slightly faster than $n^{k}$ by obtaining an
$O(n^{k-3+\omega}\cdot\ell)$ time algorithm using fast matrix multiplication
($\omega$ is the matrix multiplication constant). This seems to point to the
possibility that the ``correct'' exponent of $n$ is $k-1$. We show lower bounds
that match this: obtaining an $n^{k-1-\eps}\cdot\ell^{O(1)}$ time algorithm is
equivalent to falsifying the \ppseth, and also equivalent to obtaining an
$n^{(1-\eps)k}\ell^{O(1)}$ time algorithm. In the remaining case of
\textsc{List Coloring} parameterized by pathwidth, where there is some small
inevitable overhead in our reduction that prevents us from giving an equally
sharp bound, we still obtain that under the \ppseth\ the correct complexity of
this problem is of the form $n^{k+c}$, for some constant $c$ that belongs in
the reasonably-sized interval $[-4,1]$.

\subsection{List Coloring Parameterized by Pathwidth}

In this section we consider the \textsc{List Coloring} problem parameterized by
pathwidth. \textsc{List Coloring} is one of the first problems to be shown to
be W[1]-hard parameterized by structural parameters such as treewidth and
pathwidth (indeed even by vertex cover) in  \cite{FellowsFLRSST11}. More
recently, a much more precise characterization was given by Bodlaender,
Groenland, Nederlof, and Swennenhuis \cite{BodlaenderGNS21}, who showed that
\textsc{List Coloring} is XNLP-complete parameterized by pathwidth. As
explained, there is a strong intuitive connection between the class XNLP and
\ppseth-equivalence. In this section we confirm this intuition by showing that
obtaining an algorithm for \textsc{List Coloring} with complexity
$O(n^{(1-\eps)\pw})$ would be equivalent (that is, both necessary and
sufficient for) falsifying the \ppseth. Since it is fairly straightforward to
obtain an algorithm with complexity $n^{\pw+1}$ using standard DP, this result
is tight up to an additive constant in the exponent. Indeed, in \cref{thm:lc}
below we also give a bound on this additive constant, by showing that obtaining
an algorithm of running time $O(n^{\pw-4-\eps})$ is already enough to falsify
the \ppseth. If the \ppseth\ is true, this leaves a quite narrow range for the
possible running time of the best algorithm. 

Compared to previous reductions for this problem, the main obstacle we are
faced with is a common one for fine-grained reductions: we need to be more
efficient. We want to reduce \tsat\ to \textsc{List Coloring} and our strategy
is to group the variables of each bag of a decomposition of the given CNF
formula $\psi$ into $k$ groups, represent each with a vertex of the
\textsc{List Coloring} instance, and give this vertex roughly $2^{\pw(\psi)/k}$
possible colors, one for each assignment. This strategy allows us to obtain a
path decomposition of the new graph of width (roughly) $k$. A crucial detail,
though, is that we also need the size of the whole constructed graph $n$ to be
not much larger than $2^{\pw(\psi)/k}$, because, assuming an $n^{(1-\eps)\pw}$
algorithm for \text{List Coloring} will give a fast algorithm for \tsat\ only
if $n^{(1-\eps)k}<(2-\eps')^{\pw(\psi)}$. In particular, we cannot afford to
have a construction that is quadratic in the sizes of the lists, because
$(2^{2\pw(\psi)/k})^{(1-\epsilon)k} \approx 4^{\pw(\phi)}$ which does not give
a fast algorithm for \tsat. Typical \textsc{List Coloring} reductions, however,
almost always have this quadratic blow-up, because for any pair of vertices on
which we want to impose a constraint, one considers every pair of colors that
represents an unacceptable combination and constructs a vertex to ensure that
this combination is not picked (these are called the helper vertices in
\cite{BodlaenderGNS21}). One of the main technical difficulties we need to work
around is to avoid this type of super-linear construction, while still encoding
the \tsat\ constraints for satisfaction and consistency.

The main result of this section is stated below.

\begin{theorem}\label{thm:lc}

The following statements are equivalent:

\begin{enumerate} 

\item The \ppseth\ is false.

\item There exist $\eps>0,p>4$ and an algorithm that takes as input a
\textsc{List Coloring} instance on an $n$-vertex graph $G$ and a path
decomposition of $G$ of width $p$ and solves the instance in time
$O(n^{p-4-\eps})$.

\item There exist $\eps>0, p_0>0$ and an algorithm which for all $p>p_0$ solves
\textsc{List Coloring} on instances with $n$ vertices and pathwidth $p$ in time
$O(n^{(1-\eps)p})$.

\end{enumerate}

\end{theorem}

In order to establish \cref{thm:lc} we note that the implication
3$\Rightarrow$2 is trivial. We prove that a fast \textsc{List Coloring}
algorithm (in the sense of statement 2) falsifies the \ppseth\ in
\cref{lem:sattolc}, giving 2$\Rightarrow$1; and that a fast \textsc{SAT}
algorithm implies a fast \textsc{List Coloring} algorithm (in the sense of
statement 3) in \cref{lem:lctosat}, giving 3$\Rightarrow$1 and the theorem. 

\begin{lemma}\label{lem:sattolc} If there exist $\eps>0, p>4$ such that there
exists an algorithm that takes as input an instance of \textsc{List Coloring}
$G=(V,E)$ and a path decomposition of $G$ of width $p$ and decides the problem
in time $O(|V|^{p-4-\eps})$, then there exist $\eps'>0,c>0$ such that there
exists an algorithm that takes as input a 3-CNF formula $\psi$ and decides its
satisfiability in time $O((2-\eps')^{\pw}|\psi|^c)$.  \end{lemma}

\begin{proof}

Fix the $\eps,p$ for which the supposed \textsc{List Coloring} algorithm
exists. We are given a 3-CNF formula $\psi$ on $n$ variables and a path
decomposition of $\psi$ of width $\pw(\psi)$. Let $k=p-4$. We want to construct
a \textsc{List Coloring} instance $G=(V,E)$ with the following properties:

\begin{enumerate}

\item $G$ has a valid list coloring if and only if $\psi$ is satisfiable.

\item $|V(G)|=O(2^{\pw(\psi)/k}n^4)$ and $G$ can be constructed in time linear
in $|V(G)|$.

\item We can construct in time $O(|V(G)|)$ a path decomposition of $G$ of width
$p$.

\end{enumerate}

Before we proceed, let us argue that the above imply the theorem. We will
attempt to decide if $\psi$ is satisfiable by constructing $G$ and solving
\textsc{List Coloring} on the resulting instance using the assumed algorithm.
The running time is $O(|V(G)|^{p-4-\eps}) =
O(2^{\frac{\pw(\psi)}{(p-4)}(p-4-\eps)}n^{4p})$. Set $\eps'$ so that $2-\eps' =
2^{\frac{p-4-\eps}{p-4}}$ and $c'=4p$. The running time is then
$O((2-\eps')^{\pw(\psi)}n^{c'})$, and $\eps', c'$ are constants depending only
on $\epsilon, p$.

Assume that we are given a path decomposition of $\psi$, on which we execute
the algorithm of \cref{lem:nice}, so we now have a nice path decomposition of
$\psi$ of width $\pw(\psi)$ with $O(n^4)$ bags $B_1,\ldots,B_t$, as well as a
function $b$ that maps each clause to the index of a bag that contains all its
variables. Using \cref{lem:pwcolor} we partition the set of variables of $\psi$
into $k$ sets $V_1,\ldots,V_k$, such that for each bag $B$ and $i\in[k]$ we
have $|B\cap V_i|\le \lceil\frac{\pw(\psi)+1}{k}\rceil$.

For each $j\in[t]$ and $i\in[k]$ we construct a vertex $x_{i,j}$. The informal
meaning is that the color of $x_{i,j}$ is supposed to represent the assignment
to the variables of $B_j\cap V_i$. We give to each $x_{i,j}$ a list of
$2^{|B_j\cap V_i|}$ colors which do not yet appear in the list of any other
vertex.  We now need to add some extra vertices which ensure that assignments
remain consistent across bags and that they satisfy the formula. We add the
following gadgets:

\begin{enumerate}

\item For each $j\in[t-1]$ and $i\in[k]$ for which $B_j\cap V_i = B_{j+1}\cap
V_i$ we contract the vertex $x_{i,j+1}$ into the vertex $x_{i,j}$ and use
$L(x_{i,j})$ as the list of the new vertex.  Informally, this ensures that
these two vertices will have the same color. Note that repeated applications of
this step may contract all vertices $x_{i,j'}$ for $j'$ belonging in some
interval $[j_1,j_2]$ into the single vertex $x_{i,j_1}$. Any subsequent edges
incident on another vertex of this interval will actually be placed incident on
$x_{i,j_1}$.

\item (Consistency) Consider now a $j\in[t-1]$ and $i\in[k]$ such that
$B_{j+1}\cap V_i = (B_j\cap V_1) \cup \{v\}$, for some variable $v$ of the
\tsat\ instance.  We have that $|L(x_{i,j+1})| = 2|L(x_{i,j})|$. Let
$\ell=|L(x_{i,j})|$. Recall that each element of $L(x_{i,j})$ corresponds to an
assignment to the variables of $B_j\cap V_i$, and similarly for $L(x_{i,j+1})$
and $B_{j+1}\cap V_i$. By reading such assignments in binary where $v$ is
placed as the least significant bit, we number the elements of the lists as
$L(x_{i,j}) = \{a_0,a_1,\ldots,a_{\ell-1}\}$ and
$L(x_{i,j+1})=\{b_0,b_1,\ldots,b_{2\ell-1}\}$, with the property that element
$a_i\in L(x_{i,j})$ represents an assignment that is compatible with those
represented by elements $b_{2i}, b_{2i+1}\in L(x_{i,j+1})$. We construct a path
on $\ell-1$ vertices, call them $y_{i,j,s}$, for $s\in\{0,\ldots,\ell-2\}$,
connect all its vertices to $x_{i,j}$ and assign for each
$s\in\{0,\ldots,\ell-2\}$ the list $\{a_s,a_{s+1}\}$ to $y_{i,j,s}$. We
construct a path on $2\ell-1$ vertices, call them $z_{i,j,s}$, for
$s\in\{0,\ldots,2\ell-2\}$, connect all its vertices to $x_{i,j+1}$, and assign
for each $s\in\{0,\ldots,2\ell-2\}$ the list $\{b_s,b_{s+1}\}$. For each
$s\in\{1,\ldots,\ell-2\}$ we construct a new vertex $w_{i,j,s}$ and connect it
to $y_{i,j,s-1}$, $y_{i,j,s}$, $z_{i,j,2s-1}$, $z_{i,j,2s+1}$. We give
$w_{i,j,s}$ the list $\{a_s,b_{2s-1},b_{2s+2}\}$.  For $s=0$ we construct
$w_{i,j,0}$ and connect it to $y_{i,j,0}$, $z_{i,j,1}$, and give it list
$\{a_0,b_2\}$. For $s=\ell-1$ we construct $w_{i,j,\ell-1}$, connect it to
$y_{i,j,\ell-2}$, $z_{i,j,2\ell-3}$, and give it list
$\{a_{\ell-1},b_{2\ell-3}\}$. We refer the reader to \cref{fig:lc} for an
illustration.

\item For $j\in[t-1]$ and $i\in[k]$ for which $B_{j}\cap V_i = (B_{j+1}\cap
V_i) \cup \{v\}$ we construct the same gadget as in the previous case,
exchanging the roles of $B_j$ and $B_{j+1}$.

\item (Satisfaction) For each clause $c$ of $\psi$, consider the bag
$B_{b(c)}$.  Suppose without loss of generality that $c$ contains $3$ literals,
which utilize variables from the sets $V_{i_1}, V_{i_2}, V_{i_3}$ respectively,
where $i_1,i_2,i_3$ are not necessarily distinct (if $c$ has fewer literals, we
just repeat one of its literals). Construct three vertices $s_c^1, s_c^2,
s_c^3$, adjacent to each other, with $L(s_c^1)=L(s_c^2)=L(s_c^3)=\{F,X,Y\}$.
For each $\alpha\in\{1,2,3\}$ consider all the at most $2^{|B_{b(c)}\cap
V_{i_\alpha}|-1}$ assignments to the variables of $B_{b(c)}\cap V_{i_\alpha}$
which falsify the $\alpha$-th literal of $c$. Each such assignment $\sigma$
corresponds to a color of $x_{i_\alpha,b(c)}$. For each such assignment
$\sigma$, we construct a vertex $t_{c,\alpha,\sigma}$, connected to
$s_c^\alpha$ and $x_{i_\alpha,b(c)}$, with $L(t_{c,\alpha,\sigma})$ containing
only the color $F$ and the color from $L(x_{i_\alpha,b(c)})$ corresponding to
$\sigma$.

\end{enumerate}

This completes the construction of the \textsc{List Coloring} instance. The
number of vertices of the graph we construct can be upper-bounded as follows:
the $x_{i,j}$ vertices are at most $O(kn^4)$; for each $j\in[t-1]$ there is at
most one consistency gadget added between $B_j$ and $B_{j+1}$, which contains
$O(2^{\pw(\psi)/k})$ vertices, giving $O(2^{\pw(\psi)/k}n^4)$ vertices in
total; for each clause $c$ there is a unique $B_{b(c)}$ for which we construct
a satisfaction gadget, which contains $O(2^{\pw(\psi)/k})$ vertices, giving
$O(m2^{\pw(\psi)/k})=O(n^32^{\pw(\psi)/k})$ vertices. Thus, in total
$|V(G)|=O(2^{\pw(\psi)/k}n^4)$. It it also not hard to see that the
construction can be performed in time linear in the number of vertices of the
constructed graph.

In order to construct a path decomposition of $G$, we start with a
decomposition $B_1,\ldots, B_t$, where in each $B_j$, $j\in[t]$, we place all
$x_{i,j}$, for $i\in[k]$. Recall that in step 1 of our construction we may have
merged vertices $x_{i,j'}$ for $j'\in[j_1,j_2]$ into $x_{i,j_1}$; in such cases
we place $x_{i,j_1}$ in all $B_{j'}$ for $j'\in[j_1,j_2]$. At this point each
bag contains exactly $k$ vertices, and each vertex appears in a connected
interval, but we have not yet covered any of the vertices or edges of the
consistency and satisfaction gadgets.

For the consistency gadgets, consider a $j\in[t-1]$ for which we have added a
consistency gadget for $i\in[k]$. Observe that because the decomposition of
$\psi$ was nice, $B_j\setminus\{x_{i,j}\}=B_{j+1}\setminus\{x_{i,j+1}\}$, that
is, $B_j, B_{j+1}$ only differ in the $i$-th vertex. Therefore, $|B_j\cup
B_{j+1}|=k+1$.  We add a sequence of bags between $B_j$ and $B_{j+1}$ all of
which contain $B_j\cup B_{j+1}$. Since $x_{i,j}, x_{i,j+1}$ are in all these
bags, we actually only need to add here a path decomposition of the gadget,
which is made up of two parallel paths with the $w_{i,j,s}$ vertices connecting
them. It is not hard to see that we can produce a path decomposition of this
graph containing $4$ vertices in each bag. In this way we use at most $k+5$
vertices in each bag.

Finally, for the satisfaction gadget constructed for a clause $c$ we locate the
bag $B_{b(c)}$. Recall that we have constructed three vertices $s_c^1, s_c^2,
s_c^3$, and a large number of degree $2$ vertices adjacent to one of the former
vertices and a vertex of $B_{b(c)}$. Immediately after $B_{b(c)}$ we place a
sequence of bags containing $B_{b(c)}\cup\{s_c^1,s_c^2,s_c^3\}$ and each bag
containing a distinct degree $2$ vertex adjacent to one of $s_c^1, s_c^2,
s_c^3$. The largest bag we have added contains $k+4$ vertices. Overall, we have
a decomposition where each bag has at most $k+5$ vertices, therefore a
decomposition of width $p=k+4$.

To establish that the graph $G$ has a valid list coloring if and only if $\psi$
is satisfiable one direction is relatively straightforward: given an assignment
to $\psi$ we can pick a color for each $x_{i,j}$, since each color of
$L(x_{i,j})$ corresponds to an assignment to the variables of $B_j\cap V_i$. To
extend this to consistency gadgets, there is a unique way to color the two
paths $y_{i,j,s}$ and $z_{i,j,s}$ that uses all colors of $L(x_{i,j})$ and
$L(x_{i,j+1})$ which are not used in $x_{i,j}, x_{i,j+1}$.  Consider now a
vertex $w_{i,j,s}$: if color $a_s$ is not used in its neighborhood, we can use
that color and we are done; if color $a_s$ is used, this means that $x_{i,j}$
was \emph{not} assigned $a_s$, therefore, since the assignment is consistent,
$x_{i,j+1}$ was \emph{not} assigned $b_{2s}$ nor $b_{2s+1}$.  Consider now the
vertices $z_{i,j,2s-1}, z_{i,j,2s}, z_{i,j,2s+1}$, with lists $\{b_{2s-1},
b_{2s}\}, \{b_{2s}, b_{2s+1}\}, \{b_{2s+1}, b_{2s+2}\}$. The colors $b_{2s},
b_{2s+1}$ do not appear in any other lists, therefore they are both used here,
therefore one of $b_{2s-1}, b_{2s+2}$ is free to use at $w_{i,j,s}$. To extend
the coloring to the satisfaction gadgets, if the $\alpha$-th literal of $c$ is
satisfied by the assignment we select color $F$ for $s_c^\alpha$ and colors
$X,Y$ for the other two vertices of $\{s_c^1, s_c^2, s_c^3\}$. All vertices
$t_{c,\alpha,\sigma}$ for which color $F$ is available receive that color,
while the remaining such vertices receive a color that represents an assignment
to the variables of $B_{b(c)}\cap X_{i_\alpha}$ which falsifies the $\alpha$-th
literal of $c$. This is valid, as this color cannot have been used on
$x_{i_\alpha,b(c)}$.

For the converse direction, we need to argue that if $G$ has a valid list
coloring, then we can extract a consistent satisfying assignment to the
variables of $\psi$. The main obstacle is when we have two vertices $x_{i,j},
x_{i,j+1}$ which do not represent the same set of variables (otherwise the
vertices have been merged and consistency is not a problem). We therefore need
to argue that our consistency gadget ensures that if $x_{i,j}$ was assigned
color $a_s$, then $x_{i,j+1}$ must be assigned color $b_{2s}$ or $b_{2s+1}$. To
see this, we first observe that each color of $L(x_{i,j})$ is either used on
$x_{i,j}$ itself or one of the vertices of the path $y_{i,j,s}$. We will then
argue that if a color $a_s$ is used in the path, then $b_{2s}$ and $b_{2s+1}$
must be used in the path $z_{i,j,s}$. To see this, consider the vertex
$w_{i,j,s}$, which in this case must take one of the colors $b_{2s-1},
b_{2s+2}$; however, either one of these choices forces the use of both $b_{2s}$
and $b_{2s+1}$ in the path $z_{i,j,s}$. Since all the colors not used in
$x_{i,j}$ have their corresponding colors not used in $x_{i,j+1}$, we conclude
that the assignments extracted from the colors of $x_{i,j}$ and $x_{i,j+1}$
must be compatible. We now need to argue that the extracted assignment is
satisfying. Take a clause $c$ and consider $s_c^1, s_c^2, s_c^3$. One of these
vertices must have color $F$, since they form a $K_3$ with list $\{F,X,Y\}$,
suppose without loss of generality that $s_c^1$ is that vertex. Then, all the
vertices $t_{c,1,\sigma}$ received a color that encodes an assignment to
$x_{i_1,b(c)}$ that falsifies the first literal of $c$. Hence, $x_{i_1,b(c)}$
cannot have received such a color, hence the assignment we extract from
$x_{i_1,b(c)}$ satisfies $c$.  \end{proof}

\begin{figure}

\centering

\includegraphics[height=0.25\textheight]{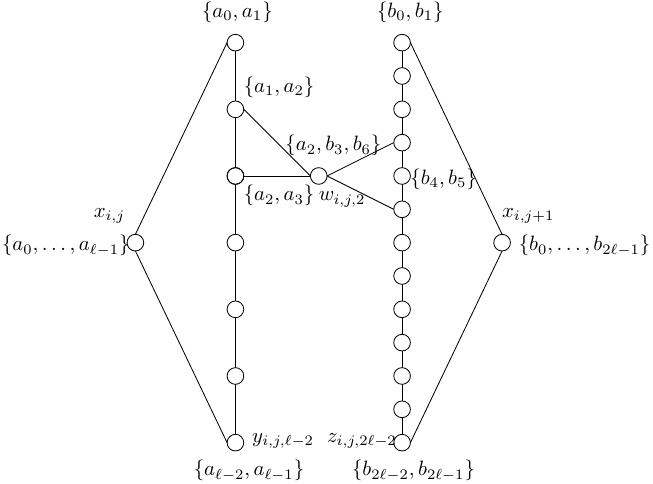} \caption{Consistency
gadget from \cref{lem:sattolc}. On the left we have vertex $x_{i,j}$ and on the
right vertex $x_{i,j+1}$.}\label{fig:lc}

\end{figure}

\begin{lemma}\label{lem:lctosat} If there exist $\eps>0, c>0$ such that there
exists an algorithm that takes as input a CNF formula $\psi$ and decides its
satisfiability in time $((2-\eps)^{\pw}|\psi|^c)$, then there exist
$p_0>0,\eps'>0$, such that there exists an algorithm which for any $p>p_0$
takes as input a \textsc{List Coloring} instance $G=(V,E)$ and a path
decomposition of $G$ of width $p$ and decides the problem in time
$O(|V|^{(1-\eps')p})$.\end{lemma}

\begin{proof}

Fix the $\eps,c$ for which a fast \textsc{SAT} algorithm exists, and suppose we
are given a graph $G=(V,E)$ on $n$ vertices, together with a list of possible
colors for each vertex, as well as a path decomposition of $G$ of width $p$. We
will construct from $G$ a CNF formula $\psi$ such that:

\begin{enumerate}

\item $\psi$ is satisfiable if and only if $G$ has a valid list coloring.

\item $|\psi| = O(n^3)$

\item $\psi$ can be constructed in time linear in $|\psi|$.

\item $\pw(\psi)=(p+1)\log n$

\end{enumerate}

Let us explain why if we achieve the above, we obtain the lemma. Given the
instance $G=(V,E)$ to \textsc{List Coloring} and the corresponding
decomposition, we construct $\psi$ and execute the supposed \textsc{SAT}
algorithm. The running time of the procedure is
$O((2-\eps)^{\pw(\psi)}|\psi|^c) = O((2-\eps)^{p\log n}n^{3c+1})$. Let
$\delta>0$ be such that $2-\eps=2^{1-\delta}$, in particular set
$\delta=1-\log(2-\eps)$.  The running time is at most
$O(n^{(1-\delta)p+3c+1})$. Set $p_0=2(3c+1)/\delta$. If $p>p_0$ then we have
that $n^{(1-\delta)p+3c+1} < n^{(1-\frac{\delta}{2})p}$, so setting
$\eps'=\delta/2$ we obtain an algorithm with the promised running time.
Observe that $p_0,\eps'$ are constants that depend only on $\eps, c$.

The construction of $\psi$ is now straightforward. Assume without loss of
generality that $n$ is a power of $2$ (for example, by adding isolated vertices
to $G$, which can at most double its size). Let $V=\{v_1,\ldots, v_n\}$ and
assume that for each $v\in V$ we have been given a list $L(v)$ of available
colors, with $|L(v)|\le n$ (this assumption is without loss of generality,
since if a vertex has a list of size $n+1$ or more, it can always be colored,
so we can remove it from the graph).  For each $v_j$, with $j\in [n]$ we
construct $\log n$ variables $x_{j,s}$ with $s\in[\log n]$, with the intended
meaning that their values will encode the element of $L(v_j)$ which is selected
in a valid coloring. We now need to add clauses to ensure that the assignment
to these variables does indeed encode a valid coloring. We add the following
clauses.

\begin{enumerate}

\item For each $j\in[n]$ for each assignment $\sigma$ to the variables
$x_{j,s}$ for $s\in[\log n]$, we interpret the assignment as a binary number.
If that number is at least $|L(v_j)|$, then we add a clause that is falsified
by the assignment $\sigma$.

\item For each $j,j'\in[n]$ with $v_jv_{j'}\in E$, for each pair of assignments
$\sigma, \sigma'$ to the variables $x_{j,s}$ and $x_{j',s}$ for $s\in[\log n]$,
we interpret the two assignments as two binary numbers $s_1,s_2$. Sort $L(v_j)$
and $L(v_{j'})$ and consider them as arrays with numbering starting at $0$. If
$L(v_j)[s_1]=L(v_{j'})[s_2]$, that is, if the $s_1$-th element of $L(v_j)$ is
the same as the $s_2$-th element of $L(v_{j'})$, then we add to $\psi$ a clause
that is falsified by the assignments $\sigma,\sigma'$ to the variables
$x_{j,s}, x_{j',s}$.

\end{enumerate}

In the first step we construct at most $n$ clauses for each vertex of $G$,
while in the second step we construct for each edge of $G$ at most $n$ clauses.
The total size of $\psi$ is then $O(n^3)$. Note that if the lists $L(v_j),
L(v_{j'})$ are sorted, then the second step can be executed in time linear in
the number of clauses produced. 

To bound the pathwidth of $\psi$ we can take the decomposition of $G$ and
replace every occurrence of a vertex $v_j$ with the variables $x_{j,s}$, for
$s\in[\log n]$. This clearly has the promised width. To see that this is a
valid decomposition, observe that each clause we add contains the variables
$x_{j,s}$ and $x_{j',s}$ for two vertices $v_j,v_{j'}$ which are neighbors in
$G$, therefore must appear together in some bag.

Finally, we show that $\psi$ is satisfiable if and only if $G$ has a proper
list coloring. The correspondence is not hard to see by interpreting the value
of $x_{j,s}$ as the index of the color selected for $v_j$ in $L(v_j)$, where
$L(v_j)$ is sorted and viewed as a zero-based array. If there is a coloring of
$G$, the corresponding assignment satisfies all clauses of the first step
(because we select a color with index at most $|L(v_j)|-1$), and of the second
step, because the coloring is proper. For the converse direction, we can
extract a coloring of $G$ from a satisfying assignment in the same way. A
satisfying assignment must assign to each vertex a color from its list (because
of clauses of the first step), and the coloring will be proper because of
clauses of the second step.  \end{proof}

To end this section, let us also briefly note that, similarly to
\cite{BodlaenderGNS21}, we can obtain from \cref{thm:lc} \ppseth-equivalence
for related problems, such as \textsc{Precoloring Extension}. In
\textsc{Precoloring Extension} we are given a graph such that some of its
vertices are already colored and are asked to extend this to a proper coloring
of the whole graph, while minimizing the number of colors used. 

\begin{corollary}

There exist $\eps>0, p_0>0$ and an algorithm which for all $p>0$ takes as input
an instance $G=(V,E)$ of \textsc{Precoloring Extension} and a path
decomposition of $G$ of width $p$ and decides the instance in time
$O(n^{(1-\eps)p)})$ if and only if the \ppseth\ is false. 

\end{corollary}

\begin{proof}

To obtain the corollary we observe that there are simple reductions in both
directions between \textsc{List Coloring} and \textsc{Precoloring Extension}
that preserve the pathwidth up to adding $1$, as explained in
\cite{BodlaenderGNS21}. \end{proof}

\subsection{DFA and NFA intersection}

In this section we consider the following basic problem: we are given a
collection of $k$ (deterministic or non-deterministic) finite automata, each
with $n$ states, working over a finite alphabet $\Sigma$. We are also supplied
a size bound $\ell$. The question is whether there exists any string over
$\Sigma$ that is accepted by all $k$ automata. This problem is called
\textsc{Bounded FA Intersection Non-Emptiness}.

Before we go on it is important to note that for DFAs this problem can easily
be solved in time $O(n^k)$. To see this, we note that we can construct the
intersection FA, which has $n^k$ states (one for each tuple of the states of
the given automata) and then perform BFS to calculate the shortest distance
from the initial state to an accepting state. Interestingly, this running time
bound does not depend at all on $\ell$.  

We choose to focus on the bounded version of the problem (where we are supplied
with an upper bound $\ell$ on the desired string) because this version of the
problem is XNLP-complete (\cite{BodlaenderGNS21}) and as we have mentioned
there is a clear intuitive connection between XNLP-completeness and
\ppseth-equivalence. We remark that Oliveira and Wehar have already shown that
an algorithm solving the intersection problem for $k$ DFAs in time $n^{k-\eps}$
would refute the SETH (\cite{OliveiraW20}). Similarly to the results of
\cref{sec:fpt}, our result improves this SETH-hardness to \ppseth-equivalence.

In the remainder, we treat NFAs (Non-deterministic Finite Automata) as directed
graphs, where there is a special initial vertex, a set of accepting vertices,
and each arc is labeled with a symbol from $\Sigma$. An NFA accepts a string
$\chi$ if and only if there exists a path from the initial vertex to an
accepting vertex, such that the concatenation of its edge labels is $\chi$. A
DFA (Deterministic Finite Automaton) is a special case of an NFA where for each
vertex there exists exactly one outgoing labeled with each character of
$\Sigma$.

\begin{theorem}\label{thm:fas}

The following statements are equivalent:

\begin{enumerate}

\item The \ppseth\ is false.

\item There exist $\eps>0, k>1, c>0$ and an algorithm that takes as input the
descriptions of $k$ DFAs of $n$ states working over the alphabet
$\Sigma=\{0,1\}$ and an integer $\ell$ and decides if there exists a string of
length at most $\ell$ accepted by all $k$ DFAs in time $O(n^{k-\eps}\ell^c)$.

\item There exist $\eps>0, k_0>1, c>0$ and an algorithm which for all $k>k_0$
takes as input the descriptions of $k$ NFAs of $n$ states and an integer $\ell$
and decides if there exists a string of length at most $\ell$ accepted by all
$k$ NFAs in time $O(n^{(1-\eps)k}\ell^c)$.

\end{enumerate}

\end{theorem}

The implication $3\Rightarrow 2$ is trivial, so we establish the theorem in two
separate lemmas showing that $2\Rightarrow 1$ and $1\Rightarrow 3$.

\begin{lemma}\label{lem:sattodfa} 

If there exist $\eps>0, k>1, c>0$ and an algorithm that takes as input the
descriptions of $k$ DFAs of $n$ states working over the alphabet
$\Sigma=\{0,1\}$ and an integer $\ell$ and decides if there exists a string of
length at most $\ell$ accepted by all $k$ DFAs in time $O(n^{k-\eps}\ell^c)$,
then there exist $\eps'>0, c'>0$ and an algorithm that takes as input a 3-CNF
formula $\psi$ and a path decomposition of width $p$ and decides the
satisfiability of $\psi$ in time $O((2-\eps')^p|\psi|^{c'})$.

\end{lemma}

\begin{proof}

Fix the $\eps>0, k>1, c$ for which the algorithm that decides the intersection
non-emptiness of $k$ DFAs exists. Suppose we are given a 3-CNF formula $\psi$
on $n$ variables and a path decomposition of width $p$. We will construct $k$
DFAs $D_1,\ldots, D_k$, over an alphabet $\Sigma=[k]\cup\{T,F\}$, and an
integer $\ell$ such that we achieve the following:

\begin{enumerate}

\item There exists a string of length at most $\ell$ accepted by all $k$ DFAs
if and only if $\psi$ is satisfiable.

\item For all $i\in [k]$, the number of states of $D_i$ is $O(2^{p/k}n^4)$. All
DFAs have the same number of states.

\item $\ell=O(n^4)$.

\item The DFAs and $\ell$ can be constructed in time $O(2^{p/k}n^4)$.

\end{enumerate}

Before we proceed, let us explain why the above imply the lemma. To decide
$\psi$ we will use the supposed algorithm to decide if there exists a string of
length at most $\ell$ accepted by all DFAs. The running time of this procedure
is dominated by the execution of this supposed algorithm, which takes time
$O(|D_1|^{k-\eps}\ell^c) = O((2^{p/k}n^4)^{k-\eps}n^{4c})$. Set $2-\eps' =
2^{\frac{k-\eps}{k}}$ and the running time is at most
$O((2-\eps')^pn^{4k+4c})$. Setting $c'=4k+4c$ we obtain the promised bound, and
we have that $\eps',c'$ are fixed constants depending only on $\eps,k, c$.

We also note that the lemma is stated for a binary alphabet, but we are using
an alphabet of size $k+2$. This is not important, as we can replace the
transitions out of each state with a binary tree of height
$\lceil\log|\Sigma|\rceil$, with the two transitions out of each node labeled
$0$ and $1$, and associate each binary string of length
$\lceil\log|\Sigma|\rceil$ with a character from $\Sigma$. For each leaf of the
tree we look at the path from the root that leads to this leaf and associate a
character $c\in\Sigma$ with this leaf. We can then identify this state with the
state that would have been reached in the original transition for character
$c$. The details of this construction are explained in \cite{BodlaenderGNS21}.
We therefore prefer to continue with a non-binary alphabet to simplify
presentation.

We invoke \cref{lem:nice} and \cref{lem:pwcolor} on our decomposition to obtain
a nice decomposition $B_1,\ldots, B_t$ with $t=O(n^4)$, to obtain a function
$b$ that inductively maps clauses of $\psi$ to indices of bags that contain all
their variables, and to partition the variable set of $\psi$ into $k$ sets
$V_1,\ldots, V_k$ such that for all $i\in[k], j\in[t]$ we have $|B_j\cap
V_i|\le \lceil \frac{p+1}{k}\rceil$. Without loss of generality, we will assume
that $|B_1|=1$. This property can always be assured, because if $|B_1|>1$, then
we can add a sequence of $|B_1|-1$ bags before $B_1$, where the first bag
contains one element of $B_1$ and we construct each subsequent bag by adding to
the previous one a new element of $B_1$.

Fix an $i\in[k]$. We describe the construction of the DFA $D_i$. For each $j\in
[t]$ we consider every truth assignment $\sigma$ to the variables of $V_i\cap
B_j$ (there are exactly $2^{|B_j\cap V_i|}$ such assignments).  For each such
assignment we construct two states $s_{i,j,\sigma}^0, s_{i,j,\sigma}^1$.  We
also add an initial state $s_i^I$. We make all states $s_{i,\ell,\sigma}^1$
accepting states. We now add the following transitions:

\begin{enumerate}

\item (Initial state): We add a transition for $s_i^I$ to all states
$s_{i,1,\sigma}^0$.  Let $B_1=\{v\}$. If $v\not\in B_i$, then there is in fact
exactly one state $s_{i,1,\sigma}^0$ (for $\sigma$ the empty assignment), and
we label the transition from $s_i^I$ to this state with both $T$ and $F$.
Otherwise, there are exactly two states $s_{i,1,\sigma}^0$ and
$s_{i,1,\sigma'}^0$ where $\sigma$ sets $v$ to True and $\sigma'$ sets $v$ to
False. We label the transition to the former state with $T$ and to the latter
state with $F$.

\item (Consistency): consider a $j\in[t-1]$ and every assignment $\sigma$ to
$B_j\cap V_i$ and $\sigma'$ to $B_{j+1}\cap V_i$. For each such pair of
assignments which are compatible, we add a transition from $s_{i,j,\sigma}^1$
to $s_{i,j+1,\sigma'}^0$. If $B_{j+1}\cap V_i \subseteq B_j\cap V_i$, there is
a unique transition going out of each $s_{i,j,\sigma}^1$, which we label with
both $T$ and $F$. Otherwise, each $s_{i,j,\sigma}^1$ has transitions to two
distinct $s_{i,j+1,\sigma'}^0$ and $s_{i,j+1,\sigma''}^0$, where $\sigma',
\sigma''$ disagree on the value of variable $v\in B_{j+1}\setminus B_j$. We
give label $T$ to the transition leading to the assignment that sets $v$ to
True and $F$ to the other transition.

\item (Satisfaction): for each $j\in[t]$ and all $\sigma$ we add a transition
from $s_{i,j,\sigma}^0$ to $s_{i,j,\sigma}^1$. If no clause $c$ has $b(c)=j$,
or if such a clause exists but $\sigma$ sets one of the literals of $c$ to
True, we label the transition with $[k]$. Otherwise, that is if some $c$ has
$b(c)=j$ but $\sigma$ sets no literal of $c$ to True, we label the transition
with $[k]\setminus\{i\}$.

\end{enumerate}

We note that in the above there exist states $s$ such that for some
$c\in\Sigma$ no transition with label $c$ is exiting $s$. For each such state,
we add a junk state $s_{\textrm{junk}}$ with self-loops labeled with all
characters and add the missing transitions to this state. The intended meaning
is that in case of a missing transition the DFA rejects. To complete the
construction we set $\ell = 2t$.  It is not hard to see that we have respected
the time bounds and it is possible to construct the DFAs in the claimed time.
It is also not hard to check that the automata are deterministic, that is, for
each state $s$ and $c\in\Sigma$ there is at most one transition with label $c$
going out of $s$.

What remains is to argue for correctness. Suppose that $\psi$ is satisfiable.
We construct a string $\chi$ of length exactly $\ell$, numbering its characters
from left to right starting at $0$. For $j\in[t-1]$ the characters at position
$2j$ and $2j+1$ are determined as follows:

\begin{itemize}

\item If $j=0$ or $(B_{j+1} \setminus B_j)\cap V_i\neq \emptyset$, let $v$ be
the new vertex, that is $v\in B_{j+1}\setminus B_j$ or $v$ is the unique
element of $B_1$.  If the satisfying assignment sets $v$ to True, then
character $2j$ is $T$, otherwise $F$. In other cases ($j>0$ and $B_{j+1}\cap
V_i\subseteq B_j$) we set character $2j$ to $T$.

\item If there exists clause $c$ such that $b(c)=j$, then let $i$ be such that
$c$ contains a satisfied literal using a variable of $V_i\cap B_j$. Such an $i$
must necessarily exist, since all variables of $c$ are contained in $B_j$.
Then, character $2j+1$ is $i$. If no such clause exists, then character $2j+1$
is $1$.

\end{itemize}

We now want to argue that each DFA $D_i$ will accept $\chi$.  For this it is
sufficient to observe that we never use a non-existing transition (which would
lead to a junk state), as all states at distance $2\ell$ from the initial state
are accepting. We claim that reading the constructed string the DFA will always
be in a state $s_{i,j,\sigma}^\alpha$, with $\alpha\in\{0,1\}$, such that
$\sigma$ is the restriction of the satisfying assignment to $V_i\cap B_j$. This
is easy to see for $j=1$ and can be shown by induction for all $j$ using the
construction of the string. The key observation we need here is that in all
interesting transitions (from $s_{i,j,\sigma}^1$ to $s_{i,j,\sigma'}^0$ with
$B_{j+1}\cap V_i$ containing a new variable) we have constructed $\chi$ so that
$\sigma'$ continues to agree with our satisfying assignment. We also note that
after an even number of characters the DFA is at a state $s_{i,j,\sigma}^1$,
where it will encounter by construction the characters $T$ or $F$; all such
states have an outgoing transition for both characters.  Finally, for
transitions out of states $s_{i,j,\sigma}^0$, the DFA will encounter a
character from $[k]$, while the state will have a transition labeled $[k]$ (in
which case we have nothing to argue), or $[k]\setminus\{i\}$.  In the latter
case, the character of our string cannot be $i$, because the assignment
$\sigma$ falsifies all literals of the clause $c$ with $b(c)=j$, so we must
have selected another character at this position.

For the converse direction, suppose that the $k$ DFAs accept some string
$\chi$.  We extract an assignment for $V_i$ by observing a run of $D_i$ on
$\chi$: if the DFA enters state $s_{i,j,\sigma}^0$ we use $\sigma$ for the
variables of $B_j\cap V_i$. This is consistent, because of the properties of
path decompositions (once a variable is forgotten, it does not reappear) and
because transitions always lead to compatible assignments. We claim that the
union of these $k$ assignments satisfies $\psi$. To see this, take a clause $c$
and suppose $b(c)=j$. Consider the character of $\chi$ which allows the $k$
DFAs to advance from some $s_{i,j,\sigma}^0$ to $s_{i,j,\sigma}^1$. Suppose
this character is $i$ and we examine the construction of $D_i$. Because the
transition from $s_{i,j,\sigma}^0$ to $s_{i,j,\sigma}^1$ is labeled $[k]$ and
not $[k]\setminus\{i\}$, while $b(c)=j$, it must be the case that $c$ contains
a literal that is satisfied by the assignment $\sigma$, therefore $\psi$ is
satisfied.  \end{proof}

\begin{lemma}\label{lem:nfatosat}

If there exist $\eps>0, c>0$ such that there exists an algorithm that takes as
input a CNF formula $\psi$ and decides its satisfiability in time
$O((2-\eps)^{\pw(\psi)}|\psi|^c)$, then there exist $\eps'>0, k_0>1,c'>0$ such
that there exists an algorithm which for any $k>k_0$ takes as input the
descriptions of $k$ NFAs of $n$ states and an integer $\ell$ and decides if
there exists a string of length at most $\ell$ accepted by all $k$ NFAs in time
$O(n^{(1-\eps)k})\ell^{c'}$.

\end{lemma}

\begin{proof}

Fix the $\eps,c$ for which a fast \textsc{SAT} algorithm exists. Suppose we are
given $k$ NFAs, each with $n$ states, using an alphabet $\Sigma$ of constant
size, and an integer $\ell$. We want to construct a CNF formula $\psi$ that
satisfies the following:

\begin{enumerate}

\item $\psi$ is satisfiable if and only if there exists a string of length
exactly $\ell$ that is accepted by all $k$ NFAs.

\item $|\psi| = O(n^2\ell)$ and $\psi$ can be constructed in time $O(n^2\ell)$.

\item $\pw(\psi) = (k+1)\log n + O(|\Sigma|)$.

\end{enumerate}

Before we proceed, let us explain why the above imply the lemma. We will decide
if such a string exists by constructing $\psi$ and checking whether it is
satisfiable using the supposed algorithm. Because $\psi$ is satisfiable only if
a string of length \emph{exactly} $\ell$ is accepted by all NFAs, while we are
looking for a string of length at most $\ell$, we repeat this procedure $\ell$
times, once for each possible length in $[\ell]$. The running time of the whole
procedure is then $O((2-\eps)^{\pw(\psi)}|\psi|^c\ell) = O( (2-\eps)^{k\log n}
n^{2c+1}\ell^{c+1})$.  Let $\delta>0$ be such that $2-\eps = 2^{1-\delta}$, in
particular set $\delta=1-\log(2-\eps)$.  The running time is at most
$O(n^{(1-\delta)k+2c+1}\ell^{c+1})$. Set $k_0 = \frac{2(2c+1)}{\delta}$.  If
$k>k_0$ we have that $n^{(1-\delta)k+2c+1}\ell^{c+1} <
n^{(1-\delta/2)k}\ell^{c+1}$ so setting $\eps'=\delta/2$, $c'=c+1$ we obtain an
algorithm with the promised running time.  Observe that $p_0, \eps', c'$ are
constants that depends only on $\eps, c$.

We now describe the construction of $\psi$. Assume $n$ is a power of $2$
(otherwise, add dummy unreachable states to each NFA, at most doubling its
size).  For each $j\in\{0,\ldots,\ell\}$, we construct $k\log n$ variables
$x_{i,j,s}$, with $i\in[k]$ and $s\in[\log n]$. Informally, the variables
$x_{i,j,s}$ encode the state of the $i$-th NFA after reading the $j$ first
letters of the supposed string. Suppose that the states of the $k$ NFAs are
numbered $s_{i,0}, s_{i,1},\ldots, s_{i,n-1}$, for $i\in[k]$ with $s_{i,0}$
being the initial state for each $i\in [k]$. Furthermore, for each $j\in[\ell]$
and $c\in\Sigma$ we construct a variables $y_{c,j}$ which informally encodes if
the $j$-th letter of the supposed accepted string is $c$.

\begin{enumerate}

\item (Initial states): We add clauses for each $i\in [k]$ that ensure that
$x_{i,0,s}$ encodes the initial state. In particular we add for each $i\in
[k]$, $n-1$ clauses, such that the only satisfying assignment sets all of
$x_{i,0,s}$ to False, which encodes the state $s_{i,0}$.

\item (Character selection): For each $j\in[\ell]$ we add the clause
$\bigvee_{c\in\Sigma} y_{c,j}$, and for each $c,c'\in \Sigma$ with $c\neq c'$
the clause $(\neg y_{c,j} \lor \neg y_{c',j})$.

\item (Transitions): For each $j\in [\ell]$ and $i\in[k]$ we add clauses that
ensure that the transition of the $i$-th NFA is correct. In particular, for
each $c\in \Sigma$ and pair of assignments $\sigma,\sigma'$ to the variables
$x_{i,j-1,s}$ and $x_{i,j,s}$ we check if the state encoded by $\sigma'$ can be
reached from the state encoded by $\sigma$ if we read character $c$. If this is
not the case, we add a clause that is falsified if we have $y_{c,j}$ and
assignments $\sigma,\sigma'$.

\item (Accepting states): For each $i\in[k]$ and each assignment $\sigma$ to
$x_{i,\ell,s}$, if $\sigma$ encodes a state that is not accepting, then we add
a clause that is falsified by this assignment to the variables $x_{i,\ell,s}$.

\end{enumerate}

The total number of clauses we have added is at most $O(n^2\ell)$, as step 3
dominates. It is not hard to see that the construction can be performed in time
$O(n^2\ell)$. Let us bound the pathwidth of $\psi$. For each
$j\in\{0,\ldots,\ell\}$ we construct a bag $B_j$ which contains all
$x_{i,j,s}$, as well as all $y_{c,j}$ (for $j>1$). This covers all the clauses
except those added to encode the transitions. To cover these, for each
$j\in[\ell]$ we add a sequence of bags between $B_{j-1}$ and $B_{j}$, call them
$B_j^1, B_j^2,\ldots, B_j^k$. Bag $B_j^i$ contains $y_{c,j}$ for all $c\in
\Sigma$, for each $i'\in[k], i'\le i$ the variables $x_{i',j-1,s}$, and for
each $i'\in[k], i'\ge i$ the variables $x_{i',j,s}$. Observe that for each
$i\in [k]$, the transition clauses constructed for $i$ are covered in $B_j^i$.
It is not hard to see that the largest bag has $(k+1)\log n+|\Sigma|$
variables.

Finally, the correctness of the construction is straightforward. If there
exists a string of length $\ell$ accepted by all NFAs, we set $y_{c,j}$ to True
if and only if the $j$-th character of that string is $c$. This satisfies all
clauses of step 2. We set all $x_{i,0,s}$ and this satisfies clauses of step 1.
Then, for each $i\in[k]$ fix an accepting run of the $i$-th NFA and set for
each $j\in[\ell]$ the variables $x_{i,j,s}$ to encode the state of this NFA in
this run after reading $j$ characters. This satisfies the transition clauses
and the clauses of the last step. For the converse direction, if $\psi$ is
satisfiable, then we can extract a string from the assignments of $y_{c,j}$, as
for each $j\in [\ell]$ exactly one such variable is True. We claim that this
string is accepted by all NFAs. Indeed, fix an $i\in[k]$ and we claim that the
states encoded by $x_{i,j,s}$ must be an accepting run of this state for the
$i$-th NFA. This follows because the clauses of the first step ensure that we
start at $s_{i,0}$; the transition clauses ensure that we always perform legal
transitions; and the clauses of the last step ensure that we end up in an
accepting state.  \end{proof}

\subsection{Independent Set Reconfiguration}\label{sec:reconf}

In this section we consider a \emph{reconfiguration} problem defined as
follows: we are given a graph $G$ on $n$ vertices and two independent sets of
$G$, call them $S,T$, each of size $k$, for $k$ being a fixed integer. We would
like to transform $S$ into $T$ via token jumping moves, where a token jumping
move is an operation that replaces a vertex in the current set with a vertex
currently outside the set, while keeping the current set independent at all
times.

Reconfiguration problems such as this one have been extensively studied in the
literature and are generally PSPACE-complete when $k$ is allowed to vary with
$n$ (we refer the reader to the surveys by Nishimura \cite{Nishimura18} and van
den Heuvel \cite{Heuvel13}). However, when $k$ is fixed, such problems become
solvable in polynomial time, as the total number of configurations is at most
$n\choose k$ and we can solve reconfiguration problems by exploring the
configuration graph. It is known that, under standard assumptions, this
complexity cannot be improved to $f(k)n^{O(1)}$, as independent set
reconfiguration is generally W[1]-hard, even if the length of the configuration
is an additional parameter (\cite{ItoKOSUY20,MouawadN0SS17}, see also the
recent survey on the parameterized complexity of reconfiguration
\cite{abs-2204-10526}). Indeed, recent work of Bodlaender, Groenland, and
Swennenhuis \cite{BodlaenderGS21} succeeded in characterizing precisely the
parameterized complexity of independent set reconfiguration when the parameter
is the size of the set $k$: the problem is XL-complete when no bound on the
length of the sequence is given; XNL-complete when a bound is given in binary;
XNLP-complete when a bound is given in unary; and W[1]-complete when the length
is also a parameter.

The previous work summarized above seems to indicate that the $n^{O(k)}$
algorithm which constructs the configuration graph is best possible. The
question we ask here is, however, more fine-grained: we would like to determine
the minimum constant $c$ for which the problem admits an $n^{ck}$ algorithm.
Note that, because previous hardness results start from $k$-\textsc{Clique}
(which does admit an $n^{ck}$-time algorithm for $c<1$), no concrete
fine-grained lower bound of this form seems to be currently known.

In this section we focus on the version of the reconfiguration problem where we
are given a desired length $\ell$ for the reconfiguration sequence, and $\ell$
is not too large (this corresponds to the XNLP-complete case studied in
\cite{BodlaenderGS21}). Our main result is that it \emph{is} in fact possible
to do slightly better than the trivial algorithm in this case. However, going a
little further than what we achieve is likely to be very hard, as we show that
significant further improvements are equivalent to falsifying the \ppseth.

\subparagraph*{Algorithmic result} Before we explain the main algorithmic
result of this section, it is worthwhile to pause for a second and take a
closer look at the complexity of the trivial algorithm. As mentioned, whether a
bound $\ell$ on the length of the sequence is given or not, one way to deal
with this problem is to construct a configuration graph that contains
${n\choose k}=O(n^k)$ vertices, one for each independent set of $G$, and has an
edge between sets which are reachable in one token jumping move. If $k$ is
fixed, this graph has $O(n^{k+1})$ edges, because each set has $O(kn)$
neighbors (we have $k$ choices for a vertex to remove and $n$ choices for a
vertex to add). Hence, using BFS, reconfiguration can be solved in $O(n^{k+1})$
time. This can be slightly improved as follows: we can construct a bipartite
graph, with independent sets of size $k$ on one side and independent sets of
size $k-1$ on the other, and place an edge between $S_1,S_2$ whenever
$S_1\subseteq S_2$. We now observe that two sets $S_1,S_2$ of size $k$ are at
distance $2$ in the new graph if and only if they are reachable in one token
jumping move. Hence, a reconfiguration sequence of length $\ell$ corresponds
exactly to a path of length at most $2\ell$ in this graph. This graph only has
$O(n^k)$ edges, since each set of size $k$ has $k$ subsets of size $k-1$, so
running BFS on this graph saves a factor of $n$.

The main question we ask then is whether we can further improve this $O(n^k)$
algorithm to some running time of the form $n^{k-\eps}$. The answer is
positive, if we accept a small cost in terms of the length of the
reconfiguration sequence: we show in \cref{thm:reconf-alg} that, for all $k\ge
3$ the problem can be solved in time $O(n^{k-3+\omega}\cdot\ell)$, where
$\omega$ is the matrix multiplication constant. In particular, this means that
whenever $\ell$ is small (say $\ell=n^{o(1)}$), we can shave a constant off the
exponent of $n$ that is slightly less than $1$ (and would be equal to $1$ if
$\omega=2$).

\subparagraph*{Lower bound} The above discussion naturally leads to the lower
bound question we ask in this section. Suppose that the correct value of the
matrix multiplication constant is $\omega=2$. Then, the reconfiguration problem
for small values of $\ell$ can be solved slightly faster than the easy $O(n^k)$
algorithm we described above as we obtain an algorithm of running time
$O(n^{k-1}\cdot\ell)$. Is there a possibility to do even better? Our main
result in this section is to give strong evidence that the answer is negative:
obtaining an algorithm with complexity $n^{k-1-\eps}\cdot\ell^{O(1)}$ is
\emph{equivalent} to falsifying the \ppseth. Hence, our results indicate that
(when $\ell$ is small) the correct dependence of the exponent of $n$ on $k$ is
exactly $k-1$ (modulo, of course, the standard conjecture on the value of
$\omega$).  Note that our lower bound is strong in the sense that we rule out
an algorithm which shaves an $\eps$ off the exponent of $n$, even if we allow
\emph{any} polynomial dependence on $\ell$ (contrast this with the linear
dependence on $\ell$ of \cref{thm:reconf-alg}). Our lower bound, given in
\cref{thm:reconf}, works for all fixed values of $k\ge 4$ and as in the case of
other XNLP-complete problems shows that it is equivalent to shave an additive
constant $\eps$ off the exponent of $n$ and to improve the \emph{coefficient}
of $k$ to $(1-\eps)$ (which a priori sounds like a more important improvement).

We now go on to present the proofs of the results of this section.

\begin{theorem}\label{thm:reconf-alg}

There exists an algorithm that takes as input a graph $G$ on $n$ vertices, two
independent sets $S,T$ of size $k\ge 3$, and an integer $\ell$ and decides if
it is possible to transform $S$ to $T$ via at most $\ell$ token jumping moves
in time $O(n^{k-3+\omega}\cdot\ell)$, where $\omega$ is the matrix multiplication
constant.

\end{theorem}

\begin{proof}

Suppose that a transformation sequence $S_0=S,S_1,\ldots, S_\ell=T$ exists, we
will concentrate on the sequence of sets $X_i=S_{i-1}\cap S_{i}$, for
$i\in[\ell]$, which all have $|X_i|=k-1$. We note that the transformation
sequence is valid if we have $X_1\subseteq S$, $X_\ell\subseteq T$, and for
each $i\in[\ell-1]$ we have that there exists $x_{i+1}$ such that
$X_{i+1}\setminus X_{i}=\{x_{i+1}\}$ but $X_i\cup \{x_{i+1}\}$ is an
independent set. In other words, the sequence is valid if we can transform each
$X_i$ into $X_{i+1}$ by first adding a vertex to $X_i$ (while keeping the set
independent) and then removing a vertex. We call such a move, which is slightly
more restrictive than token jumping, because we need to maintain an independent
set before we remove the old token, an addition/removal move.

We start by guessing $X_1$ (there are only $k$ choices, as $X_1\subseteq S$ and
$|X_1|=k-1$). We will construct a sequence of matrices $R_i$, which will encode
which independent sets are reachable from $X_1$ within $i$ addition/removal
moves. The matrices $R_i$ will have ${n\choose k-2} = O(n^{k-2})$ rows and $n$
columns each, where rows are indexed by sets of vertices of size $k-2$ and
columns by vertices. We initialize $R_0[Y,x]=1$ if and only if $Y\cup\{x\} =
X_1$ and we set all other values to $0$ (so $R_0$ contains $k-1$ non-zero
entries). 

We now explain how to obtain $R_{i+1}$ from $R_i$, assuming that $R_i[Y,x]=1$
if and only if $Y\cup\{x\}$ is reachable from $X_1$ within at most $i$
addition/removal moves. Let $A$ be the complement of the adjacency matrix of
$G$, that is $A$ is an $n\times n$ matrix which has $A[x,y]=1$ if and only if
$xy\not\in E$ (in particular, $A[x,x]=1$ for all $x\in V$). We compute the
matrix $P_i=R_i\cdot A$, where addition is replaced by binary disjunction, and
multiplication by binary conjunction, and then set $R_{i+1}[Y,x]=1$ if and only
if $Y\cup\{x\}$ is an independent set and there exist $Y',x'$ such that
$Y\cup\{x\}=Y'\cup\{x'\}$ and $P_i[Y',x']=1$.  Note that since $|Y|=k-2$, there
are only $k-1$ combinations $(Y',x')$ to check for each pair $(Y,x)$.

We claim that $R_{i+1}$ contains the desired information, that is
$R_{i+1}[Y,x]=1$ if and only if $Y\cup\{x\}$ is an independent set that is
reachable from $X_1$ within at most $i+1$ addition/removal moves. First, we
argue that if $Y\cup \{x\}$ is a set reachable via $i+1$ moves from $X_1$, then
$R_{i+1}[Y,x]=1$. To see this, suppose that $x$ is the last vertex added to the
set in such a configuration sequence (this is without loss of generality,
because if $y\in Y$ is actually the last vertex, then
$R_{i+1}[Y\setminus\{y\}\cup\{x\},y]=1$ by the construction of $R_{i+1}$).
Since $x$ is the last vertex added, the previous set in the sequence must be of
the form $Y\cup \{x'\}$, which is reachable in at most $i$ moves, so
$R_i[Y,x']=1$. Furthermore, it must be the case that $xx'\not\in E$. We then
have that $P_i[Y,x] = 1$ because $P_i[Y,x] = \lor_{z\in V} R_i[Y,z]\land
A[z,x]$ and the conjunction evaluates to $1$ for $z=x'$. Because $Y\cup\{x\}$
is independent, we have $R_{i+1}[Y,x]=1$ as desired.

Second, we argue that if $R_{i+1}[Y,x]=1$, then there does exist a sequence of
at most $i+1$ addition/removal moves that transforms $X_1$ to $Y\cup\{x\}$. If
$R_{i+1}[Y,x]=1$, then there exist $Y',x'$ with $Y\cup\{x\} = Y'\cup\{x'\}$
such that $P_i[Y',x']=1$. Suppose without loss of generality that $P_i[Y,x]=1$.
Then, since $P_i[Y,x] = \lor_{z\in V} R_i[Y,z]\land A[z,x]$ there exists $z\in
V$ such that $R_i[Y,z]=1$ and $xz\not\in E$. The set $Y\cup\{z\}$ is reachable
from $X_1$ in at most $i$ addition/removal moves, so we make a further move,
adding $x$ to $Y\cup\{z\}$ and then removing $z$. Because $xz\not\in E$ and
$R_{i+1}[Y,x]=1$ we have that $Y\cup\{z,x\}$ is independent, so we obtain a
valid sequence of at most $i+1$ addition/removal moves that transform $X_1$ to
$Y\cup\{x\}$.

We now iteratively compute $R_i$ for $i\in[\ell]$. If there exist $Y,x$ such
that $R_{\ell}[Y,x]=1$ and $Y\cup\{x\}\subseteq T$ we answer Yes, otherwise we
answer No. Correctness follows from the correctness of $R_i$ for $i\in[\ell]$,
so what remains is to compute the running time. We iterate over $i\in [\ell]$
and at each step we multiply a matrix of size $n^{k-2}\times n$ with a matrix
of size $n\times n$. The running time of this multiplication dominates other
operations for each iteration, which can be achieved in $O(n^{k-1})$ time.
Calculating the product of an $n^{k-2}\times n$ matrix with an $n\times n$
matrix can be broken down into $n^{k-3}$ multiplications of $n\times n$
matrices, so the running time is $O(n^{k-3+\omega}\ell)$.  \end{proof}

\begin{theorem}\label{thm:reconf}
The following statements are equivalent:

\begin{enumerate}

\item The \ppseth\ is false.

\item There exist $\eps>0, k\ge 4$ and an algorithm that takes as input a graph
$G$ on $n$ vertices, two independent sets $S,T$ of $G$ of size $k$, and an
integer $\ell$, and decides if $S$ can be transformed to $T$ via at most $\ell$
token jumping moves in time $n^{k-1-\eps}\ell^{O(1)}$.

\item There exist $\eps>0, k_0\ge 4$ and an algorithm which for all $k\ge k_0$
takes an input as in the previous case and returns a correct response in time
$n^{(1-\eps)k}\ell^{O(1)}$.

\end{enumerate}

\end{theorem}

Since the implication $3\Rightarrow 2$ is trivial (set
$k=\max\{k_0,\lceil2/\eps\rceil\}$), we break down the proof of
\cref{thm:reconf} into two lemmas from the implications $2\Rightarrow 1$ and
$1\Rightarrow 3$.

\begin{lemma}\label{lem:reconf13} If the \ppseth\ is false, then there exist
$\eps>0, k_0\ge 4$ and an algorithm that takes as input a graph $G$ on $n$
vertices, two independent sets $S,T$ of $G$ of size $k\ge k_0$, and an integer
$\ell$, and decides if $S$ can be transformed to $T$ via at most $\ell$ token
jumping moves in time $n^{(1-\eps)k}\ell^{O(1)}$.  \end{lemma}

\begin{proof}

Suppose that the \ppseth\ is false and there exists an algorithm that takes as
input a \textsc{SAT} instance $\phi$ and a path decomposition of its primal
graph of width $p$ and decides if $\phi$ is satisfiable in time
$O(2^{(1-\delta)p})|\phi|^{O(1)}$.  Suppose we are given an instance of the
reconfiguration problem for a graph $G$ on $n$ vertices and two sets $S,T$ of
size $k$, with a target length $\ell$. We will construct a formula $\phi$ and a
path decomposition of its primal graph such that:

\begin{enumerate}

\item $\phi$ is satisfiable if and only if there exists a token-jumping
reconfiguration of length at most $\ell$ between $S$ and $T$ in $G$.

\item The constructed path decomposition of $\phi$ has width at most $k\log
n+O(k)$

\item $\phi$ has size polynomial in $n+\ell$ and can be constructed in
polynomial time.

\end{enumerate}

If we achieve the above we obtain the lemma, because we can decide the
reconfiguration problem by constructing $\phi$ and then deciding its
satisfiability via the supposed algorithm. Indeed, this would take time
$O(2^{(1-\delta)k\log n}n^{O(1)}\ell^{O(1)}) =
n^{(1-\delta)k+O(1)}\ell^{O(1)}$. Suppose now that the constant hidden in the
$O(1)$ in the exponent of $n$ is $c$. Let $k_0=\lceil2c/\delta\rceil$. Then,
the running time becomes $n^{(1-\delta/2)k}\ell^{O(1)}$, giving the lemma.

Assume, to ease presentation, that $n$ is a power of $2$ (otherwise add some
dummy universal vertices, increasing $n$ by at most a factor of $2$). We
construct a CNF formula as follows:

\begin{enumerate}

\item For each $j\in[\ell]$, $i\in[k]$ we construct $k$ groups of $\log n$
variables each, call them $X_{j,i}$. The idea is that for a given $j,i$ the
$\log n$ variables of $X_{j,i}$ encode the index of the $i$-th vertex of the
$j$-th independent set in the reconfiguration sequence.

\item For each $j\in[\ell-1]$ we construct $k$ variables $c_{j,i}$, for
$i\in[k]$. Informally, $c_{j,i}$ is true if we intend to remove the $i$-th
token of the $j$-th set of the sequence to obtain the following set.

\item For each $j\in[\ell-1]$ and $i,i'\in [k]$ with $i\neq i'$ we add the
clause $(\neg c_{j,i} \lor \neg c_{j,i'})$.

\item For each $j\in[\ell-1]$ and $i\in [k]$ we add clauses ensuring that
$X_{j,i}$ and $X_{j+1,i}$ encode the same value (this can be done with clauses
of size $2$ by enforcing this constraint bit by bit). We add the literal
$c_{j,i}$ to all these clauses.

\item For each $j\in[\ell]$  and $i,i'\in [k]$, for each edge of $G$, we
construct a clause involving $X_{j,i}\cup X_{j,i'}$ that is falsified if
$X_{j,i}, X_{j,i'}$ encode the endpoints of the edge. Do the same for each pair
$(x,x)$ where $x\in V(G)$ (that is, perform this step as if $G$ contains a
self-loop in every vertex).

\item Add clauses ensuring that for all $i\in[k]$, $X_{1,i}$ encodes the $i$-th
element of $S$ (for some arbitrary ordering of $S$).

\item Add clauses ensuring that for all $i\in[k]$, $X_{\ell,i}$ encodes a
vertex of $T$ (that is, add a clause that is falsified if $X_{\ell,i}$ encodes
the index of a vertex outside of $T$).

\end{enumerate}

The correctness of the construction is now easy to verify. In one direction, if
there is a reconfiguration sequence, we give the variables $X_{j,i}$ values
that reflect the position of the $i$-th token in the $j$-th set of the
sequence. We assign $c_{j,i}$ to True if and only if the $i$-th token of the
$j$-th set jumps to obtain the next set. It is not hard to see that this
satisfies all clauses. For the converse direction, by interpreting the
assignments to $X_{j,i}$ as indices of vertices we infer a sequence of sets of
size at most $k$. These sets must be independent and have size exactly $k$
because of the clauses of step 5. Furthermore, if $c_{j,i}$ is False, then the
$j$-th and $(j+1)$-th set agree on their $i$-th vertex. However, at most one
$c_{j,i}$ is True for each $j$, because of clauses of step 3. Therefore, the
sets we have constructed can be obtained from each other with token jumping
moves.

Finally, let us argue about the pathwidth of $\phi$. We can construct a path
decomposition by making a bag for each $j\in[\ell]$ and placing into it all
variables of $\bigcup_{i\in[k]}X_{j,i}$ as well as the $k$ variables $c_{j,i}$
for $i\in[k]$. This covers clauses of steps 3,5,6, and 7. To cover the
remaining clauses (of step 4) it suffices to observe that the subgraph of the
primal graph induced by $X_{j,i}\cup X_{j+1,i}$ is a matching, so we can
cosntruct a sequence of bags that exchanges one by one the vertices of
$X_{j,i}$ with those of $X_{j+1,1}$, while keeping $c_{j,i}$ in the bags,
covering clauses of step 4.  \end{proof}

\begin{lemma}\label{lem:reconf21} If there exist $\eps>0, k\ge 4$ and an
algorithm that takes as input a graph $G$ on $n$ vertices, two independent sets
$S,T$ of $G$ of size $k$, and an integer $\ell$, and decides if $S$ can be
transformed to $T$ via at most $\ell$ token jumping moves in time
$n^{k-1-\eps}\cdot\ell^{O(1)}$, then the \ppseth\ is false.  \end{lemma}

\begin{proof}

Suppose that there exist $\eps>0, k\ge 3$ and an algorithm that solves the
reconfiguration problem in time $n^{k-1-\eps}\cdot\ell^{O(1)}$.  We start with
an instance of \tsat\ $\phi$ and suppose that we are given a path decomposition
of its primal graph where each bag contains at most $p$ variables. We will
construct from $\phi$ an instance of the reconfiguration problem, that is a
graph $G$, an integer $\ell$, and two independent sets $S,T$ of size $k$ such
that:

\begin{enumerate}

\item $\phi$ is satisfiable if and only if $S$ can be reconfigured to $T$ with
at most $\ell$ token jumping moves.

\item $G$ has $n$ vertices, where $n=2^{\frac{p}{k-1}}|\phi|^{O(1)}$. 

\item $\ell = |\phi|^{O(1)}$

\item Constructing $G$, $S,T,\ell$, can be done in time at most
$2^{2p/3}|\phi|^{O(1)}$.

\end{enumerate}

If we achieve the above, we falsify the \ppseth. Indeed, we can decide $\phi$
by performing the reduction and then using the supposed algorithm for the
reconfiguration problem. The running time is dominated by the execution of the
supposed algorithm, which takes time $n^{k-1-\eps}\ell^{O(1)} =
2^{\frac{p}{k-1}(k-1-\eps)}|\phi|^{O(1)} = (2-\delta)^p |\phi|^{O(1)}$, where
$2-\delta = 2^{\frac{k-1-\eps}{k-1}}$.

We now invoke \cref{lem:nice} and \cref{lem:pwcolor} to obtain a nice path
decomposition of $\phi$, with the bags numbered $B_1,\ldots, B_t$, with
$t=|\phi|^{O(1)}$, an injective function mapping clauses of $\phi$ to bags that
contain their variables, and a partition of the variables into $k-1$ groups
$V_1,\ldots,V_{k-1}$ so that for all bags $B_j$ we have $|B_j\cap V_i|\le
\frac{p}{k-1}+1$. We now perform some further preprocessing to ensure that each
clause of $\phi$ involves variables from at most $2$ distinct groups $V_i$. To
achieve this, as long as there is a clause $c$ involving three variables
$x_1,x_2,x_3$ which belong in distinct sets, let $B_j$ be the bag given by the
injective mapping for $c$. We insert in the decomposition $4$ copies of $B$
immediately after $B$ and insert in the first three copies a new variable
$x_1'$ which belongs in the same group as $x_2$. We replace in $c$ the variable
$x_1$ with $x_1'$ and construct the clauses $(\neg x_1\lor x_1')$ and $(\neg
x_1'\lor x_1)$, which ensure that $x_1,x_1'$ must have the same value. Using
the copies of $B$ where we added $x_1'$ we update the mapping function so that
$c$ and the two new clauses are mapped to distinct bags that contain their
literals. Doing this exhaustively we obtain the property, while we now have
that for each bag $B_j$ and $i\in[k-1]$ that $|B_j\cap V_i|\le
\frac{p}{k-1}+2$. We will further assume, without loss of generality, that
$B_1=B_t=\emptyset$.

In order to simplify the presentation of the rest of the construction, we will
formulate it in terms of clique reconfiguration (that is, $S,T$ are cliques of
size $k$ and we need to reconfigure them with token jumping moves while
maintaining a clique at all times). Clearly, this is equivalent to the original
problem if we take the complement of the graph we construct.

We now construct a graph as follows: for each $j\in[t]$ and $i\in[k-1]$ we
construct a set of $2^{|B_j\cap V_i|}=O(2^{\frac{p}{k-1}})$ vertices $V_{j,i}$
which form an independent set.  Each vertex represents an assignment to the
variables of $B_j\cap V_i$ (in particular, if $|B_j\cap V_i|=0$, then $V_{j,i}$
contains a unique vertex).  If some such assignment falsifies a clause (whose
variables are contained in $B_j\cap V_i$) we delete from the graph the
corresponding vertex.  We now add the following edges:

\begin{enumerate}

\item  For all $j\in[t]$ and $i,i'\in[k-1]$ with $i\neq i'$, we connect all
vertices of $V_{j,i}$ with all vertices of $V_{j,i'}$. 

\item We remove some of the edges added in the previous step as follows: for
each clause $c$ that involves variables from two groups $V_{i_1}, V_{i_2}$,
suppose that $c$ is mapped to $B_j$. For each pair $(u,v)\in V_{j,i_1}\times
V_{j,i_2}$ such that the assignment encoded by $u,v$ together falsifies $c$,
remove the edge $uv$ from the graph.

\item For all $j\in[t-1]$, $i\in[k-1]$, and $i'\in[i-1]$ we connect all
vertices of $V_{j,i}$ with all vertices of $V_{j+1,i'}$. 

\item For all $j\in[t-1]$ and $i\in[k-1]$ we connect each vertex of $V_{j,i}$
with each vertex of $V_{j+1,i}$ such that the corresponding assignments agree
on all variables of $B_j\cap B_{j+1}\cap V_i$.  

\end{enumerate}

Observe that as $B_1=B_t=\emptyset$, the sets $V_{1,i}$ and $V_{t,i}$ are
singletons which form two cliques of size $k-1$. We add a vertex $s_0$ and
connect it to $\bigcup_{i\in[k-1]}V_{1,i}$ and a vertex $t_0$ and connect it to
$\bigcup_{i\in[k-1]}V_{t,i}$. We set $S=\{s_0\}\cup N(s_0)$ and $T=\{t_0\}\cup
N(t_0)$. We also set $\ell=(k-1)(t-1)+1$.  

It is not hard to see that the size of the graph and the value of $\ell$
satisfy the promised properties and that, since $k-1\ge 3$ the graph (and its
complement) can be constructed in time at most $2^{2p/3}|\phi|^{O(1)}$, because
the number of possible edges is at most $2^{\frac{2p}{k-1}}|\phi|^{O(1)}$.
What remains is to show correctness.

For the forward direction, suppose that $\phi$ is satisfiable. We identify a
unique vertex of each $V_{j,i}$ which agrees with our satisfying assignment and
show that we can reconfigure $S$ to $T$ using these vertices. Notice that for
each $j$, this includes $k-1$ vertices, one from each $V_{j,i}$, $i\in[k-1]$,
which form a clique (since the assignment is satisfying).  Now, for each
$j\in[t-1]$, suppose that in the current state we have tokens in one vertex of
each $V_{j,i}$, for $i\in[k-1]$ and no tokens in $V_{j+1,i}$ for any
$i\in[k-1]$. We take the token that is not in any $V_{j,i}$ (initially this is
the token in $s_0$) and place it in the selected vertex of $V_{j+1,1}$; then
for each $i\in[k-2]$ we take the token from $V_{j,i}$ and place it on the
selected vertex of $V_{j+1,i+1}$. After these $k-1$ moves we move on to the
next value of $j$. The last move is to take the token that is not in any
$V_{t,i}$ and place it in $t_0$. We have performed exactly $\ell=(t-1)(k-1)+1$
moves. At any point we keep tokens, for some $j\in[t]$ and $i\in[k-1]$ in a
single vertex from each of the sets $V_{j,i}, V_{j,i+1},\ldots, V_{j,k-1},
V_{j+1,1},\ldots, V_{j+1,i}$, which form a clique if our assignment was
satisfying and consistent.

For the converse direction, Suppose there is a sequence of cliques
$S_0=S,S_1,\ldots, S_{\ell}=T$ which represent a token jumping reconfiguration.
Take a value $r\in[\ell]\setminus\{0,\ell\}$ and let us say that the set $S_r$
is described by the pair $(j,i)$, for $j\in[t], i\in[k-1]$, if $(j,i)$ is the
lexicographically smallest pair for which $V_{j,i}\cap S_r\neq\emptyset$. We
now observe that if $S_r$ is described by $(j,i)$ then it must contain exactly
one vertex from each of the sets $V_{j,i}, V_{j,i+1},\ldots, V_{j,k-1},
V_{j+1,1},\ldots, V_{j+1,i}$. This is because all sets $V_{j',i'}$ are
independent (so at most one vertex from each set is in $S_r$); for all
$(j',i')$ which come lexicographically after $(j+1,i)$, there is no edge from
$V_{j,i}$ to $V_{j',i'}$, so $S_r\cap V_{j',i'}=\emptyset$; $S_r$ has size
exactly $k$. Thanks to these observations we can now see that if $S_r$ is
described by $(j,i)$, then $S_{r+1}$ must be described by a pair $(j',i')$
which is lexicographically either immediately after $(j,i)$ (that is, $j=j'$
and $i'=i+1$, or $i=k-1, i'=1$, and $j'=j+1$), or immediately before, because
sets which are not consecutive in the lexicographic ordering have symmetric
difference of size at least $2$. Because there are $(t-1)(k-1)$ pairs $(j,i)$,
$S_1$ is described by $(1,1)$ and $S_{\ell-1}$ by $(t-1,k-1)$, taking into
account the two moves required to remove a token from $s_0$ and to place a
token into $t_0$, we conclude that the only way to obtain a feasible sequence
of sets is to have exactly one set described by each pair $(j,i)$. 

It is not hard to see that the above imply that our sequence uses exactly one
vertex from each $V_{j,i}$, from which we can infer an assignment to $\phi$.
Specifically, for each $V_{j,i}$, there is exactly one vertex that is used at
some point and we deduce from this an assignment to $B_j\cap V_i$. This
assignment must be consistent with the one deduced from $V_{j+1,i}$, due to the
edges of step 4; hence this assignment is consistent throughout. The assignment
must be satisfying, since the removal of edges from step 2 forbids us from
using vertices in the clique that would falsify clauses involving variables
from two groups.  \end{proof}



\newpage

\bibliography{ppseth}

\end{document}